\documentclass[11pt]{article}
\usepackage[utf8]{inputenc}
\usepackage{amsfonts,latexsym,amsmath,amsthm,amssymb,amscd,euscript,mathrsfs,graphicx}
\usepackage{framed}
\usepackage{fullpage}
\usepackage[dvipsnames]{xcolor}
\usepackage[shortlabels]{enumitem}
\usepackage[obeyFinal,textsize=scriptsize,shadow]{todonotes}
\usepackage{tikz}
\usetikzlibrary{matrix}
\usepackage[hidelinks]{hyperref}
\usepackage{nicefrac}
\newcommand{\nf}{\nicefrac}
\usepackage{marginnote}
\usepackage{algorithmicx}
\usepackage{algorithm}
\usepackage[noend]{algpseudocode}

\usepackage{multirow} 
\usepackage{subfigure}

\usepackage{accents}
\renewcommand\bar[1]{\accentset{\rule{.4em}{.7pt}}{#1}}

\newtheorem{theorem}{Theorem}[section]
\newtheorem{proposition}[theorem]{Proposition}
\newtheorem{lemma}[theorem]{Lemma}
\theoremstyle{definition}
\newtheorem{defn}[theorem]{Definition}

\theoremstyle{remark}
\newtheorem*{remark}{Remark}

\newcommand{\kmeansratio}{5.912}
\newcommand{\kmedianratio}{2.406}

\newcommand{\BE}{\mathbb E}

\newcommand{\BP}{\mathbb P}

\newcommand{\BR}{\mathbb R}

\newcommand{\eps}{\varepsilon}

\newcommand{\clients}{\mathcal{D}}
\newcommand{\facilities}{\mathcal{F}}
\newcommand{\opt}{\text{OPT}}

\title{Improved Approximations for Euclidean $k$-means and $k$-median, via Nested Quasi-Independent Sets}
\author{Vincent Cohen-Addad\thanks{Google Research. cohenaddad@google.com}~, Hossein Esfandiari\thanks{Google Research. esfandiari@google.com}~, Vahab Mirrokni\thanks{Google Research. mirrokni@google.com}~, Shyam Narayanan\thanks{Massachusetts Institute of Technology. Work done as an intern at Google Research. shyamsn@mit.edu}}
\date{}

\begin{document}

\maketitle
\begin{abstract}
    Motivated by data analysis and machine learning applications, we consider the popular high-dimensional Euclidean $k$-median and $k$-means problems. We propose a new primal-dual algorithm, inspired by the classic algorithm of Jain and Vazirani~\cite{jain2001lagrangian} and the recent algorithm of Ahmadian, Norouzi-Fard, Svensson, and Ward~\cite{ahmadian2017better}. Our algorithm achieves an approximation ratio of $\kmedianratio$ and $\kmeansratio$ for Euclidean $k$-median and $k$-means, respectively, improving upon the 2.633 approximation ratio of Ahmadian et al.~\cite{ahmadian2017better} and the 6.1291 approximation ratio of Grandoni, Ostrovsky, Rabani, Schulman, and Venkat~\cite{Grandoni21}.
    
    Our techniques involve a much stronger exploitation of the Euclidean metric than previous work on Euclidean clustering. In addition, we introduce a new method of removing excess centers using a variant of independent sets over graphs that we dub a ``nested quasi-independent set''. In turn, this technique may be of interest for other optimization problems in Euclidean and $\ell_p$ metric spaces.
\end{abstract}
\thispagestyle{empty}
\newpage

\thispagestyle{empty}
\tableofcontents
\newpage
\clearpage
\setcounter{page}{1}

 \newpage
\section{Introduction}
The $k$-means and $k$-median problems are among the oldest and most fundamental
clustering problems. Originally motivated by operations research and statistics problems when they
first appeared in the late 50s~\cite{zbMATH03129892,lloyd1957least,hakimi1964optimum} they
are now at the heart of several unsupervised and semi-supervised machine learning models and data mining techniques and are thus
a main part of the toolbox of modern data analysis techniques in a variety of fields.
In addition to their practical relevance, these two problems exhibit strong ties with some classic
optimization problems, such as set cover, and understanding their complexity has thus been a
long standing problem which has inspired several breakthrough techniques.

Given two sets $\clients$, $\facilities$ of points in a metric space and an
integer $k$, the goal of the $k$-median problem is to find a set $S$ of $k$ points
in $\facilities$, called \emph{centers}, minimizing
the sum of distances from each point in $\clients$ to the closest point in $S$. The goal
of the $k$-means problem is to minimize the sum of distances squared.
The complexity of the $k$-median or $k$-means problems heavily depends on the underlying metric space.
In general metric spaces, namely when the distances only need to obey the triangle inequality, the $k$-median and
$k$-means problem
are known to admit a 2.675-approximation~\cite{ByrkaPRST17} and a 9-approximation~\cite{ahmadian2017better}, respectively, and cannot be approximated
better than $1+2/e \sim 1.736$ for $k$-median and $1+8/e \sim 3.943$ for $k$-means assuming P $\neq$ NP.
We know that the upper and lower bounds are tight when we allow a running time of
$f(k) n^{O(1)}$ (i.e.: in the fixed-parameter tractability setting) for arbitrary computable functions $f$~\cite{cohenaddad2019fpt},
suggesting that the lower bounds 
cannot be improved for 
the general case either.
When we turn to slightly more structured metric spaces such as Euclidean metrics
the picture changes drastically. While the problem remains NP-hard when the dimension $d = 2$
(and $k$ is large)~\cite{megiddo1984complexity} or $k=2$ (and $d$ is large)~\cite{dasgupta2008hardness}, both problems admit
$(1+\eps)$-approximation algorithms with running times $f(k,\eps) n d$~\cite{KumarSS10}, with an exponential dependency in $k$, and
$f(d, \eps) n \log^{O(1)} n$~\cite{cohen2019near,Cohen-Addad18}, with
a doubly exponential dependency in $d$ (the latter extends to doubling metrics), a prohibitive running time
in practice.

Arguably, the most practically important setting is when the input points lie in Euclidean space
of large dimension and the number of clusters is non-constant, namely when both $k$ and $d$ are
part of the input~\footnote{Note here that $d$ can always be assumed to
be of $O(\log k/\eps^{-2})$ using dimensionality reduction techniques~\cite{MakarychevMR19} (see also~\cite{BecchettiBC0S19} for a slightly worse bound)}. Unfortunately,
the complexity of the problem in this regime is far from being understood. Recent results
have proven new inapproximability results:
respectively 1.17 and 1.07 for the Euclidean k-means and k-median problems assuming P $\neq$ NP
and 1.73 and 1.27 assuming the Johnson Coverage
hypothesis of~\cite{Cohen-AddadS19,Cohen-AddadLS22}. For the continuous
case, the same series of papers show a hardness of 1.06 and 1.015 for Euclidean $k$-means and $k$-median respectively assuming P $\neq $ NP
and 1.36 and 1.08 assuming the Johnson-Coverage hypothesis
(see also~\cite{CKL21} for further related work on continuous $k$-median and $k$-means in other metric spaces). 
Interestingly, the above hardness results 
implies that there is no algorithmic benefit that could be gained from the $\ell_1$-metric: Assuming the Johnson Coverage hypothesis the hardness bounds for $k$-median and $k$-means are the same in the $\ell_1$-metric that in general metrics~\cite{Cohen-AddadLS22}. However,
it seems plausible
to leverage the structure of the $\ell_2$-metric to obtain approximation algorithms bypassing the
lower bounds for the general metric or $\ell_1$ case (e.g.: obtaining an approximation ratio better than $1+2/e$ for $k$-median).

In a breakthrough result, Ahmadian et al.~\cite{ahmadian2017better} were the first to exploit the structure
of high-dimensional Euclidean metrics to obtain better bounds than the current best-known bounds
for the general metric case. Concretely, they showed how to obtain a 6.3574-approximation
for $k$-means (improving upon the 9-approximation of Kanungo et al.~\cite{KanungoMNPSW04}) and a $1+\sqrt{8/3}+\eps \approx 2.633$-approximation
for $k$-median (improving upon the 2.675-approximation for general metrics~\cite{ByrkaPRST17}).
The bound for $k$-means was recently improved to a 6.1291-approximation (or more precisely, the unique real root to $4x^3-24x^2-3x-1 = 0$) by \cite{Grandoni21}, by tweaking the analysis of Ahmadian et al.~\cite{ahmadian2017better}, and no progress has been made for Euclidean $k$-median after \cite{ahmadian2017better}.

This very active line of research mixes both new hardness of approximation results and approximation algorithms
and aims at answering a very fundamental question: How much can we leverage the Euclidean geometry
to obtain better approximation algorithms?
And conversely, what do we learn about Euclidean geometry when studying basic computational problems?
Our results aim at making progress toward answering the above questions.

\subsection{Our Results}
Our main result consists of better approximation algorithms for both $k$-median and $k$-means,
with ratio $\kmedianratio$ for $k$-median and $\kmeansratio$ for $k$-means, improving
the 2.633-approximation of Ahmadian et al.~\cite{ahmadian2017better} for $k$-median and $6.1291$-approximation of \cite{Grandoni21} for $k$-means.
\begin{theorem}
\label{thm:main:approxratio}
For any $\eps > 0$, there exists a polynomial-time algorithm that returns a solution to the Euclidean $k$-median problem whose cost is at most $\kmedianratio+\eps$ times the optimum.

For any $\eps > 0$, there exists a polynomial-time algorithm that returns a solution to the Euclidean $k$-means problem whose cost is at most $\kmeansratio+\eps$ times the optimum.
\end{theorem}
Our approximation ratio for $k$-median breaks the natural barrier of $1+\sqrt{2} > 2.41$ and our approximation ratio for $k$-means is the first
below 6.
The approximation bound of $1+\sqrt{2}$ for Euclidean $k$-median is indeed a natural barrier for the state-of-the-art approach of Ahmadian et al.~\cite{ahmadian2017better} that relies on using the primal-dual approach of 
Jain and Vazirani~\cite{jain2001lagrangian}.
At a high level, the approximation bound of
3 for general metrics for the algorithm of Jain and Vazirani can be interpreted as an approximation bound of 1+2, where 1 is the optimum service cost and an additional cost of 2 times the optimum 
is added for input clients poorly served by the solution. Since general metrics are only required to 
satisfy the triangle inequality, the 2 naturally arises from bounding the distance from a client to its center in the solution output and an application of the triangle inequality 
to bound this distance. Therefore, one can hope to then obtain a substantial gain when working in Euclidean spaces: The triangle inequality is rarely tight (unless points are aligned) and this leads to the hope of replacing the 2 by $\sqrt{1+1}$ in the above the bound, making the approximation ratio of $1+\sqrt 2$ a natural target for $k$-median. 
In fact, this high-level discussion can be made slightly more formal:
we show that the analysis of the result of Ahmadian et al.~\cite{ahmadian2017better} cannot be
improved below $1+\sqrt{2}$ for $k$-median,
exhibiting a limit of the state-of-the-art approaches.

In this paper, we take one step further, similar to the result of Li and Svensson~\cite{LiS16} for the general metric case who improved for the first below
the approximation ratio of 3, we show how to bypass $1+\sqrt 2$ for 
Euclidean $k$-median (and the bound of 6.129 for $k$-means).

Furthermore, one of our main contributions is to obtain better Lagrangian Multiplier Preserving (LMP) approximations
for the Lagrangian relaxations of both problems. 
To understand this, we need to give a little more
background on previous work and how previous approximation algorithms were derived.
A natural approach to the $k$-median and $k$-means problem is to (1) relax the constraint on the number of
centers $k$ in the solution, (2) find an approximate solution for the relaxed problem, and (3) derive an approximate solution
that satisfies the constraint on the number of centers.
Roughly, an LMP approximate solution $S$ is a solution where we bound the ratio of the cost of $S$ to the optimum cost, but pay a penalty proportional to some $\lambda \ge 0$ for each center in $S$. Importantly, if $|S| = k$, an LMP $\rho$-approximation that outputs $S$ is a $\rho$-approximation for $k$-means (or $k$-median). We formally define an LMP approximation in Section \ref{sec:prelim}.
LMP solutions have played a central role in obtaining better approximation algorithms for $k$-median
in general metrics and more recently in high-dimensional Euclidean metrics. Thus obtaining better LMP solutions for the Lagrangian relaxation of $k$-median and $k$-means has been an important problem. 
A byproduct of our approach is a new 2.395-LMP for Euclidean $k$-median and $3+2\sqrt{2}$-LMP for Euclidean $k$-means.

Our techniques may be of use in other clustering and combinatorial optimization problems over Euclidean space as well, such as Facility Location. In addition, by exploiting the geometric structure similarly, these techniques likely extend to $\ell_p$-metric spaces (for $p > 1$).

  


\subsection{Related Work}
\label{sec:related-work}
The first $O(1)$-approximation for the $k$-median problem in general metrics
is due to Charikar et al.~\cite{CharikarGTS02}. The $k$-median
problem has then been a testbed for a variety of powerful approaches 
such as the primal-dual
schema~\cite{jain2001lagrangian,CharikarG99}, 
greedy algorithms (and dual fitting)~\cite{JainMMSV03}, improved LP
rounding~\cite{CharikarL12}, local-search~\cite{AryaGKMMP04,CohenAddadGHOS22}, and 
LMP-based-approximation~\cite{LiS16}. The current best approximation
guarantee is 2.675~\cite{ByrkaPRST17} and the best
hardness result is
$(1+2/e)$~\cite{GuK99}.
For $k$-means in general metrics, the current best approximation guarantee is 9~\cite{ahmadian2017better} and the current best hardness result is $(1+8/e)$ (which implicitly follows from~\cite{GuK99}, as noted in \cite{ahmadian2017better}).

We have already covered in the introduction the history of 
algorithms for
high-dimensional Euclidean $k$-median and $k$-means 
with running time polynomial in both $k$ and the dimension and that leverage the properties of the Euclidean metrics
(\cite{KanungoMNPSW04,ahmadian2017better,Grandoni21}). In terms of lower
bounds, the first to show that the high-dimensional 
$k$-median and $k$-means problems were APX-hard were Guruswami and Indyk~\cite{GI03},
and later Awasthi et al.~\cite{AwasthiCKS15} showed that the APX-hardness
holds even if the centers can be placed arbitrarily in $\mathbb{R}^d$. The inapproximability 
bound was later slightly improved by Lee et al.~\cite{DBLP:journals/ipl/LeeSW17} until the recent best known bounds of~\cite{Cohen-AddadS19,Cohen-AddadLS22}.
From a more practical point of view, Arthur and Vassilvitskii showed that the widely-used
popular heuristic of Lloyd~\cite{lloyd1957least} can lead to solutions with arbitrarily
bad approximation guarantees~\cite{ArV09}, but can be improved by a simple seeding strategy, called $k$-means++, so as to guarantee that the output is within an $O(\log k)$
factor of the optimum~\cite{ArV07}.

For fixed $k$, there are several known approximation schemes, typically using small coresets \cite{BecchettiBC0S19,FL11,KumarSS10}
There also exists a large body of bicriteria approximations (namely outputting a solution with $(1+c)k$ centers for 
some constant $c > 0$): see, e.g.,~\cite{BaV15,CharikarG05,CoM15,KPR00,MakarychevMSW16}.
There has also been a long line or work on the metric facility location problem,
culminating with the result of Li~\cite{Li13} who gave a 1.488-approximation
algorithm, almost matching the lower bound of 1.463 of Guha and Khuller~\cite{GuK99}.
Note that no better bound is known for high-dimensional Euclidean facility
location.

\subsection{Roadmap}
In Section \ref{sec:prelim}, we describe some preliminary definitions. We also formally define the LMP approximation and introduce the LMP framework of Jain and Vazirani~\cite{jain2001lagrangian} and Ahmadian et al.~\cite{ahmadian2017better}. In Section \ref{sec:overview}, we provide an overview of the new technical results we developed to obtain the improved bounds. In Section \ref{sec:lmp_k_means}, we obtain a $3+2\sqrt{2} \approx 5.828$ LMP approximation for the Euclidean $k$-means problem. In Section \ref{sec:poly_time_alg}, extend our LMP approximation to a $\kmeansratio$-approximation for standard Euclidean $k$-means. Finally, in Section \ref{sec:k_median}, we obtain a $2.395$ LMP approximation for Euclidean $k$-median that can be extended to a $\kmedianratio$-approximation for standard Euclidean $k$-median.

In Appendix \ref{app:limit}, we briefly show that the result of \cite{ahmadian2017better} cannot be extended beyond $1+\sqrt{2}$ for Euclidean $k$-median: this demonstrates the need of our new techniques for breaking this barrier.

\section{Preliminaries} \label{sec:prelim}

Our goal is to provide approximation algorithms for either the $k$-means or $k$-median problem in Euclidean space on a set $\mathcal{D}$ of \textit{clients} of size $n$.
For the entirety of this paper, we consider the discrete $k$-means and $k$-median problems, where rather than having the $k$ centers allowed to be anywhere, we are given a fixed set of facilities $\mathcal{F}$ of size $m$ which is polynomial in $n$, from which the $k$ centers must be chosen from. It is well-known (e.g., \cite{Mat00}) that a polynomial-time algorithm providing a $\rho$-approximation for discrete $k$-means (resp., median) implies a polynomial-time $\rho+\eps$-approximation for standard $k$-means (resp., median) for an arbitrarily small constant $\eps$. 

For two points $x, y$ in Euclidean space, we define $d(x, y)$ as the Euclidean distance  (a.k.a. $\ell_2$-distance)  between $x$ and $y$. In addition, to avoid redefining everything or restating identical results for both $k$-means and $k$-median, we define $c(x, y) := d(x, y)^2$ in the context of $k$-means and $c(x, y) := d(x, y)$ in the context of $k$-means. For a subset $S$ of Euclidean space, we define $d(x, S) := \min_{s \in S} d(x, s)$ and $c(x, S) := \min_{s \in S} c(x, s)$.

For the $k$-means (or $k$-median) problem, for a subset $S \subset \mathcal{F}$, we define $\text{cost}(\mathcal{D}, S) := \sum_{j \in \mathcal{D}} c(j, S)$. In addition, we define $\text{OPT}_k$ to be the optimum $k$-means (or $k$-median) cost for a set $\mathcal{D}$ and a set of facilities $\mathcal{F}$, i.e., $\text{OPT}_k = \min_{S \subset \mathcal{F}, |S| = k} \text{cost}(\mathcal{D}, S)$. Recall that a $\rho$-approximation algorithm is an algorithm that produces a subset of $k$ facilities $S \subset \mathcal{F}$ with $\text{cost}(\mathcal{D}, S) \le \rho \cdot \text{OPT}_k$ in the worst-case. 

\subsection{The Lagrangian LP Relaxation and LMP Solutions} \label{subsec:lagrangian_lp}

We first look at the standard LP formulation for $k$-means/medians. The variables of the LP include a variable $y_i$ for each facility $i \in \mathcal{F}$ and a variable $x_{i, j}$ for each pair $(i, j)$ for $i \in \mathcal{F}$ and $j \in \mathcal{D}$. The standard LP relaxation is the following:

\begin{alignat}{5}
    \text{minimize } & \sum_{i \in \mathcal{F}, j \in \mathcal{D}} x_{i, j} \cdot c(j, i) \hspace{0.5cm} && && && \label{eq:lp_A}\\
    \text{such that } & && \sum_{i \in \mathcal{F}} x_{i, j} && \ge 1 &&\qquad \forall j \in \mathcal{D} \label{eq:lp_B} \\
    & && \quad \sum_{i \in \mathcal{F}} y_i && \le k && \label{eq:lp_C}\\
    & && 0 \le x_{i, j} && \le y_i &&\qquad \forall j \in \mathcal{D}, i \in \mathcal{F} \label{eq:lp_D}
\end{alignat}
The intuition behind this linear program is that we can think of $x_{i, j}$ as the indicator variable of client $j$ being assigned to facility $i$, and $y_i$ as the indicator variable of facility $i$ being opened. We need every facility $j \in \mathcal{D}$ to be assigned to at least one client, that at most $k$ facilities $i$ are opened, and that $x_{i, j}$ is $1$ only if $y_i = 1$ (since clients can only be assigned to open facilities). We also ensure a nonnegativity constraint on $x_{i, j}$ and $y_i$ by ensuring that $0 \le x_{i, j}$. Finally, our goal is to minimize the sum of distances (for $k$-median) or the sum of squared distances (for $k$-means) from each client to its closest facility -- or simply the facility it is assigned to, and if exactly one of the $x_{i, j}$ values is $1$ for a fixed client $j$ and the rest are $0$, then $\sum_{i \in \mathcal{F}} x_{i, j} c(j, i)$ is precisely the distance (or squared distance) from $j$ to its corresponding facility. By relaxing the linear program to have real variables, we can only decrease the optimum, so if we let $L$ be the optimum value of the LP relaxation, then $L \le \text{OPT}_k$.

Jain and Vazirani \cite{jain2001lagrangian} considered the \emph{Lagrangian relaxation} of this linear program, by relaxing the constraint \eqref{eq:lp_C} and adding a dependence on a Lagrangian parameter $\lambda \ge 0$. By doing this, the number of facilities no longer has to be at most $k$ in the relaxed linear program but the objective function penalizes for opening more than $k$ centers. Namely, the goal becomes to minimize
\begin{equation} \label{eq:lp_A'}
    \sum_{i \in \mathcal{F}, j \in \mathcal{D}} x_{i, j} \cdot c(j, i) + \lambda \cdot \left(\sum_{i \in \mathcal{F}} y_i - k\right)
\end{equation}
subject to Constraints \eqref{eq:lp_B} and \eqref{eq:lp_D}. Indeed, for $\lambda \ge 0$, the objective only decreases from \eqref{eq:lp_A} to \eqref{eq:lp_A'} for any feasible solution to the original LP. Therefore, this new linear program, which we will call $\text{LP}(\lambda)$, has optimum $L(\lambda) \le L$. Now, it is known that the Dual linear program to this Lagrangian relaxation of the original linear program can be written as the following, which has variables $\alpha = \{\alpha_j\}_{j \in \mathcal{D}}$:
\begin{alignat}{5}
    \text{maximize } & \left(\sum_{j \in \mathcal{D}} \alpha_j\right) - \lambda \cdot k \hspace{0.5cm} && && && \label{eq:lp_E} \\
    \text{such that} & && \sum_{j \in \mathcal{D}} \max(\alpha_j-c(j, i), 0) && \le \lambda && \qquad \forall i \in \mathcal{F} \label{eq:lp_F}\\
    & && \hspace{3.6cm} \alpha &&\ge 0 && \label{eq:lp_G}
\end{alignat}
We call this linear program $\text{DUAL}(\lambda)$. Because the optimum to $\text{DUAL}(\lambda)$ equals the optimum to the primal $\text{LP}(\lambda)$ by strong duality, this means that for any $\alpha = \{\alpha_j\}_{j \in \mathcal{D}}$ satisfying Conditions \eqref{eq:lp_F} and \eqref{eq:lp_G}, we have that $\left(\sum_{j \in \mathcal{D}} \alpha_j\right) - \lambda \cdot k \le L(\lambda) \le L \le \text{OPT}_{k}$.

For a fixed $\lambda$, we say that $\alpha$ is \emph{feasible} if it satisfies both \eqref{eq:lp_F} and \eqref{eq:lp_G}. Thus, to provide a $\rho$-approximation to $k$-means (or $k$-median), it suffices to provide both a feasible $\alpha$ and a subset $S \subset \mathcal{F}$ of size $k$ such that $\text{cost}(\mathcal{D}, S)$ is at most $\rho \cdot \left(\sum_{j \in \mathcal{D}} \alpha_j - \lambda \cdot |S|\right).$

\medskip

In both the work of Jain and Vazirani~\cite{jain2001lagrangian} and the work of Ahmadian et al.~\cite{ahmadian2017better}, they start with a weaker type of algorithm, called a \emph{Lagrangian Multiplier Preserving} (LMP) approximation algorithm. To explain this notion, let $\text{OPT}(\lambda)$ represent the optimum (minimum) for the modified linear program $\text{LP}'(\lambda)$, which is the same as $\text{LP}(\lambda)$ except without the subtraction of $\lambda \cdot k$ in the objective function \eqref{eq:lp_A'}. (So, $\text{LP}'(\lambda)$ has no dependence on $k$). Note that this is also the optimum (maximum) for $\text{DUAL}'(\lambda)$, which is the same as $\text{DUAL}(\lambda)$ except without the subtraction of $\lambda \cdot k$ in the objective function \eqref{eq:lp_E}. Then, for some fixed $\lambda \ge 0$, we say that a $\rho$-approximation algorithm is \emph{LMP} if it returns a solution $S \subset \mathcal{F}$ satisfying
\[\sum_{j \in \mathcal{D}} c(j, S) \le \rho \cdot \left(\text{OPT}(\lambda) - \lambda \cdot |S|\right).\]
Indeed, if we could find a choice of $\lambda$ and an LMP $\rho$-approximate solution $S$ with size $|S| = k$, we would have found a $\rho$-approximation for $k$-means (or $k$-median) clustering.

\subsection{Witnesses and Conflict Graphs} \label{subsec:witness_and_conflict}

Jain and Vazirani \cite{jain2001lagrangian} proposed a simple primal-dual approach to create a feasible solution $\alpha$ of $\text{DUAL}(\lambda)$ with certain additional properties that are useful for providing an efficient solution to the original $k$-median (or $k$-means) problem efficiently. We describe it briefly as follows, based on the exposition of Ahmadian et al.~\cite[Subsection 3.1]{ahmadian2017better}.

Start with $\alpha = \textbf{0}$, i.e., $\alpha_j = 0$ for all $j \in \mathcal{D}$. We increase all $\alpha_j$'s continuously at a uniform rate, but stop growing each $\alpha_j$ once one of the following two events occurs:

\begin{enumerate}
    \item For some $i \in \mathcal{F}$, a dual constraint $\sum_{j \in \mathcal{D}} \max(\alpha_j-c(j, i), 0) \le \lambda$ becomes tight (i.e., reaches equality). Once this happens, we stop growing $\alpha_j$, and declare that facility $i$ is \emph{tight} for all $i$ such that the constraint became equality. In addition, we will say that $i$ is the \emph{witness} of $j$.
    \item For some already tight facility $i$, we grow $\alpha_j$ until $\alpha_j = c(j, i)$. In this case, we also say that $i$ is the \emph{witness} of $j$.
\end{enumerate}
    We note that this process must eventually terminate for all $j$ (e.g., once $\alpha_j$ reaches $\min_{i \in \mathcal{F}} c(j, i)+\lambda$). This completes our creation of the dual solution $\alpha$ (it is simple to see that $\alpha$ is feasible). 
    
    For any client $j$, we define $N(j) := \{i \in \mathcal{F}: \alpha_j > c(j, i)\}$, and likewise, for any client $i$, we define $N(i) := \{j \in \mathcal{D}: \alpha_j > c(j, i)\}$. For any tight facility $i$, we will define $t_i := \max_{j \in N(i)} \alpha_j$, where $t_i = 0$ by default if $N(i) = \emptyset$. For each client $j$, its witness $i$ will have the useful properties that $t_{i} \le \alpha_j$ and $c(j, i) \le \alpha_j$.
    
    We have already created our dual solution $\alpha$: to create our set of $k$ facilities, we will choose a subset of the tight facilities. First, we define the \textit{conflict graph} on the set of tight facilities. Indeed, Jain and Vazirani \cite{jain2001lagrangian}, Ahmadian et al. \cite{ahmadian2017better}, and we all have slightly different definitions: so we contrast the three.
\begin{itemize}
    \item \cite{jain2001lagrangian} Here, we say that $(i, i')$ forms an edge in the conflict graph $H$ if there exists a client $j$ such that $i, i' \in N(j)$ (or equivalently, $\alpha_j \ge c(j, i)$ and $\alpha_j \ge c(j, i')$).
    \item \cite{ahmadian2017better} Here, we say that $(i, i')$ forms an edge in the conflict graph $H(\delta)$ (where $\delta > 0$ is some parameter) if $c(i, i') \le \delta \cdot \min(t_i, t_{i'})$ and there exists a client $j$ such that $i, i' \in N(j)$.
    \item In our definition, we completely drop the condition from \cite{jain2001lagrangian}, and just say that $(i, i')$ forms an edge in the conflict graph $H(\delta)$ if $c(i, i') \le \delta \cdot \min(t_i, t_{i'})$.
\end{itemize}
    It turns out that in the algorithm of Ahmadian et al.~\cite{ahmadian2017better}, the approximation factor is not affected by whether they use their definition or our definition. But in our case, it turns out that dropping the condition from \cite{jain2001lagrangian} in fact allows us to obtain a better approximation.
    
    To provide an LMP approximation, both Jain and Vazirani~\cite{jain2001lagrangian} and Ahmadian et al.~\cite{ahmadian2017better} constructed a maximal independent set $I$ of the conflict graph $H$ (or $H(\delta)$ for an appropriate choice of $\delta > 0$) and used $I = S$ as the set of centers. For Jain and Vazirani's definition, the independent set $I$ obtains an LMP $9$-approximation for metric $k$-means. For Ahmadian et al.'s definition, the independent set $I$ obtains an LMP $6.1291$-approximation for Euclidean $k$-means if $\delta$ is chosen properly. (We note that Ahmadian et al. only proved a factor of $6.3574$, though
    their argument can be improved to show a $6.1291$-approximation factor as proven by Grandoni et al.~\cite{Grandoni21}). While we will not explicitly prove it, our definition of $H(\delta)$ also obtains the same bound with the same choice of $\delta$. For the Euclidean $k$-median problem, using either Ahmadian et al.'s or our definition, one can obtain an LMP $(1+\sqrt{2})$-approximation: Ahmadian et al.~\cite{ahmadian2017better} only proved it for the weaker $1+\sqrt{8/3}$ approximation factor, but we prove the improved approximation in Subsection \ref{subsec:lmp_k_median_easy}.
    We then show how to obtain a better LMP solution
    and then a better approximation bound.
    


\section{Technical Overview} \label{sec:overview}


The algorithms of both Jain and Vazirani \cite{jain2001lagrangian} and Ahmadian et al \cite{ahmadian2017better} begin by constructing an LMP approximation. Their approximation follows two phases: a \emph{growing} phase and a \emph{pruning} phase. In the growing phase, as described in Subsection \ref{subsec:witness_and_conflict}, they grow the solution $\alpha$ starting from $\alpha = \textbf{0}$, until they obtain a suitable dual solution $\alpha$ for $\text{DUAL}(\lambda)$. In addition, they have a list of \emph{tight} facilities $i$, which we think of as our candidate centers. The pruning phase removes unnecessary facilities: as described in Subsection \ref{subsec:witness_and_conflict}, we create a conflict graph $H(\delta)$ over the tight facilities, and only choose a maximal independent set $I$. Hence, we are pruning out tight facilities to make sure we do not have too many nearby facilities. This way we ensure that the total number of centers is not unnecessarily large. Our main contributions are to improve the pruning phase with a new algorithm and several new geometric insights, and to show how our LMP approximation can be extended to improved algorithms for standard Euclidean $k$-means (and $k$-median) clustering.

\paragraph{Improved LMP Approximation:}
To simplify the exposition, we focus on Euclidean $k$-means. To analyze the approximation, Ahmadian et al.~\cite{ahmadian2017better} compare the cost of each client $j$ in the final solution to its contribution to the dual objective function. The cost of a client $j$ is simply $c(j, I)$ where $I$ is our set of centers, and $j$'s contribution to the dual is $\alpha_j - \sum_{i \in N(j) \cap I} (\alpha_j-c(j, i))$, where we recall the definition of $N(j)$ from Subsection \ref{subsec:witness_and_conflict}. One can show that the sum of the individual dual contributions equals the dual objective \eqref{eq:lp_E}. 
By casework on the size $a = |N(j) \cap I|$, \cite{Grandoni21} (by modifying the work of \cite{ahmadian2017better}) shows that $c(j, I) \le \rho \cdot \left[\alpha_j - \sum_{i \in N(j) \cap I} (\alpha_j-c(j, i))\right],$ where $\rho \approx 6.1291$ if $\delta \approx 2.1777$ is chosen appropriately (in general, we think of $\delta$ as slightly greater than $2$). The only bottlenecks (i.e., where the cost-dual ratio could equal $\rho$) are when $a \in \{0, 2\}$. Our LMP approximation improves the pruning phase by reducing the cost-dual ratio in these cases.

Due to the disjointed nature of the bottleneck cases, our first observation was that averaging between two or more independent sets may be beneficial. This way, if for some client $j$, the first independent set $I_1$ had $a = 0$ and the second set $I_2$ had $a = 2$, perhaps by taking a weighted average of $I_1$ and $I_2$ we can obtain a set $I$ where $a = 1$ with reasonable probability. Hence, the expected cost-dual ratio of $j$ will be below $\rho$. The second useful observation comes from the fact that $t_{i^*} \le \alpha_j$ and $c(j, i^*) \le \alpha_j$, if $i^*$ is the witness of $j$. In the $a=0$ case, \cite{ahmadian2017better} applies this to show that $d(j, i^*) \le \sqrt{\alpha_j}$ and $d(i^*, I) \le \sqrt{\delta \cdot t_{i^*}} \le \sqrt{\delta \cdot \alpha_j}$, which follows by the definition of conflict graph. Hence, the bottleneck occurs when $i^*$ has distance exactly $\sqrt{\alpha_j}$ from $j$, and the nearest point $i \in I$ has distance exactly $\sqrt{\delta \cdot \alpha_j}$ from $i^*$, in the direction opposite $j$. This causes $c(j, I)$ to be $(1+\sqrt{\delta})^2 \cdot \alpha_j$ in the worst case, whereas $j$'s contribution to the dual is merely $\alpha_j$. To reduce this ratio, we either need to make sure that the distance from $i^*$ to either $j$ or $i$ is smaller, or that $j, i^*, i$ do not lie in a line in that order.

A reasonable first attempt is to select two choices $\delta_1 \ge \delta \ge \delta_2$, and consider the nested conflict graphs $H(\delta_1) \supset H(\delta_2)$. We can then create nested independent sets by first creating $I_1$ as a maximal independent set of $H(\delta_1)$, then extending it to $(I_1 \cup I_2) \supset I_1$ for $H(\delta_2)$. Our final set $S$ will be an average of $I_1$ and $(I_1 \cup I_2)$: we include all of $I_1$ and each point in $I_2$ with some probability. The motivation behind this attempt is that if both $N(j) \cap I_1$ and $N(j) \cap I_2$ are empty, the witness $i^*$ of $j$ should be adjacent in $H(\delta_1)$ to some point $i_1 \in I_1$, and adjacent in $H(\delta_2)$ to some point $i_2 \in (I_1 \cup I_2)$. Hence, either $i_1 = i_2$, in which case the distance from $j$ to $I_1$ is now only $(1+\sqrt{\delta_2}) \cdot \sqrt{\alpha_j}$ instead of $(1+\sqrt{\delta_1}) \cdot \sqrt{\alpha_j}$ (see Figure \ref{fig:1a}), or $i_1 \neq i_2$, in which case $i_1, i_2$ must be far apart because $i_1, i_2$ are both in the independent set $I_1 \cup I_2$ (see Figure \ref{fig:1b}). In the latter case, we cannot have the bottleneck case for both $i_1$ and $i_2$, as that would imply $j, i^*, i_2, i_1$ are collinear in that order with $d(i^*, i_2) = \sqrt{\delta_2 \cdot \alpha_j}$ and $d(i^*, i_1) = \sqrt{\delta_1 \cdot \alpha_j}$, so $i_2, i_1$ are too close. Hence, it appears that we have improved one of the main bottlenecks.

Unfortunately, we run into a new issue, if $i_2 = j$ and no other points in $I_1 \cup I_2$ were in $N(j)$. While this case may look good because $|N(j) \cap I_2| = 1$ which is not a bottleneck, it is only not a bottleneck because the contribution of $j$ to the clustering cost and the dual both equal $0$ in this case. So, if we condition on $i_2 \in S$, the cost and dual are not affected by $j$, but if we condition on $i_2 \not\in S$, the cost-dual ratio of $j$ could be $(1+\sqrt{\delta_1})^2$ -- hence, we have again made no improvement. While one could attempt to fix this by creating 
a series of nested independent sets, this approach also fails to work for the same reasons. Hence, we have a new main bottleneck case, where $i_2 = j$, and $j, i^*, i_1$ are collinear in that order with $d(j, i^*) = \sqrt{\alpha_j}$ and $d(i^*, i_1) = \sqrt{\delta_1 \cdot \alpha_j}$ (see Figure \ref{fig:1c}).

We now explain the intuition for fixing this. 
In the main bottleneck case, if we could add the witness $i^*$ of $j$ to $S$, this would reduce $c(j, S)$ significantly if $i_2 \not\in S$, yet does not affect $j$'s contribution to the dual. Unfortunately, adding $i^*$ to $S$ reduces the dual nonetheless, due to other clients. Instead, we will consider a subset of tight facilities that are close, but not too close, to exactly one tight facility in $I_2$ but not close to any other facilities in $I_1 \cup I_2$. In the main bottleneck cases, the witnesses $i^*$ precisely satisfy the condition, as may some additional points. We again prune these points by creating a conflict graph just on these vertices, and pick another maximal independent set $I_3$. Finally, with some probability we will replace each point $i_2 \in I_2$ with the points in $I_3$ that are close to $i_2$. In our main bottleneck case, we will either pick $i^* \in I_3$, or pick another point $i_3$ that is within distance $\sqrt{\delta_1 \cdot \alpha_j}$ of $i^*$, but now must be far from all points in $I_1 \cup I_2$ and therefore will not have distance $(1+\sqrt{\delta_1}) \cdot \sqrt{\alpha_j}$ from $j$ (see Figure \ref{fig:1c}). In addition, $i_3$ might be close, but is not allowed to be too close to $i_2 = j$, so if we replace $i_2$ with $i_3$, we do not run into the issue of $j$'s contribution being $0$ for both the clustering cost and the dual solution. 

Given $I_1, I_2, I_3$, our overall procedure for generating $S$ is to include all of $I_1$, and include points in $I_2$ and $I_3$ with some probability. In addition, each point $i_3 \in I_3$ is close to a unique point $i_2 \in I_2$, and we anti-correlate them being in $S$.
We call the triple $(I_1, I_2, I_3)$ a \emph{nested quasi-independent set}, since $I_1$ and $I_1 \cup I_2$ are independent sets and $I_1 \cup I_2 \cup I_3$ has similar properties, and these sets $I_1, I_1 \cup I_2, I_1 \cup I_2 \cup I_3$ are nested.

We remark that the anti-correlation between selecting $i_2$ and points $i_3$ that are close to $i_2$ is reminiscent of a step of the rounding of bipoint solution of Li and Svensson~\cite{LiS16} (there the
authors combine two solutions by building a collection of stars where the center of the star belongs to a solution, while the leaves belong to another, and choose to open either the center, or the leaves at random). In this context, we will also
allow our algorithm to open slightly $k + C$ centers
(instead of $k$) for some absolute constant $C$. 
Similar to Li and Svensson~\cite{LiS16} and Byrka et al.~\cite{ByrkaPRST17}, we show (Lemma~\ref{lem:center-reduction}) that this is without loss of generality (such an algorithm can be used to obtain an approximation algorithm opening at most $k$ centers). 

The analysis of the approximation bound is very casework heavy, depending on the values of $a=|N(j) \cap I_1|,$ $b=|N(j) \cap I_2|$, and $c=|N(j) \cap I_3|$ for each $j$, and requiring many geometric insights that heavily exploit the structure of the Euclidean metric.

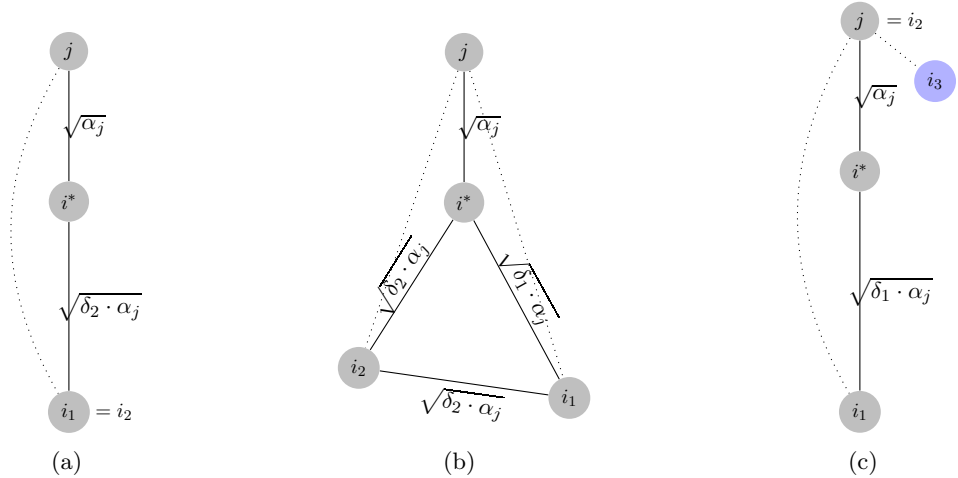
\begin{figure}[htbp]
\centering
    \subfigure[]{
    \begin{tikzpicture}[scale=2.0]
        \node[circle, scale=0.7, fill=gray!50] (1) at (0.0, 0) {$j$};
        \node[circle, scale=0.7, fill=gray!50] (2) at (0.0, -1) {$i^*$};
        \node[circle, scale=0.7, fill=gray!50] (3) at (0.0, -2.4) {$i_1$};
        \node[scale=0.7] at (0.3,-2.4) {$=i_2$};
        
        \draw (1) -- (2);
        \draw (2) -- (3);
        \draw[dotted] (1) to [bend right] (3);

        \node[scale=0.8] at (0.1,-0.5) {$\sqrt{\alpha_j}$};
        \node[scale=0.8] at (0.2,-1.7) {$\sqrt{\delta_2 \cdot \alpha_j}$};
    \label{fig:1a}
    \end{tikzpicture}
    }
    \hspace{2cm}
    \subfigure[]{
    \begin{tikzpicture}[scale=2.0]
        
        \node[circle, scale=0.7, fill=gray!50] (1) at (0, 0) {$j$};
        \node[circle, scale=0.7, fill=gray!50] (2) at (0, -1) {$i^*$};
        \node[circle, scale=0.7, fill=gray!50] (3) at (0.7, -2.3) {$i_1$};
        \node[circle, scale=0.7, fill=gray!50] (4) at (-0.7, -2.1) {$i_2$};
        
        \draw (1) -- (2);
        \draw (2) -- (3);
        \draw (2) -- (4);
        \draw (3) -- (4);
        \draw[dotted] (1) -- (3);
        \draw[dotted] (1) -- (4);
        
        \node[scale=0.8] at (0.1,-0.5) {$\sqrt{\alpha_j}$};
        \node[rotate=58, scale=0.8] at (-0.43,-1.5) {$\sqrt{\delta_2 \cdot \alpha_j}$};
        \node[rotate=-61, scale=0.8] at (0.42,-1.56) {$\sqrt{\delta_1 \cdot \alpha_j}$};
        \node[rotate=-8, scale=0.8] at (-0.02,-2.34) {$\sqrt{\delta_2 \cdot \alpha_j}$};
    \label{fig:1b}
    \end{tikzpicture}
    }
    \hspace{2cm}
    \subfigure[]{
    \begin{tikzpicture}[scale=2.0]
        \node[circle, scale=0.7, fill=gray!50] (1) at (0.0, 0) {$j$};
        \node[circle, scale=0.7, fill=gray!50] (2) at (0.0, -1) {$i^*$};
        \node[circle, scale=0.7, fill=gray!50] (3) at (0.0, -2.6) {$i_1$};
        \node[scale=0.7] at (0.3,0) {$=i_2$};
        
        \node[circle, scale=0.7, fill=blue!30] (4) at (0.5, -0.4) {$i_3$};
        
        \draw (1) -- (2);
        \draw (2) -- (3);
        \draw[dotted] (1) -- (4);
        \draw[dotted] (1) to [bend right] (3);
        
        \node[scale=0.8] at (0.1,-0.5) {$\sqrt{\alpha_j}$};
        \node[scale=0.8] at (0.2,-1.8) {$\sqrt{\delta_1 \cdot \alpha_j}$};
    \label{fig:1c}
    \end{tikzpicture}
    }
\label{fig:euclidean}
\caption{Subfigures (a) and (b) represent the cases when $N(j) \cap I_1, N(j) \cap I_2$ are both empty. In (a), we can bound $d(i^*, i_1)$ as $\sqrt{\delta_2 \cdot \alpha_j} < \sqrt{\delta_1 \cdot \alpha_j}$, and in (b), we can obtain better bounds for $d(j, i_1)$ and $d(j, i_2)$ than by just the triangle inequality. Subfigure (c) represents the problem with having nested two independent sets, when $N(j) \cap I_1 = \emptyset$ but $N(j) \cap I_2$ has a unique point $i_2 = j$. If we replace $i_2 = j$ with the blue node $i_3$ with some probability, the cost-to-dual ratio improves.}
\end{figure}

Finally, we remark that this method of generating a set $S$ of centers also can be used to provide an LMP approximation below $1+\sqrt{2}$ for Euclidean $k$-median. However, due to the differing nature of what the bottleneck cases are, the choices of $\delta_1, \delta_2$, and the construction of the final set $I_3$ will differ. Both to obtain and break the $1+\sqrt{2}$-approximation for Euclidean $k$-median, one must understand how the sum of distances from a facility $j$ to its close facilities $i_1, \dots, i_a \in N(j) \cap S$ relates to the pairwise distances between $i_1, \dots, i_a$. In the $k$-means case, the squared distances make this straightfoward, but the $k$-median case requires more geometric insights (see Lemma \ref{lem:geometric_median}) to improve the $1+\sqrt{8/3}$-approximation of 
Ahmadian et al.~\cite{ahmadian2017better} to a $1+\sqrt{2}$-approximation, even with the same algorithm.

\paragraph{Polynomial-Time Algorithm:}
While we can successfully obtain an LMP $(3+2\sqrt{2})$-approximation for Euclidean $k$-means, this does not imply a $(3+2\sqrt{2})$-approximation for the overall Euclidean $k$-means problem, since our LMP approximation may not produce the right number of cluster centers. To address this issue, Ahmadian et al.~\cite{ahmadian2017better} developed a method of slowly raising the parameter $\lambda$ and generating a series of dual solutions $\alpha^{(0)}, \alpha^{(1)}, \alpha^{(2)}, \dots$, where each pair of consecutive solutions $\alpha^{(j)}$ and $\alpha^{(j+1)}$ are sufficiently similar. The procedure of generating dual solutions $\alpha^{(j)}$ in~\cite{ahmadian2017better} is very complicated, but will not change significantly from their result to ours.
Given these dual solutions, \cite{ahmadian2017better} shows how to interpolate between consecutive dual solutions $\alpha^{(j)}$ and $\alpha^{(j+1)}$, creating a series of conflict graphs where the subsequent graph removes at most one facility at a time (but occasionally adds many facilities).
This property of removing at most one facility ensures that 
the maximal independent set decreases by at most 1 at each step (but could increase by a large number if the removed vertex allows for more points to be added). Using an intermediate-value theorem, they could ensure that at some point, there is an independent set $I$ of the conflict graph 
with size exactly $k$. Hence, they apply the LMP technique to obtain the same approximation ratio.

In our setting, we are not so lucky, because we are dealing with nested independent sets (and even more complicated versions of them). Even if the size of the first part $I_1$ never decreases by more than $1$ at a time, $I_1$ could increase which could potentially make the sizes of $I_2$ or $I_3$ decrease rapidly, even if we only remove one facility from the conflict graph. To deal with this, we instead consider the first time that the expected size of $S$ drops below $k$ (where we recall that all points in $I_1$ are in $S$, but points in $I_2$ and $I_3$ are inserted into our final set $S$ with some probability). Let $(I_1, I_2, I_3)$ represent the set of chosen facilities right before the expected size of $S$ drops below $k$ for the first time, and $(I_1', I_2', I_3')$ represent the set of chosen facilities right after.

To obtain size exactly $k$, we show that one can always do one of the following: either (i) modify the probabilities of selecting points in $I_2$ and $I_3$ so that the size is exactly $k$, or (ii) use submodularity properties of the $k$-means objective function to interpolate between the set $S$ generated by $(I_1, I_2, I_3)$ and the set $S'$ generated by $(I_1', I_2', I_3')$. While $|S| \ge k$ and $|S'| < k$, we do not necessarily have that $S' \subset S$, and in fact we will interpolate between $S'$ and $S \cup S'$ instead by adding random points from $S \backslash S'$ to $S'$. These procedures are relatively computational and require a good understanding of how the dual objective behaves as we modify the probabilities of selecting elements. In addition, because we modify the probabilities to make the size exactly $k$ and because we interpolate between $S'$ and $S \cup S'$ instead of $S'$ and $S$ to obtain our final set of centers, we lose a slight approximation factor from the LMP approximation, but we still significantly improve over the old bound. In addition, this procedure works well for the Euclidean $k$-means \emph{and} $k$-median problems.

One minor problem that we will run into is that a point $i_2$ may have many points in $I_3$ that are all ``close'' to it. As a result, even if the expected size of $S$ is exactly $k$, the variance may be large because we are either selecting $i_2$ or all of the points that are close to $i_2$. In this case, we can use the submodularity properties of the $k$-means objective and coupling methods from probability theory to show that if $i_2$ has too many close points in $I_3$, we could instead include $i_2$ and include an average number of these close points from $I_3$, without increasing the expected cost.

\section{LMP Approximation for Euclidean $k$-means} \label{sec:lmp_k_means}

In this section, we provide a $(3+2\sqrt{2})$ LMP approximation algorithm (in expectation) for the Euclidean $k$-means problem. Our algorithm will be parameterized by four parameters, $\delta_1, \delta_2, \delta_3$, and $p$. We fix $\delta_1, \delta_2, \delta_3$ and allow $p$ to vary, and obtain an approximation constant $\rho(p)$ that is a function of $p$: for appropriate choices of $p$, $\rho(p)$ will equal $3+2\sqrt{2}$. 

In Subsection \ref{subsec:k_means_alg}, we describe the LMP algorithm, which is based on the LMP approximation algorithms by Jain and Vazirani~\cite{jain2001lagrangian} and Ahmadian et al.~\cite{ahmadian2017better}, but using our technique of generating what we call a nested quasi-independent set. In Subsection \ref{subsec:lmp_k_means_bash}, we analyze the approximation ratio, which spans a large amount of casework.

\subsection{The algorithm and setup} \label{subsec:k_means_alg}

Recall the conflict graph $H := H(\delta)$, where we define two tight facilities $(i, i')$ to be connected if $c(i, i') \le \delta \cdot \min(t_i, t_{i'}).$ We set parameters $\delta_1 \ge \delta_2 \ge 2 \ge \delta_3$ and $0 < p < 1$, and define $V_1$ to be the set of all tight facilities. Given the set of tight facilities $V_1$ and conflict graphs $H(\delta)$ over $V_1$ for all $\delta > 0$, our algorithm works by applying the procedure described in Algorithm \ref{alg:lmp} (on the next page) to $V_1$, with parameters $\delta_1, \delta_2, \delta_3$, and $p$.

\begin{figure}[htbp]
\centering
\begin{algorithm}[H]
    \caption{Generate a Nested Quasi-Independent Set of $V_1$, as well as a set of centers $S$ providing an LMP approximation for Euclidean $k$-means}\label{alg:lmp}
\Call{LMP}{$V_1, \{H(\delta)\}, \delta_1, \delta_2, \delta_3, p$}:
    \begin{algorithmic}[1] 
        \State Create a maximal independent set $I_1$ of $H(\delta_1)$. \label{step:I1}
        \State Let $V_2$ be the set of points in $V_1 \backslash I_1$ that are not adjacent to $I_1$ in $H(\delta_2)$. \label{step:V2}
        \State Create a maximal independent set $I_2$ of the induced subgraph $H(\delta_2)[V_2]$. \label{step:I2}
        \State Let $V_3$ be the set of points $i$ in $V_2 \backslash I_2$ such that there is exactly one point in $I_2$ that is a neighbor of $i$ in $H(\delta_2)$, there are no points in $I_1$ that are neighbors of $i$ in $H(\delta_2)$, and there are no points in $I_2$ that are neighbors of $i$ in $H(\delta_3)$. \label{step:V3}
        \State Create a maximal independent set $I_3$ of the induced subgraph $H(\delta_2)[V_3]$. \label{step:I3}
        \State Note that every point $i \in I_3$ has a unique adjacent neighbor $q(i) \in I_2$ in $H(\delta_2)$. We create the final set $S$ as follows: \label{step:S}
        \begin{itemize}
            \item Include every point $i \in I_1$.
            \item For each point $i \in I_2$, flip a fair coin. If the coin lands heads, include $i$ with probability $2p$. Otherwise, include each point in $q^{-1}(i)$ independently with probability $2p$.
        \end{itemize}
    \end{algorithmic}
\end{algorithm}
\end{figure}

We will call the triple $(I_1, I_2, I_3)$ generated by Algorithm \ref{alg:lmp} a \textbf{nested quasi-independent set}. Although $I_1, I_2, I_3$ are disjoint, we call it a nested quasi-independent set since $I_1 \subset I_1 \cup I_2 \subset I_1 \cup I_2 \cup I_3$ are nested, and $I_1$ is a maximal independent set for $H(\delta_1)$ and $I_1 \cup I_2$ is a maximal independent set for $H(\delta_2)$. While $I_1 \cup I_2 \cup I_3$ is not an independent set, it shares similar properties. 
As described in the technical overview (and in Algorithm \ref{alg:lmp}), the LMP approximation algorithm uses $(I_1, I_2, I_3)$ to create our output set of centers $S$. $S$ contains all of $I_1$ and each point in $I_2 \cup I_3$ with probability $p$, but the choices of which points in $I_2 \cup I_3$ are in $S$ are not fully independent.

For the dual solution $\alpha = \{\alpha_j\}_{j \in \mathcal{D}}$ and values $t_i$ for tight $i$ as generated as in Subsection \ref{subsec:witness_and_conflict}, we note the following simple yet crucial facts.

\begin{proposition} \label{prop:t_bounds} \cite{ahmadian2017better} The following hold.
\begin{enumerate}
    \item For any client $j$ and its witness $i$, $i$ is tight and $\alpha_j \ge t_i$. 
    \item For any client $j$ and its witness $i$, $\alpha_j \ge c(j, i)$.
    \item For any tight facility $i$ and any client $j' \in N(i)$, $t_i \ge \alpha_{j'}$.
\end{enumerate}
\end{proposition}

These will essentially be the only facts relating witnesses and clients that we will need to use.

\subsection{Main lemma} \label{subsec:lmp_k_means_bash}

We consider a more general setup, as it will be required when converting the LMP approximation to a full polynomial-time algorithm. Let $\mathcal{V} \subset \mathcal{F}$ be a subset of facilities (for instance, this may represent the set of tight facilities) and let $\mathcal{D}$ be the full set of clients. For each $j \in \mathcal{D},$ let $\alpha_j \ge 0$ be some real number, and for each $i \in \mathcal{V}$, let $t_i \ge 0$ be some real number. For each client $j \in \mathcal{D}$, we associate with it a set $N(j) \subset \mathcal{V}$. (For instance, this could be the set $N(j) = \{i \in \mathcal{V}: \alpha_j > c(j, i)\}$).
In addition, suppose that for each client $j \in \mathcal{D},$ there exists a ``witness'' facility $w(j) \in \mathcal{V}$.
Finally, suppose that we have the following assumptions. (These assumptions will hold by Proposition \ref{prop:t_bounds} when $\alpha = \{\alpha_j\}$ is generated by the procedure in Subsection \ref{subsec:witness_and_conflict}, $\mathcal{V}$ is the set of tight facilities, and $N(j) =  \{i \in \mathcal{V}: \alpha_j > c(j, i)\}.$)

\begin{enumerate}
    \item For any client $j \in \mathcal{D}$, the witness $w(j) \in \mathcal{V}$ satisfies $\alpha_j \ge t_{w(j)}$ and $\alpha_j \ge c(j, w(j))$.
    \item For any client $j \in \mathcal{D}$ and any facility $i \in N(j)$, $t_i \ge \alpha_{j} > c(j, i)$.
\end{enumerate}



For the Euclidean k-means problem with the above assumptions, we will show the following:
\begin{lemma} \label{lem:main_lmp}
    Consider the set of conflict graphs $\{H(\delta)\}_{\delta > 0}$ created on the vertices $\mathcal{V}$, where $(i, i')$ is an edge in $H(\delta)$ if $c(i, i') \le \delta \cdot \min(t_i, t_{i'})$. Fix $\delta_1 = \frac{4+8\sqrt{2}}{7} \approx 2.1877, \delta_2 = 2, \delta_3 = 6-4\sqrt{2} \approx 0.3432,$ and let $p < 0.5$ be variable. Now, let $S$ be the randomized set created by applying Algorithm \ref{alg:lmp} on $V_1 = \mathcal{V}$.
    Then, for any $j \in \mathcal{D}$,
\begin{equation}
    \BE[c(j, S)] \le \rho(p) \cdot \BE\left[\alpha_j - \sum_{i \in N(j) \cap S} (\alpha_j-c(j, i))\right], \label{eq:main_lmp}
\end{equation}
    where $\rho(p)$ is some constant that only depends on $p$ (since $\delta_1, \delta_2, \delta_3$ are fixed).
\end{lemma}

\begin{remark}
    While we have not defined $\rho(p)$, we will implicitly define it through our cases. We provide detailed visualizations of the bounds we obtain for each of our cases in Desmos (see Appendix \ref{app:files} for the links). Importantly, we show that we can set $p$ such that $\rho(p) \le 3+2\sqrt{2}$.
\end{remark}

To see why Lemma \ref{lem:main_lmp} implies an LMP approximation (in expectation), fix some $\lambda \ge 0$. Then, we perform the procedure in Subsection \ref{subsec:witness_and_conflict} and let $\mathcal{V}$ be the set of tight facilities, $N(j) = \{i \in \mathcal{V}: \alpha_j > c(j, i)\}$, and $N(i) = \{j \in \mathcal{D}: \alpha_j > c(j, i)\}$. Then, by adding \eqref{eq:main_lmp} over all $j \in \mathcal{D}$, we have that
\begin{align*}
    \BE[\text{cost}(\mathcal{D}, S)] &\le \rho(p) \cdot \BE\left[\sum_{j \in \mathcal{D}} \alpha_j - \sum_{j \in \mathcal{D}} \sum_{i \in N(j) \cap S} (\alpha_j-c(j,i))\right] \\
    &= \rho(p) \cdot \left(\sum_{j \in \mathcal{D}} \alpha_j - \BE\left[\sum_{i \in S} \sum_{j \in N(i)} (\alpha_j-c(j, i))\right]\right) \\
    &= \rho(p) \cdot \left(\sum_{j \in \mathcal{D}} \alpha_j - \lambda \cdot \BE[|S|]\right).
\end{align*}
    Above, the final line follows because $\sum_{j \in N(i)} (\alpha_j-c(j, i)) = \sum_{j \in \mathcal{D}} \max(\alpha_j-c(j, i), 0) = \lambda$, since $i$ is assumed to be tight.
    Thus, we obtain an LMP approximation with approximation factor $\rho(p)$ for any choice of $p$. Given this, we now prove Lemma \ref{lem:main_lmp}.

\begin{proof}[Proof of Lemma \ref{lem:main_lmp}]
    We fix $j$ and do casework based on the sizes $a = |N(j) \cap I_1|$, $b = |N(j) \cap I_2|,$ and $c = |N(j) \cap I_3|$. We show that $\BE[c(j, S)] \le \rho(p) \cdot \BE\left[\alpha_j - \sum_{i \in N(j) \cap S} (\alpha_j-c(j, i))\right]$ for each case of $a, b, c.$ We will call $\BE[c(j, S)]$ the numerator and $\BE\left[\alpha_j - \sum_{i \in N(j) \cap S} (\alpha_j-c(j, i))\right]$ the denominator, and attempt to show this fraction is at most $\rho(p).$ By scaling all distances (and all $\alpha_j, t_i$ values accordingly), we will assume WLOG that $\alpha_j = 1$, so $d(j, w(j)), t_{w(j)} \le 1$.
    
    First, we have the following basic proposition.
\begin{proposition} \label{prop:crude_distance}
    $d(j, S) \le d(j, I_1) \le 1+\sqrt{\delta_1}$.
\end{proposition}

\begin{proof}
    Note that $d(j, I_1) \le d(j, w(j)) + d(w(j), I_1)$. But $d(j, w(j)) \le 1$ and by properties of the independent set $I_1$, there exists $i_1 \in I_1$ such that $d(w(j), i_1) \le \sqrt{\delta_1 \cdot \min(t_{w(j)}, t_{i_1})} \le \sqrt{\delta_1}$, since $t_{w(j)} \le 1$. So, $d(j, I_1) \le 1 + \sqrt{\delta_1}$, as desired.
\end{proof}

    In addition, we have the following simple proposition about Euclidean space, the proof of which is deferred to Appendix \ref{app:bash}.

\begin{proposition} \label{prop:triangle}
    Suppose that we have $4$ points $A, B, C, D$ in Euclidean space and parameters $\nu_1, \nu_2, \nu_3, \sigma_1, \sigma_2, \sigma_3 \ge 0$ such that $d(A, B)^2 \le 1$, $d(B, C)^2 \le \nu_1 \cdot \min(\sigma_1, \sigma_2)$, $d(B, D)^2 \le \nu_2 \cdot \min(\sigma_1, \sigma_3)$, and $d(C, D)^2 \ge \nu_3 \cdot \min(\sigma_2, \sigma_3)$. Moreover, assume that $\sigma_1 \le 1$ and that $\nu_1, \nu_2 \ge \nu_3$. Then, 
\[p \cdot \|C-A\|_2^2 + (1-p) \cdot \|D-A\|_2^2 \le 1 + p \cdot \nu_1 + (1-p) \cdot \nu_2 + 2\sqrt{p \cdot \nu_1 + (1-p) \cdot \nu_2 - p(1-p) \cdot \nu_3}.\]
\end{proposition}

    Finally, we will also make frequent use of the well-known fact that for any set of $h \ge 1$ points $I = \{i_1, \dots, i_h\}$ and any $j$ in Euclidean space, $\frac{1}{2h} \sum_{i, i' \in I} \|i-i'\|_2^2 \le \sum_{i \in I} \|i-j\|_2^2.$ As a direct result of this, if $c(j, i) \le 1$ for all $i \in I$ but $c(i, i') \ge 2$ for all $i \neq i' \in I$, then $\sum_{i \in I} (1-c(j, i)) \le 1$.

    \medskip

    We are now ready to investigate the several cases needed to prove Lemma \ref{lem:main_lmp}.
    
\paragraph{Case 1: $\boldsymbol{a=0, b=1, c=0}$.}
Let $i_2$ be the unique point in $I_2 \cap N(j),$ and let $i^* = w(j)$ be the witness of $j$. Recall that $i^*$ is tight, so $i^* \in \mathcal{V}_1$. Note that $d(j, i^*) \le 1$ and $t_{i^*} \le 1$.

There are numerous sub-cases to consider, which we enumerate.

\begin{enumerate}[label=\alph*)]
    \item $\boldsymbol{i^* \not\in V_2}$.
    In this case, either $i^* \in I_1$ so $d(i^*, I_1) = 0$, or there exists $i_1 \in I_1$ such that $d(i^*, i_1) \le \sqrt{\delta_2 \cdot \min(t_{i^*}, t_{i_1})} \le \sqrt{\delta_2}$. So, $d(j, I_1) \le 1+\sqrt{\delta_2}$. In addition, we have that $i_2 \in S$ with probability $p$. So, if we let $t := d(j, i_2)$, we can bound the fraction by
\[\frac{p \cdot t^2 + (1-p) \cdot (1+\sqrt{\delta_2})^2}{1-p(1-t^2)} = \frac{p \cdot t^2 + (1-p) \cdot (1+\sqrt{\delta_2})^2}{p \cdot t^2 + (1-p)}.\]
    Note that $0 \le t < 1$ since $i_2 \in N(j),$ and the above fraction is maximized for $t = 0$, in which case we get that the fraction is at most 
\begin{equation} \tag{1.a}\label{eq:1.a}
    (1+\sqrt{\delta_2})^2.
\end{equation}

    \item $\boldsymbol{i^* \in V_3}.$ In this case, there exists $i_3 \in I_3$ (possibly $i_3 = i^*$) such that $d(i^*, i_3) \le \sqrt{\delta_2 \cdot \min(t_{i^*}, t_{i_3})}.$ In addition, there exists $i_1 \in I_1$ such that $d(i^*, i_1) \le \sqrt{\delta_1 \cdot \min(t_{i^*}, t_{i_1})}$. Finally, since $I_3 \subset V_2$, we must have that $d(i_1, i_3) \ge \sqrt{\delta_2 \cdot \min(t_{i_1}, t_{i_3})}$. If we condition on $i_2 \in S$, then the numerator and denominator both equal $c(j, i_2)$, so the fraction is $1$ (or $0/0$). Else, if we condition on $i_2 \not\in S$, then the denominator is $1$, and the numerator is at most $p \cdot \|i_3-j\|_2^2 + (1-p) \cdot \|i_1-j\|_2^2$, since $i_1 \in S$ always, and either $q(i_3) \neq i_2$, in which case $\BP(i_3 \in S|i_2 \not\in S) = p$, or $q(i_3) = i_2$, in which case $\BP(i_3 \in S|i_2 \not\in S) = \frac{p}{1-p} \ge p$.
    
    Note that $d(j, i^*) \le 1$, that $t_{i^*} \le 1$, and that $\delta_2, \delta_1 \ge \delta_2$. So, we may apply Proposition \ref{prop:triangle} with $A=j, B=i^*, C=i_3, D=i_1$ and $\nu_1 = \nu_3 = \delta_2, \nu_2 = \delta_1, \sigma_1=t_{i^*}, \sigma_2=t_{i_3}, \sigma_3=t_{i_1}$ to bound numerator (and thus the overall fraction since the denominator equals $1$) by
\begin{equation} \tag{1.b}\label{eq:1.b}
    1 + p \cdot \delta_2 + (1-p) \cdot \delta_1 + 2\sqrt{p^2 \cdot \delta_2 + (1-p) \cdot \delta_1}.
\end{equation}
\end{enumerate}
    In the remaining cases, we may assume that $i^*\in V_2 \backslash V_3$. Then, one of the following must occur:
\begin{enumerate}[label=\alph*)] \setcounter{enumi}{2}
    \item $\boldsymbol{i^* = i_2}$. In this case, define $t = d(j, i^*) \in [0, 1]$, and note that $d(j, I_1) \le d(j, i^*)+d(i^*, I_1) \le t + \sqrt{\delta_1}$. So, with probability $p$, we have that $d(j, S) \le d(j, i^*) = t$, and otherwise, we have that $d(j, S) \le d(j, I_1) = t + \sqrt{\delta_1}$. So, we can bound the ratio by
\begin{equation} \tag{1.c}\label{eq:1.c}
    \max_{0 \le t \le 1} \frac{p \cdot t^2 + (1-p) \cdot (t + \sqrt{\delta_1})^2}{1-p \cdot (1-t^2)} \le \max\left((\sqrt{0.75}+\sqrt{\delta_1})^2, \frac{(1-p) \cdot (1+\sqrt{\delta_1})^2 + 3p/4}{1-p/4}\right).
\end{equation}
    We prove this final inequality in Appendix \ref{app:bash}.

    \item $\boldsymbol{i^* \in I_2}$ \textbf{but} $\boldsymbol{i^* \neq i_2}$. First, we recall that $d(j, i^*) \le 1$. Now, let $t = d(j, i_2)$. In this case, with probability $p$, $d(j, S) \le t$ (if we select $i_2$ to be in $S$), with probability $p(1-p)$, $d(j, S) \le 1$ (if we select $i^*$ but not $i_2$ to be in $S$), and in the remaining event of $(1-p)^2$ probability, we still have that $d(j, S) \le d(j, I_1) \le 1 + \sqrt{\delta_1}$ by Proposition \ref{prop:crude_distance}. So, we can bound the ratio by
\[\max_{0 \le t \le 1} \frac{p \cdot t^2 + p(1-p) \cdot 1 + (1-p)^2 \cdot (1+\sqrt{\delta_1})^2}{1-p \cdot (1-t^2)}.\]
    Note that this is maximized when $t = 0$ (since the numerator and denominator increase at the same rate when $t$ increases), so we can bound the ratio by
\begin{equation} \tag{1.d} \label{eq:1.d}
    \frac{p(1-p) + (1-p)^2 \cdot (1+\sqrt{\delta_1})^2}{1-p} = p+(1-p) \cdot (1+\sqrt{\delta_1})^2.
\end{equation}
    
    \item \textbf{There is more than one neighbor of} $\boldsymbol{i^*}$ \textbf{in} $\boldsymbol{H(\delta_2)}$ \textbf{that is in} $\boldsymbol{I_2}$. In this case, there is some other point $i_2' \in I_2$ not in $N(j)$ such that $d(i^*, i_2') \le \sqrt{\delta_2 \cdot \min(t_{i^*}, t_{i_2'})}.$ So, we have four points $j, i^*, i_1 \in I_1, i_2' \in I_2$ such that $d(j, i^*) \le 1,$ $d(i^*, i_2') \le \sqrt{\delta_2 \cdot \min(t_{i^*}, t_{i_2'})},$ $d(i^*, i_1) \le \sqrt{\delta_1 \cdot \min(t_{i^*}, t_{i_1})},$ and $d(i_1, i_2') \ge \sqrt{\delta_2 \cdot \min(t_{i_1}, t_{i_2'})}.$
    
    If we condition on $i_2 \in S$, then the denominator equals $c(j, i_2)$ and the numerator is at most $c(j, i_2)$, so the fraction is $1$ (or $0/0$). Else, if we condition on $i_2 \not\in S$, then the denominator is $1$, and the numerator is at most $p \cdot \|i_2'-j\|_2^2 + (1-p) \cdot \|i_1-j\|_2^2$. Note that $d(j, i^*) \le 1$, that $t_{i^*} \le 1$, and that $\delta_2, \delta_1 \ge \delta_2$. So, as in \ref{eq:1.b}, we may apply Proposition \ref{prop:triangle} to bound the ratio by
\begin{equation} \tag{1.e} \label{eq:1.e}
    1 + p \cdot \delta_2 + (1-p) \cdot \delta_1 + 2\sqrt{p^2 \cdot \delta_2 + (1-p) \cdot \delta_1}.
\end{equation}
    
    \item \textbf{There are no neighbors of} $\boldsymbol{i^*}$ \textbf{in} $\boldsymbol{H(\delta_2)}$ \textbf{that are in} $\boldsymbol{I_2}$. 
    In this case, we would actually have that $i^* \in I_2,$ because we defined $I_2$ to be a maximal independent set in the induced subgraph $H(\delta_2)[V_2].$ So, if there were no such neighbors and $i^* \in V_2$, then we could add $i^*$ to $I_2$, contradicting the maximality of $I_2$. Having $i^* \in I_2$ was already covered by subcases c) and d).

    \item \textbf{There is a neighbor of} $\boldsymbol{i^*}$ \textbf{in} $\boldsymbol{H(\delta_3)}$ \textbf{that is also in} $\boldsymbol{I_2}$, which means that either $d(i^*, i_2) \le \sqrt{\delta_3 \cdot t_{i^*}}$ so $d(i_2, j) \ge \max(0, d(j, i^*)-\sqrt{\delta_3 \cdot t_{i^*}})$, or there is some other point $i_2' \in I_2$ not in $N(j)$ such that $d(i^*, i_2') \le \sqrt{\delta_3 \cdot \min(t_{i^*}, t_{i_2'})}.$ If $d(i^*, i_2) \le \sqrt{\delta_3 \cdot t_{i^*}}$, then define $t = t_{i^*}$ and $u = d(j, i^*)$. In this case, $d(j, I_1) \le u + \sqrt{\delta_1 \cdot t},$ and $d(j, i_2) \ge \max(0, u-\sqrt{\delta_3 \cdot t}).$ 
    Since $t = t_{i^*} \le 1$ and $u = d(j, i^*) \le 1$, we can bound the overall fraction as at most
\begin{align}
    &\hspace{0.5cm} \max_{0 \le t \le 1} \max_{0 \le u \le 1} \max_{d(j, i_2) \ge \max(0, u-\sqrt{\delta_3 \cdot t})} \frac{(1-p) \cdot (u+\sqrt{\delta_1 \cdot t})^2+p \cdot d(j, i_2)^2}{1-p+p \cdot d(j, i_2)^2} \nonumber \\
    &\le \max\left((\sqrt{\delta_1}+\sqrt{\delta_3})^2, \frac{(1-p) \cdot (1+\sqrt{\delta_1})^2 + p \cdot (1-\sqrt{\delta_3})^2}{1-p+p \cdot (1-\sqrt{\delta_3})^2}\right). \tag{1.g.i} \label{eq:1.g.i}
\end{align}
    We derive the final inequality in Appendix \ref{app:bash}.
    
    Alternatively, if $d(i^*, i_2') \le \sqrt{\delta_3 \cdot \min(t_{i^*}, t_{i_2'})},$ then if we condition on $i_2 \in S,$ the fraction is $1$ (or $0/0$), and if we condition on $i_2 \not\in S$, the denominator is $1$ and the numerator is at most $p \cdot d(j, i_2')^2 + (1-p) \cdot d(j, i_1)^2 \le p \cdot (1+\sqrt{\delta_3})^2+(1-p) \cdot (1+\sqrt{\delta_1})^2.$ (Note that $i_2 \in S$ and $i_2' \in S$ are independent.) Therefore, we can also bound the overall fraction by 
\begin{equation} \tag{1.g.ii} \label{eq:1.g.ii}
    p \cdot (1+\sqrt{\delta_3})^2+(1-p) \cdot (1+\sqrt{\delta_1})^2.
\end{equation}

    \item \textbf{There is a neighbor of} $\boldsymbol{i^*}$ \textbf{in} $\boldsymbol{H(\delta_2)}$ \textbf{that is also in} $\boldsymbol{I_1}$. In this case, $i^*$ would not be in $V_2$, so we are back to sub-case 1.a.
\end{enumerate}

\paragraph{Case 2: $\boldsymbol{a = 0, b = 1, c \ge 1}$.}

Let $i_2$ be the unique point in $N(j) \cap I_2,$ and let $i_3^{(1)}, \dots, i_3^{(c)}$ represent the points in $N(j) \cap I_3$. Let $c_1$ be the number of points in $N(j) \cap I_3$ that are in $q^{-1}(i_2)$, and let $c_2 = c-c_1$ be the number of points in $N(j) \cap I_3$ not in $q^{-1}(i_2).$ We will have four subcases. For simplicity, in this case we keep $\delta_2 = 2.$

Before delving into the subcases, we first prove the following propositions regarding the probability of some point in $I_3$ being selected.

\begin{proposition} \label{prop:anticor_1}
    Let $c = |N(j) \cap I_3|$. Then, the probability that no point in $N(j) \cap I_3$ is in $S$ is at most $\frac{1}{2} + \frac{1}{2} (1-2p)^c$.
\end{proposition}

\begin{proof}
    First, note that $\left(\frac{1}{2} + \frac{1}{2} x\right) \cdot \left(\frac{1}{2} + \frac{1}{2} y\right) + \left(\frac{1}{2} - \frac{1}{2} x\right) \cdot \left(\frac{1}{2} - \frac{1}{2} y\right) = \frac{1}{2} + \frac{1}{2} xy$. Therefore, if $0 \le x, y \le 1$, then $\left(\frac{1}{2} + \frac{1}{2} x\right) \cdot \left(\frac{1}{2} + \frac{1}{2} y\right) \le \frac{1}{2} + \frac{1}{2} xy.$ In general, through induction we have that for any $0 \le x_1, \dots, x_r \le 1,$ that $\prod_{s = 1}^{r} \left(\frac{1}{2} + \frac{1}{2} x_s\right) \le \frac{1}{2} + \frac{1}{2} x_1 \cdots x_r$.
    
    Now, group the points $i_3^{(1)}, \dots, i_3^{(c)} \in N(j) \cap I_3$ into $r \le c$ groups of sizes $c_1, \dots, c_r$, where each group is points that map to the same point in $I_2$ under $q$. Then, for each group $s$, the probability that no point in the group is in $S$ is precisely $\frac{1}{2} + \frac{1}{2} (1-2p)^{c_s}$, because with probability $\frac{1}{2}$ we will only consider picking the point $q(i)$ (for $i$ in group $s$), and otherwise each point in the group will still not be in $S$ with probability $1-2p$. So, the overall probability is 
\[\prod_{s = 1}^{r} \left(\frac{1}{2} + \frac{1}{2} (1-2p)^{c_s}\right) \le \frac{1}{2} + \frac{1}{2} (1-2p)^{c_1+\cdots+c_r} = \frac{1}{2} + \frac{1}{2} (1-2p)^c. \qedhere\] 
\end{proof}

We also note the following related proposition.

\begin{proposition} \label{prop:anticor_2}
    Let $c = |N(j) \cap I_3|$, and $i' = q(i)$ for some arbitrary $i \in N(j) \cap I_3$. Then, the probability that no point in $N(j) \cap I_3$ nor $i'$ is in $S$ is at most $\frac{1}{2} (1-2p) + \frac{1}{2} (1-2p)^c$.
\end{proposition}

\begin{proof}
    Similar to the previous proposition, we group the points $i_3^{(1)}, \dots, i_3^{(c)} \in N(j) \cap I_3$ into $r \le c$ groups of sizes $c_1, \dots, c_r$. Assume WLOG that $i$ is in the first group. Then, the probability that no point in the first group nor $i'$ is in $S$ is $\frac{1}{2} (1-2p) + \frac{1}{2} (1-2p)^{c_1} = (1-2p) \cdot \left(\frac{1}{2} + \frac{1}{2} (1-2p)^{c_1-1}\right)$. So, the overall probability is 
\begin{align*}
    (1-2p) \cdot \left(\frac{1}{2} + \frac{1}{2} (1-2p)^{c_1-1}\right) \cdot \prod_{s = 2}^{r} \left(\frac{1}{2} + \frac{1}{2} (1-2p)^{c_s}\right) &\le (1-2p) \cdot \left(\frac{1}{2} + \frac{1}{2} (1-2p)^{(c_1-1)+c_2+\cdots+c_r}\right) \\
    &= \frac{1}{2} (1-2p) + \frac{1}{2} (1-2p)^c. \qedhere
\end{align*} 
\end{proof}

\begin{enumerate}[label=\alph*)]
    \item $\boldsymbol{c_1 = 0}$. In this case, we have that no pair of points in $N(j) \cap (I_2 \cup I_3)$ are connected in $H(\delta_2)$, which means that they have pairwise distances at least $\sqrt{\delta_2}$ from each other (since $t_i \ge 1$ if $i \in N(j)$). So, $\sum_{i \in N(j) \cap (I_2 \cup I_3)} c(j, i) \ge \frac{1}{2(c+1)} \cdot c(c+1)(\sqrt{\delta_2})^2 = c$ since $|N(j) \cap (I_2 \cup I_3)| = c+1$ and $\delta_2 = 2$. Consequently, $\sum_{i \in N(j) \cap (I_2 \cup I_3)} (1-c(j, i)) \le (c+1)-c = 1$. Therefore, the denominator is at least $1-p$. To bound the numerator, we note that the probability of none of the points in $I_2 \cup I_3$ being in $S$ is at most $(1-p) \cdot \left(\frac{1}{2} + \frac{1}{2}(1-2p)^{c_2}\right)$. This is because $i_2 \not\in S$ with probability $1-p$, no point in $N(j) \cap I_3$ is in $S$ with probability at most $\frac{1}{2} + \frac{1}{2} (1-2p)^{c_2}$ by Proposition \ref{prop:anticor_1}, and these two events are independent since $c_1 = 0$.
    If some point in $N(j) \cap (I_2 \cup I_3)$ is in $S$, then $c(j, S) \le 1$, and otherwise, $c(j, S) \le (1+\sqrt{\delta_1})^2$. Therefore, we can bound the numerator as at most $(1-p) \cdot \left(\frac{1}{2} + \frac{1}{2}(1-2p)^{c_2}\right) \cdot (1+\sqrt{\delta_1})^2 + \left(1 - (1-p) \cdot \left(\frac{1}{2} + \frac{1}{2}(1-2p)^{c_2}\right)\right).$ Overall, we have that the fraction is at most
\begin{equation} \tag{2.a} \label{eq:2.a}
    \frac{(1-p) \cdot \left(\frac{1}{2} + \frac{1}{2}(1-2p)^{c_2}\right) \cdot (1+\sqrt{\delta_1})^2 + \left(1 - (1-p) \cdot \left(\frac{1}{2} + \frac{1}{2}(1-2p)^{c_2}\right)\right)}{1-p}.
\end{equation}

    \item $\boldsymbol{c_2 = 0}$ \textbf{and} $\boldsymbol{c \ge 2}$. In this case, we have that $q(i_3^{(k)}) = i_2$ for all $k \in [c]$ (note that $c = c_1$). Since the points $i_3^{(k)}$ for all $k \in [c]$ have pairwise distance at least $\sqrt{2}$ from each other, we have that $\sum_{i \in N(j) \cap I_3} (1-c(j, i)) \le 1$. Letting $t$ be such that $1-t = \sum_{i \in N(j) \cap I_3} (1-c(j, i)),$ we have that $\BE_{i \sim N(j) \cap I_3} c(j, i) = \frac{c-1}{c} +\frac{t}{c}$. In addition, let $u = d(j, i_2).$ In this case, the denominator of our fraction is $1-p(1-u^2)-p(1-t) = p \cdot u^2 + p \cdot t + (1-2p).$ To bound the numerator, with probability $p$ we have that $i_2 \in S$, in which case $c(j, S) \le u^2$. In addition, there is a disjoint $\frac{1}{2} \cdot \left(1-(1-2p)^c\right)$-probability event where some $i_3^{(k)} \in S$, and conditioned on this event, $\BE[c(j, S)] \le \frac{c-1}{c} + \frac{t}{c}$. Otherwise, we still have that $c(j, S) \le c(j, I_1) \le (1+\sqrt{\delta_1})^2$. So, overall we have that the fraction is at most
\[\frac{p \cdot u^2 + \frac{1}{2} \cdot \left(1-(1-2p)^c\right) \cdot \left(\frac{c-1}{c}+\frac{t}{c}\right) + \left(1-p-\frac{1}{2} \cdot \left(1-(1-2p)^c\right)\right) \cdot (1+\sqrt{\delta_1})^2}{p \cdot u^2 + p \cdot t + (1-2p)}.\]
    This function clearly decreases as $u$ increases (since the numerator and denominator increase at the same rate). In addition, since $\frac{1}{2} \cdot (1-(1-2p)^c) \cdot \frac{1}{c} < p$ (whenever $0 < p < \frac{1}{2}$), we have that the numerator increases at a slower rate than the denominator when $t$ increases, so this function also decreases as $t$ increases. So, we may assume that $t = u = 0$ to get that the fraction is at most
\begin{equation} \tag{2.b} \label{eq:2.b}
    \frac{\frac{1}{2} \cdot (1-(1-2p)^c) \cdot \frac{c-1}{c} + (1-p-\frac{1}{2} \cdot (1-(1-2p)^c)) \cdot (1+\sqrt{\delta_1})^2}{1-2p}.
\end{equation}

    \item $\boldsymbol{c_1, c_2 \ge 1}$. In this case, we have that there exists a point $i_3^{(k)} \in N(j) \cap I_3$ not in $q^{-1}(i_2)$, so it has distance at least $\sqrt{\delta_2}$ from $i_2$. Therefore, by the triangle inequality, $d(i_2, j) \ge \sqrt{\delta_2}-1.$ Let $t = d(i_2, j)$. Next, since all of the points in $N(j) \cap I_3$ have pairwise distance at least $\sqrt{2}$ from each other, $\sum_{i \in N(j) \cap I_3}(1-c(j, i)) \le 1$. Therefore, the denominator of the fraction is at most $1-p(1-t^2)-p \le 1-2p+t^2 p$. 
    
    We now bound the numerator. First, by Proposition \ref{prop:anticor_2} and  since $c \ge 2$, the probability that no point in $N(j) \cap (I_2 \cup I_3)$ is in $S$ is some $p'$, where $p' \le \frac{1}{2} (1-2p)+\frac{1}{2} (1-2p)^c \le (1-2p)(1-p)$. In this event, we have that $c(j, S) \le (1+\sqrt{\delta_1})^2$. In addition, there is a $p$ probability that $i_2 \in S$, in which case $c(j, S) \le t^2$. Finally, with $1-p-p'$ probability, we have that $i_2 \not\in S$ but there is some $i_3^{(k)} \in S$, so $c(j, S) \le 1$. Overall, the fraction is at most
\[\frac{p \cdot t^2 + p' \cdot (1+\sqrt{\delta_1})^2 + (1-p-p')}{1-2p+t^2 p}.\]
    This function clearly decreases as $t$ increases (since the numerator and denominator increase at the same rate), so we may assume that $t = \sqrt{2}-1$ as this is our lower bound on $t$. So, the fraction is at most
\begin{equation} \tag{2.c} \label{eq:2.c}
    \frac{p \cdot (\sqrt{2}-1)^2 + p' \cdot (1+\sqrt{\delta_1})^2 + (1-p-p')}{1-2p+(\sqrt{2}-1)^2 p} \le \frac{1-p(2\sqrt{2}-2) + (1-2p)(1-p) \cdot \left((1+\sqrt{\delta_1})^{2}-1\right)}{1-p(2\sqrt{2}-1)}.
\end{equation}

    \item $\boldsymbol{c_1 = 1}$ \textbf{and} $\boldsymbol{c_2 = 0}$. In this case, $c = 1$, so we simply write $i_3$ as the unique point in $I_3 \cap N(j)$. Let $i^* = w(j)$ be the witness of $j$. Since $c_1 = 1$, this means that $i_2, i_3$ are neighbors in $H(\delta_2)$ and $q(i_3) = i_2$. Finally, we have that $i_2, i_3$ are not connected in $H(\delta_3)$. So, $d(i_2, i_3) \ge \sqrt{\delta_3 \cdot \min(t_{i_2}, t_{i_3})}.$
    Now, note that since $i_2, i_3$ are not in $I_1$, we either have that the witness $i^*$ is in $I_1$, in which case $d(j, I_1) = 1$, or all of $i_2, i_3, i^*$ are adjacent to $I_1$ in $H(\delta_1)$ since $I_1$ was a maximal independent set in $H(\delta_1)$. Therefore, if we define $\beta = d(j, i_2)$ and $\gamma = d(j, i_3)$, this means that $d(j, I_1) \le \min\left(1+\sqrt{\delta_1}, \beta+\sqrt{\delta_1 \cdot t_{i_2}}, \gamma+\sqrt{\delta_1 \cdot t_{i_3}}\right),$ and $d(i_2, i_3) \le \beta+\gamma$ by triangle inequality.
    
    Note that the denominator equals $1-p(1-\beta^2)-p(1-\gamma^2)$. To bound the numerator, note that with probability $p$, $i_2 \in S$ in which case $d(j, S) \le \beta$ and with probability $p$, $i_3 \in S$ in which case $d(j, S) \le \gamma$. Also, these two events are disjoint since $q(i_3) = i_2$. Finally, in the remaining $1-2p$ probability event, $d(j, S) \le d(j, I_1) \le \min(1+\sqrt{\delta_1}, \beta+\sqrt{\delta_1 \cdot t_{i_2}}, \gamma+\sqrt{\delta_1 \cdot t_{i_3}})$.
    
    Letting $t = \min(t_{i_2}, t_{i_3}) \ge 1$, we have that $\min\left(\beta+\sqrt{\delta_1 \cdot t_{i_2}}, \gamma + \sqrt{\delta_1 \cdot t_{i_3}}\right) \le \max(\beta, \gamma) + \sqrt{\delta_1 \cdot t}$. We also know that
    $\sqrt{\delta_3 \cdot t} \le d(i_2, i_3) \le \beta+\gamma$. Therefore, we can bound the ratio by
\begin{align}
    &\hspace{0.5cm}\max_{\substack{t \ge 1 \\ \beta+\gamma \ge \sqrt{\delta_3 \cdot t}}}\frac{(1-2p) \cdot \min(1+\sqrt{\delta_1}, \max(\beta, \gamma)+\sqrt{\delta_1 \cdot t})^2 + p \cdot \beta^2 + p \cdot \gamma^2}{1-p (1-\beta^2) - p (1-\gamma^2)} \nonumber \\
    &= \frac{\left(1-2p\right)\cdot\left(1+\sqrt{\delta_1}\right)^{2}\cdot\left(\delta_1+\left(\sqrt{\delta_1}+\sqrt{\delta_3}\right)^{2}\right)+p\cdot\left(1+\sqrt{\delta_1}\right)^{2}\cdot \delta_3}{\left(1-2p\right)\cdot\left(\delta_1+\left(\sqrt{\delta_1}+\sqrt{\delta_3}\right)^{2}\right)+p\cdot\left(1+\sqrt{\delta_1}\right)^{2}\cdot \delta_3}. \tag{2.d} \label{eq:2.d}
\end{align}
    We prove the final equality in Appendix \ref{app:bash}.

\end{enumerate}

\paragraph{Case 3: $\boldsymbol{a = 0, b \ge 2}$.}

We split this case into three cases. First, we recall that each point $i \in I_3$ corresponds to some point $q(i) \in I_2$. Let $c_1$ represent the number of points such $i \in N(j) \cap I_3$ such that $q(i) \in N(j) \cap I_2$, and let $c_2 = c-c_1$. Note that if $c = 0$, then $c_1 = c_2 = 0$.

\begin{enumerate}[label=\alph*)]
    \item $\boldsymbol{c_1 = 0}$. In this case, all of the points in $N(j) \cap (I_2 \cup I_3)$ must not be connected in $H(\delta_2)$. Therefore, $\sum_{i \in N(j) \cap (I_2 \cup I_3)} (1-c(j, i)) \le 1$ and since $a = 0$, this means that the denominator is at least $1-p$ in expectation. In addition, since $b \ge 2,$ the probability of no point in $N(j) \cap I_2$ being in $S$ is at most $(1-p)^2$, in which case $d(j, S) \le (1+\sqrt{\delta_1})^2$. If there is some point in $N(j) \cap I_2$ in $S$, then $d(j, S) \le 1.$ Therefore, the numerator is at most $(1-p)^2 \cdot (1+\sqrt{\delta_1})^2 + (1-(1-p)^2).$ So, we can bound the fraction by
\begin{equation} \tag{3.a} \label{eq:3.a}
    \frac{(1-p)^2 \cdot (1+\sqrt{\delta_1})^2 + (1-(1-p)^2)}{1-p}.
\end{equation}
    
    \item $\boldsymbol{c_1 = 1, c_2 = 0}$. In this case, the probability of there being a point in $N(j) \cap (I_2 \cup I_3)$ that is part of $S$ is $(1-2p) \cdot (1-p)^{b-1}$. Conditioned on this event, $c(j, S) \le 1$, and otherwise, $c(j, S) \le (1+\sqrt{\delta_1})^2$. So, the numerator is at most $(1-2p) \cdot (1-p)^{b-1} \cdot (1+\sqrt{\delta_1})^2 + \left(1-(1-2p) \cdot (1-p)^{b-1}\right) \cdot 1.$ Finally, we have that all of the points in $N(j) \cap (I_2 \cup I_3)$ are separated by at least $\sqrt{2}$, except for the unique point $i_3 \in N(j) \cap I_3$ and $q(i_3)$, which are separated by at least $\sqrt{\delta_3}$. So, $\sum_{N(j) \cap (I_2 \cup I_3)} c(j, i) \ge \frac{1}{b+1} \cdot \left(\frac{b(b+1)}{2} \cdot 2 + (\delta_3-2)\right) = b + \frac{\delta_3-2}{b+1},$ which means that $\sum_{N(j) \cap (I_2 \cup I_3)} (1-c(j, i)) \le 1 + \frac{2-\delta_3}{b+1}.$ Therefore, the denominator is at least $1-p \cdot \left(1 + \frac{2-\delta_3}{b+1}\right).$ Overall, the fraction is at most
\begin{equation} \tag{3.b} \label{eq:3.b}
    \frac{(1-p)^{b-1}\cdot(1-2p)\cdot(1+\sqrt{\delta_1})^2+\left(1-(1-p)^{b-1}(1-2p)\right)}{1-\left(1+\frac{2-\delta_3}{b+1}\right) \cdot p}
\end{equation}

    \item $\boldsymbol{c_1 \ge 2}$ \textbf{or} $\boldsymbol{c_1 = 1, c_2 \ge 1}$. In this case, we first note that since all of the points in $N(j) \cap I_2$ have distance at least $\sqrt{2}$ from each other and all of the points in $N(j) \cap I_3$ have distance at least $\sqrt{2}$ from each other, both $\sum_{i \in N(j) \cap I_2} (1-c(j, i))$ and $\sum_{i \in N(j) \cap I_3} (1-c(j, i))$ are at most $1$. Let $t$ be such that $1-t = \sum_{i \in N(j) \cap I_2} (1-c(j, i)).$ Then, the denominator is at least $1-p(2-t)$. In addition, with probability $1-(1-p)^b$, at least one of the points in $N(j) \cap I_2$ will be in $S$, conditioned on which the expected value of $c(j, S)$ is at most $\frac{1}{b} \cdot \sum_{i \in N(j) \cap I_2} c(j, i) = \frac{b-1}{b} + \frac{t}{b}.$ Next, note that the probability of no point in $N(j) \cap (I_2 \cup I_3)$ being in $S$ is maximized when all points $i \in N(j) \cap I_3$ with $q(i) \in N(j)$ map to a single point $i_2 \in N(j) \cap I_2$, and all other points in $N(j) \cap I_3$ map to a single point $i_2' \in I_2 \backslash N(j)$. In this case, the probability that no point in $N(j) \cap (I_2 \cup I_3)$ is in $S$ is at most $p' := (1-p)^{b-1} \cdot \left(\frac{1}{2}(1-2p)+\frac{1}{2}(1-2p)^{c_1}\right) \cdot \left(\frac{1}{2} + \frac{1}{2}(1-2p)^{c_2}\right).$
    
    Overall, we have that with probability $1-(1-p)^b$, some point in $I_2 \cap N(j)$ is in $S$, conditioned on which $\BE[c(j, S)] \le \frac{b-1}{b} + \frac{t}{b}$, with probability at most $p',$ no point in $N(j) \cap (I_2 \cup I_3)$ is in $S$, conditioned on which $c(j, S) \le (1+\sqrt{\delta_1})^2,$ and otherwise, some point in $N(j) \cap I_3$ is in $S$, which means $c(j, S) \le 1.$ So, we can bound this fraction overall by
\[\frac{(1-(1-p)^b) \cdot \left(\frac{b-1}{b}+\frac{t}{b}\right)+p' \cdot (\sqrt{\delta_1})^2 + ((1-p)^b-p') \cdot 1}{1-2p+p \cdot t}.\]
    Noting that $(1-(1-p)^b) \cdot \frac{1}{b} \le p$, we have that the numerator increases at a slower rate than the denominator as $t$ increases. Therefore, this fraction is maximized when $t = 0$. So, we can bound the fraction by
\begin{align}
    &\hspace{0.5cm} \frac{(1-(1-p)^b) \cdot \frac{b-1}{b} + p' \cdot (1+\sqrt{\delta_1})^2 + ((1-p)^b-p')}{1-2p} \nonumber \\
    \tag{3.c} \label{eq:3.c}
    &= \frac{\frac{b-1}{b}+\frac{(1-p)^b}{b}+(1-p)^{b-1} \cdot \left(\frac{1}{2}(1-2p)+\frac{1}{2}(1-2p)^{c_1}\right) \cdot \left(\frac{1}{2} + \frac{1}{2}(1-2p)^{c_2}\right) \cdot \left((1+\sqrt{\delta_1})^2-1\right)}{1-2p}.
\end{align}
\end{enumerate}

\paragraph{Case 4: $\boldsymbol{a=0, b=0}$.}

We split this case into three subcases.

\begin{enumerate}[label=\alph*)]
    \item $\boldsymbol{c = 0}.$ In this case, $N(j) \cap S$ is always empty, so the denominator is $1$. To bound the numerator, we consider the witness $i^*$ of $j$. If $i^* \in I_1$, the numerator is at most $c(j, i^*) \le 1$ so we can bound the fraction by $1$. Else, if $i^* \not\in V_2,$ then there exists $i_1 \in I_1$ such that $d(i^*, i_1) \le \sqrt{\delta_2 \cdot \min(t_{i^*}, t_{i_2})} \le \sqrt{\delta_2}$, and since $d(j, i^*) \le 1,$ we have that $d(j, I_1) \le 1+\sqrt{\delta_2})$. Thus, the fraction in the case where $i \in I_1$ or $i^* \not\in V_2$ is at most
\begin{equation} \tag{4.a.i} \label{eq:4.a.i}
    (1+\sqrt{\delta_2})^2.
\end{equation}
    
    Otherwise, there is some $i_1 \in I_1$ of distance at most $\sqrt{\delta_1 \cdot \min(t_{i^*}, t_{i_1})} \le \sqrt{\delta_1}$ away from $i^*$. Next, if $i^* \in I_2,$ the numerator is at most $p \cdot 1 + (1-p) \cdot (1+\sqrt{\delta_1})^2$. Otherwise, there is some $i_2 \in I_2$ of distance at most $\sqrt{\delta_2 \cdot \min(t_{i^*}, t_{i_2})}$ away from $i^*$. Finally, $d(i_1, i_2) \ge \sqrt{\delta_2 \cdot \min(t_{i_1}, t_{i_2})}$. Therefore, as in \ref{eq:1.b}, we can apply Proposition \ref{prop:triangle} to obtain that the numerator, and therefore, the fraction is at most
\begin{equation} \tag{4.a.ii} \label{eq:4.a.ii}
    1+p \cdot \delta_2 + (1-p) \cdot \delta_1 + 2\sqrt{p^2 \cdot \delta_2 + (1-p) \cdot \delta_1},
\end{equation}
    since the above \eqref{eq:4.a.ii} is greater than both $1$ and $p + (1-p) \cdot (1+\sqrt{\delta_1})^2$ for any $0 < p < 1.$
    
    \item $\boldsymbol{c=1}.$ In this case, let $i_3$ be the unique element in $N(j) \cap I_3$. Conditioned in $i_3 \in S$, the denominator equals $c(j, i_3)$ and the numerator is at most $c(j, i_3)$. Otherwise, the denominator is $1$, and the numerator can again be bounded in an identical way to the previous case, since the probability of $i_2 \in S$ is either $p$ (if $q(i_3) \neq i_2$) or $\frac{p}{1-p} > p$ (if $q(i_3) = i_2$). Therefore, the fraction is again at most 
\begin{equation} \tag{4.b.i} \label{eq:4.b.i}
    (1+\sqrt{\delta_2})^2
\end{equation}
    if the witness $i^*$ of $j$ satisfies $i^* \in I_1$ or $i^* \not\in V_2$, and is at most
\begin{equation} \tag{4.b.ii} \label{eq:4.b.ii}
    1+p \cdot \delta_2 + (1-p) \cdot \delta_1 + 2\sqrt{p^2 \cdot \delta_2 + (1-p) \cdot \delta_1}
\end{equation}
    otherwise.

    \item $\boldsymbol{c \ge 2}$. In this case, note that all points in $N(j) \cap I_3$ are separated by at least $\sqrt{2}$, which means that $\sum_{i \in N(j) \cap I_3} (1-c(j, i)) \le 1$. Letting $t$ be such that $1-t = \sum_{i \in N(j) \cap I_3} (1-c(j, i)),$ we have that $\sum_{i \in N(j) \cap I_3} c(j, i) = c-1+t$, so there exists $i_3 \in N(j) \cap I_3$ such that $c(j, i_3) \le \frac{c-1+t}{c}$. In addition, we know that $c(j, I_1) \le (1+\sqrt{\delta_1})^2$. Finally, we also note the denominator equals $1-p(1-t) = 1-p+pt$.
    
    Next, note that
\begin{align*}
    \sum_{i, i' \in N(j) \cap I_3} c(i, i') &= \sum_{i, i' \in N(j) \cap I_3} \left[\|i-j\|^2 + \|i'-j\|^2 - 2 \langle i-j, i'-j \rangle\right] \\
    &= 2 c \cdot \left[\sum_{i \in N(j) \cap I_3} c(j, i)\right] - 2 \left\langle \sum_{i \in N(j) \cap I_3} (i-j), \sum_{i' \in N(j) \cap I_3} (i'-j)\right\rangle \\
    &\le 2 c \cdot \left[\sum_{i \in N(j) \cap I_3} c(j, i)\right] \\
    &= 2c \cdot (c-1+t).
\end{align*}
    Since $c(i, i) = 0$, this means there exists $i \neq i'$ such that $c(i, i') \le \frac{2(c-1+t)}{c-1},$ and since $c(i, i') \ge 2 \cdot \min(t_i, t_{i'})$ for any $i, i' \in I_3$, this means that $\min_{i \in N(j) \cap I_3} t_i \le \frac{c-1+t}{c-1}.$
    
    Let $i = \arg\min_{i \in N(j) \cap I_3} t_i$, and let $i_2 = q(i) \in I_2$. Let $\mathcal{E}_1$ be the event that no point in $N(j) \cap I_3$ nor $i_2$ is in $S$, let $\mathcal{E}_2$ be the event that no point in $N(j) \cap I_3$ is in $S$, and let $\mathcal{E}_3$ be the event that $i_3$ is not in $S$.
    Note that $\mathcal{E}_1$ implies $\mathcal{E}_2$, which implies $\mathcal{E}_3$.
    Now, by Proposition \ref{prop:anticor_2}, $\BP(\mathcal{E}_1)$ equals some $p_1 \le \frac{1}{2} (1-2p) + \frac{1}{2} (1-2p)^c$.
    Likewise, By Proposition \ref{prop:anticor_1}, $\BP(\mathcal{E}_2)$ equals some $p_2 \le \frac{1}{2} + \frac{1}{2} (1-2p)^c$. 
    In addition, $\BP(\mathcal{E}_3) = p_3 = 1-p$.
    Under the event $\mathcal{E}_3^c$, we have that $c(j, S) \le c(j, i_3) \le \frac{c-1+t}{c}$. Next, under the event $\mathcal{E}_3 \backslash \mathcal{E}_2,$ we have that some point in $N(j) \cap I_3$ is in $S$, so $c(j, S) \le 1$. Under the event $\mathcal{E}_2 \backslash \mathcal{E}_1$, we know that $i_2 \in S$, so $c(j, S) \le d(j, i_2)^2 \le (c(j, i) + \sqrt{\delta_2 \cdot t_i})^2 \le \left(1 + \sqrt{2 \cdot \frac{c-1+t}{c-1}}\right)^2$. Finally, we always have that $c(j, S) \le (1+\sqrt{\delta_1})^2$.

    Therefore, we can bound the overall fraction by
\[\frac{p_1 \cdot (1+\sqrt{\delta_1})^2 + (p_2-p_1) \cdot \min\left(1+\sqrt{\delta_1}, 1+\sqrt{2 \cdot \frac{c-1+t}{c-1}}\right)^2 + (p_3-p_2) \cdot 1 + (1-p_3) \cdot \frac{c-1+t}{c}}{1-p+pt}.\]
    Since $(1+\sqrt{\delta_1})^2 \ge \min\left(1+\sqrt{\delta_1}, 1+\sqrt{2 \cdot \frac{c-1+t}{c-1}}\right)^2 \ge 1 \ge \frac{c-1+t}{c},$ the above fraction is an increasing function in the variables $p_1, p_2, p_3$.
    So, we can upper bound this fraction by replacing $p_1, p_2, p_3$ with their respective upper bounds $\frac{1}{2} (1-2p) + \frac{1}{2} (1-2p)^c,$ $\frac{1}{2} + \frac{1}{2} (1-2p)^c$, and $1-p$, as well as replacing $\min\left(1+\sqrt{\delta_1}, 1+\sqrt{2 \cdot \frac{c-1+t}{c-1}}\right)$ with simply $1+\sqrt{2 \cdot \frac{c-1+t}{c-1}}$.
    %
    Next, since $p_2-p_1 = p$ and $1-p_3 = p$, the derivative of the numerator with respect to $t$ is $p \cdot \left(\frac{2}{c-1} + \sqrt{\frac{2}{(c-1)(c-1+t)}}\right) + p \cdot \frac{1}{c} \le p \cdot (2.5 + \sqrt{2})$, and the derivative of the denominator with respect to $t$ is $p t$. 
    Hence, this fraction decreases as $t$ increases, unless the fraction is less than $2.5+\sqrt{2}$.
    Therefore, we can bound the fraction by
\begin{align} 
 \nonumber
    &\hspace{0.5cm}
    \max\left(2.5+\sqrt{2}, \frac{p_1 \cdot (1+\sqrt{\delta_1})^2 + (p_2-p_1) \cdot (1+\sqrt{2})^2 + (p_3-p_2) \cdot 1 + (1-p_3) \cdot \frac{c-1}{c}}{1-p}\right) \\
 \nonumber
    &=\max\left(2.5+\sqrt{2}, \frac{p_1 \cdot (1+\sqrt{\delta_1})^2 + p \cdot (1+\sqrt{2})^2 + (1-p-p_2) \cdot 1 + p \cdot \frac{c-1}{c}}{1-p}\right) \\
\nonumber
    &= \max\left(2.5+\sqrt{2}, \frac{\left(\frac{1}{2}(1-2p)+\frac{1}{2}(1-2p)^c\right) \cdot (1+\sqrt{\delta_1})^2 - \left(\frac{1}{2} + \frac{1}{2}(1-2p)^c\right) +p(1+\sqrt{2})^2+1-\frac{p}{c}}{1-p}\right), \tag{4.c} \label{eq:4.c}
\end{align}
    where we used the facts that $p_3 = 1-p$ and $p_2-p_1 = p$.
\end{enumerate}

\paragraph{Case 5: $\boldsymbol{a \ge 1}$.}
In this case, note that $\BE[c(j, S)] \le \BE[c(j, I_1)] \le 1$ deterministically, since $a = |I_1 \cap N(j)| \ge 1$. However, we can improve upon this, since we may have some points in $I_1 \cap N(j)$ much closer to $j$, or we may have some points in $(I_2 \cup I_3) \cap N(j)$ which are closer and appear with some probability.

Recall that in our algorithm, we flip a coin for each $i \in I_2$ to decide whether we include $i \in S$ with probability $2p$ or include each $q^{-1}(i)$ in $S$ independently with probability $2p$. Let us condition on all of these fair coin flips, and say that a point $i \in I_2 \cup I_3$ \emph{survives} the coin flips if they could be in the set $S$ with probability $2p$ afterwards.
For simplicity, we replace $p$ with $\bar{p} := 2p$. We also let $\bar{I}$ represent the points in $I_2 \cup I_3$ that survive the fair coin flips.

Let the squared distances from $j$ to each of the points in $I_1 \cap N(j)$ be $r_1, \dots, r_a$, and the squared distances from $j$ to each of the points in $\bar{I} \cap N(j)$ be $s_1, \dots, s_h$, where $h = |\bar{I} \cap N(j)|$.
It is trivial to see that $\BE[c(j, S)] \le \min_{1 \le i \le a} r_i \le \frac{r_1+\cdots+r_a}{a},$ and conditioned on at least one of the points in $\bar{I} \cap N(j)$ being selected, we have that $c(j, S)$ in expectation is at most $\frac{s_1 + \cdots + s_h}{h}.$ The probability of at least one of the points in $\bar{I} \cap N(j)$ being selected in $S$ is 
$$1-(1-\bar{p})^h \ge 1-\frac{1}{(1+\bar{p})^h} \ge 1-\frac{1}{1+\bar{p} \cdot h} = \frac{\bar{p} \cdot h}{1+\bar{p} \cdot h} \ge \frac{\bar{p} \cdot h}{a+\bar{p} \cdot h},$$
since conditioned on the initial coin flips, each surviving point in $(\bar{I_2} \cup \bar{I_3})$ is included in $S$ independently with probability $\bar{p}$. Therefore, we can say that
\[\BE[c(j, S)] \le \frac{r_1+\cdots+r_a}{a} \cdot \frac{a}{a+\bar{p} \cdot h} + \frac{s_1+\cdots+s_h}{h} \cdot \frac{\bar{p} \cdot h}{a+\bar{p} \cdot h} = \frac{(r_1+\cdots+r_a)+\bar{p} \cdot (s_1+\cdots+s_h)}{a+\bar{p} \cdot h}.\]

Next, we have that
\begin{align*}
    \BE\left[\alpha_j-\sum_{i \in N(j) \cap S} (\alpha_j-c(j, i))\right] &= \alpha_j - (a \cdot \alpha_j -(r_1+\cdots+r_a))-\bar{p} \cdot (h \cdot \alpha_j-(s_1+\cdots+s_h))\\
    &= \alpha_j \cdot (1-(a+\bar{p} \cdot h)) + [(r_1+\cdots+r_a)+\bar{p} \cdot (s_1+\cdots+s_h)].
\end{align*}

Now, we provide a lower bound for $(r_1+\cdots+r_a)+\bar{p}(s_1+\cdots+s_h).$ To do so, we 
use the fact that all the points in $N(j) \cap I_1$ are separated by at least $\delta_1 \cdot \alpha_j$ in squared distance, and all the surviving points in $N(j) \cap (I_1 \cup \bar{I_2} \cup \bar{I_3})$ are separated by at least $\delta_2 \cdot \alpha_j$ in squared distance, to get 
\[(r_1+\cdots+r_a)+\bar{p}(s_1+\cdots+s_h) \ge \frac{1}{a+\bar{p} \cdot h} \cdot \alpha_j \cdot \left( \delta_1 \cdot \frac{a(a-1)}{2} + \delta_2 \cdot \bar{p} \cdot a \cdot h + \delta_2 \cdot \bar{p}^2 \cdot \frac{h(h-1)}{2}\right).\]

So, if $a \ge 1$ and $(a, h) \neq (1, 0)$, then if we let $T_1 = a+\bar{p} \cdot h,$ $T_3 = \delta_1 \cdot \frac{a(a-1)}{2} + \delta_2 \cdot \bar{p} \cdot a \cdot h + \delta_2 \cdot \bar{p}^2 \cdot \frac{h(h-1)}{2},$ and $T_2 = T_3/T_1$, then the ratio is at most
\begin{equation} \tag{5.a} \label{eq:5.a}
\frac{T_2/T_1}{1-T_1+T_2} = \frac{T_3}{T_1(T_1-T_1^2+T_3)} = \frac{1}{T_1} + \frac{T_1-1}{T_1-T_1^2+T_3}.
\end{equation}

In the case where $a = 1, h = 0,$ this fraction is undefined. However, we note that in this case, $N(j) \cap S$ deterministically contains a unique center $i^*$ and nothing else, so $\BE[c(j, S)] \le c(j, i^*)$ and $\BE\left[\alpha_j - \sum_{i \in N(j) \cap S} (\alpha_j-c(j, i))\right] = \alpha_j - (\alpha_j - c(j, i^*)) = c(j, i^*).$ Therefore, the fraction is $1$.
\end{proof}

Therefore, we have the LMP approximation is at most $\rho(p)$, where $\rho(p)$ is determined via the numerous cases in the proof of Lemma \ref{lem:main_lmp}. The final step is to actually bound $\rho(p)$ based on the cases. Indeed, by casework one can show the following proposition.

\begin{proposition} \label{prop:more_bash}
    For $p \in [0.096, 0.402]$ and $\delta_1 = \frac{4+8\sqrt{2}}{7}, \delta_2 = 2$, and $\delta_3 = 0.265$, we have that 
\begin{multline*}
    \rho(p) \le \max\bigg(3+2\sqrt{2},
    1 + 2p + (1-p) \cdot \delta_1 + 2 \sqrt{2 p^2 + (1-p) \cdot \delta_1},
    \frac{(1-p)\cdot(1+\sqrt{\delta_1})^2+p\cdot(1-\sqrt{\delta_3})^2}{1-p+p (1-\sqrt{\delta_3})^2}, \\ 
    \frac{(1-2p)\cdot(1+\sqrt{\delta_1})^{2}\cdot(\delta_1+(\sqrt{\delta_1}+\sqrt{\delta_3})^{2})+p\cdot\left(1+\sqrt{\delta_1}\right)^2\cdot \delta_3}{(1-2p)\cdot(\delta_1+(\sqrt{\delta_1}+\sqrt{\delta_3})^{2})+p\cdot(1+\sqrt{\delta_1})^2\cdot \delta_3}\bigg).
\end{multline*}

As a consequence, for $p_1 := 0.402$, $\rho(p_1) \le 3+2\sqrt{2}$.
\end{proposition}

We defer the casework to Lemma \ref{lem:more_bash}, and the above proposition follows immediately from it.
Therefore, 
we get a $(3+2\sqrt{2}) \approx 5.828$ LMP approximation for the Euclidean $k$-means problem.

\section{Polynomial-time Approximation Algorithm for Euclidean $k$-means} \label{sec:poly_time_alg}

In this section, we describe how we improve the LMP approximation for Euclidean $k$-means to a polynomial-time approximation algorithm. Unfortunately, we will lose a slight factor in our approximation, but we still obtain a significant improvement over the previous state-of-the-art approximation factor.
While we focus on the $k$-means problem, we note that this improvement can also be applied to the $k$-median problem as well, with some small modifications that we will make note of. In Section \ref{sec:k_median}, we will provide an LMP approximation for $k$-median, and explain how we can use the same techniques in this section to also obtain an improved polynomial-time algorithm for $k$-median as well.

In Subsection \ref{subsec:k_means_alg_prelim}, we describe the polynomial-time algorithm to generate two nested quasi-independent sets $I$ and $I'$, which will be crucial in developing our final set of centers of size $k$ with low clustering cost. This procedure is based on a similar algorithm of Ahmadian et al.~\cite{ahmadian2017better}, but with some important changes in how we update our graphs and independent sets. In Subsection \ref{subsec:more_prelims}, we describe and state a few additional preliminary results. In Subsection \ref{subsec:k_means_analysis}, we analyze the algorithm and show how we can use $I$ and $I'$ to generate our final set of centers, to obtain a $6.013$-approximation algorithm. Finally, in Subsection \ref{subsec:k_means_improved}, we show that our analysis in Subsection \ref{subsec:k_means_analysis} can be further improved, to obtain a $(\kmeansratio+\eps)$-approximation guarantee.

We remark that our approximation guarantee of $(\kmeansratio+\eps)$ is only \emph{in expectation}. However, this can be made to be with exponential failure probability, as with probability $\eps$ the approximation ratio will be $(\kmeansratio+O(\eps))$. So, by running the algorithm $\eps^{-1} \cdot \text{poly}(n)$ times in parallel, and outputting the best solution of these, we obtain a $(\kmeansratio+\eps)$-approximation ratio with probability at least $1-(1-\eps)^{\eps^{-1} \cdot \text{poly}(n)} \ge 1-e^{-\text{poly}(n)}$.

\subsection{The algorithm and setup} \label{subsec:k_means_alg_prelim}

First, we make some assumptions on the clients and facilities. We first assume that the number of facilities, $m = |\mathcal{F}|$, is at most polynomial in the number of clients $n = |\mathcal{D}|$. In addition, we also assume that the distances between clients and facilities are all in the range $[1, n^6]$. Indeed, both of these assumptions can be made via standard discretization techniques, and we only lose an $1+o(1)$-approximation factor by removing these assumptions~\cite{ahmadian2017better}. Note that this means the optimal clustering cost, which we call $\text{OPT}_k$, is at least $n$ for both $k$-means and $k$-medians.
Finally, we assume that $k \le n-1$ (else this problem is trivial in polynomial time).

Next, we describe the setup relating to dual solutions.
Consider the tuple $(\alpha, z, \mathcal{F}_S, \mathcal{D}_S)$, where $\alpha \in \BR^{\mathcal{D}}$, $z \in \BR^{\mathcal{F}}$, $\mathcal{F}_S \subset \mathcal{F}$, and $\mathcal{D}_S: \mathcal{F}_S \to \{0, 1\}^{\mathcal{D}}$. Here, $\alpha$ represents the set $\{\alpha_j\}_{j \in \mathcal{D}}$ which will be a solution to the dual linear program, $z$ represents $\{z_i\}_{i \in \mathcal{F}}$, where each $z_i \in \{\lambda, \lambda+\frac{1}{n}\}$ will be a modified value representing the threshold for tightness of facility $i$, $\mathcal{F}_S$ represents a subset of facilities that we deem ``special'', and $\mathcal{D}_S$ is a function that maps each special facility to a subset of the clients that we deem special clients for that facility.

When talking about a single solution $(\alpha, z, \mathcal{F}_S, \mathcal{D}_S)$, we define $\beta_{ij} = \max(0, \alpha_j-c(j, i))$ for any $i \in \mathcal{F}, j \in \mathcal{D}$, and define $N(i) = \{j \in \mathcal{D}: \beta_{ij} > 0\}$. We say that a facility $i$ is tight if $\sum_{j \in \mathcal{D}} \beta_{ij} = z_i$. Now, we define $\tau_i$ for each $i$ that is either tight or special (i.e., in $\mathcal{F}_S$). For each tight facility, we define $\tau_i = \max_{j \in N(i)} \alpha_j$, and for each special facility, we define $\tau_i = \max_{j \in N(i) \cap \mathcal{D}_S(i)} \alpha_j$.  We default the maximum of an empty set to be $0$. We also consider a modified conflict graph $H := H(\delta)$ on the set of tight or special facilities, with an edge between $i$ and $i'$ if $c(i, i') \le \delta \cdot \min(\tau_i, \tau_{i'})$.

\medskip
We can now define the notion of \emph{roundable} solutions: our definition is slightly modified from \cite[Definition 5.1]{ahmadian2017better}.

\begin{defn} \label{def:roundable}
    Let $\alpha \in \BR^\mathcal{D},$ $z \in \BR^\mathcal{F},$ $\mathcal{F}_S \subset \mathcal{F}$ be the set of special facilities, and $\mathcal{D}_S: \mathcal{F}_S \to \{0, 1\}^{\mathcal{D}}$ be the function assigning each special facility $i \in \mathcal{F}_S$ to a subset of special clients $\mathcal{D}_S(i)$. Then, the tuple $(\alpha, z, \mathcal{F}_S, \mathcal{D}_S)$ is \emph{$(\lambda, k')$-roundable} if 
\begin{enumerate}
    \item $\alpha$ is a feasible solution of $\text{DUAL}(\lambda + \frac{1}{n})$ and $\alpha_j \ge 1$ for all $j$.
    \item For all $i \in \mathcal{F},$ $\lambda \le z_i \le \lambda + \frac{1}{n}.$
    \item There exists a subset $\mathcal{D}_B$ of ``bad'' clients so that for all $j \in \mathcal{D}$, there is a facility $w(j)$ that is either tight or in $\mathcal{F}_S$, such that:
    \begin{enumerate}[label=(\alph*)]
        \item For all $j \in \mathcal{D} \backslash \mathcal{D}_B$, $(1+\eps) \cdot \alpha_j \ge c(j, w(j))$
        \item For all $j \in \mathcal{D} \backslash \mathcal{D}_B$, $(1+\eps) \cdot \alpha_j \ge \tau_{w(j)}$
        \item $\gamma \cdot \text{OPT}_{k'} \ge \sum_{j \in \mathcal{D}_B} \left(c(j, w(j)) + \tau_{w(j)}\right)$
    \end{enumerate}
    \item $\sum_{i \in \mathcal{F}_S} \sum_{j \in \mathcal{D}_S(i)} \max(0, \alpha_{j}-c(j, i)) \ge \lambda \cdot |\mathcal{F}_S| - \gamma \cdot \text{OPT}_{k'}$, and $|\mathcal{F}_S| \le n.$ \label{cond:4}
\end{enumerate}
Here, $\gamma \ll \eps \ll 1$ are arbitrarily small constants, which are implicit parameters in the definition. Finally, we say that $(\alpha, z, \mathcal{F}_S, \mathcal{D}_S)$ is $k'$-roundable if it is $(\lambda, k')$-roundable for some choice of $\lambda \ge 0$.
\end{defn}

\begin{figure}
\centering
\begin{algorithm}[H]
    \caption{Generate Sequence of Nested Quasi-Independent Sets}\label{alg:main}
    \begin{algorithmic}[1] 
        \State Initialize $\mathcal{S}^{(0)} = (\alpha^{(0)}, z^{(0)}, \mathcal{F}_S^{(0)}, \mathcal{D}_S^{(0)})$, and set $\lambda \leftarrow 0$, $I_1^{(0)} \leftarrow \mathcal{F}$, $I_2^{(0)} \leftarrow \emptyset$, $I_3^{(0)} \leftarrow \emptyset$
        \State Set $I^{(0)} = (I_1^{(0)}, I_2^{(0)}, I_3^{(0)})$, $\eps_z \leftarrow n^{-\text{poly}(\gamma^{-1})}$, $L \leftarrow 4n^7 \cdot \eps_z^{-1}$, $k' \leftarrow \min(k, |I_1^{(0)}|),$ and $p_1 = 0.402$.
        \For{$\lambda = 0, \eps_z, \dots, L \cdot \eps_z$}
            \For{$i \in \mathcal{F}$}
                \State Call \Call{RaisePrice}{$\alpha^{(0)}, z^{(0)}, I_1^{(0)}, i$} to generate a polynomial-size sequence $\mathcal{S}^{(1)}, \dots, \mathcal{S}^{(q)}$ of close, $k'$-roundable solutions
                \For{$\ell = 0$ to $q-1$}
                    \State Call \Call{GraphUpdate}{$\mathcal{S}^{(\ell)}, \mathcal{S}^{(\ell+1)}, I^{(\ell)}$} to produce a sequence $\{I^{(\ell, r)}\}_{r = 0}^{p_\ell}$
                    \For{$r = 1$ to $p_\ell$}
                        \If{$|I_1^{(\ell, r)}|+p_1|I_2^{(\ell, r)} \cup I_3^{(\ell, r)}| < k$}
                            \State Let $I = (I_1^{(\ell, r-1)}, I_2^{(\ell, r-1)}, I_3^{(\ell, r-1)})$, $I' = (I_1^{(\ell, r)}, I_2^{(\ell, r)}, I_3^{(\ell, r)})$, and \textbf{return} $(I, I')$
                        \Else
                            \State $I^{(\ell+1)} \leftarrow I^{(\ell, p_{\ell})}$
                        \EndIf
                    \EndFor
                \EndFor
                \State $\mathcal{S}^{(0)} \leftarrow \mathcal{S}^{(q)}$, $I^{(0)} \leftarrow I^{(q)}$
                \State $k' \leftarrow \min(k', |I_1^{(0)}|)$
            \EndFor
        \EndFor
    \end{algorithmic}
\end{algorithm}
\end{figure}

Our main algorithm is described in Figure~\ref{alg:main}. This algorithm outputs two nested quasi-independent sets $I$ and $I'$. The final set $S$ will be obtained either from one of these two sets, or from some hybridization of them. We will defer the actual construction of $S$ to Theorem \ref{thm:main}.

The method of \Call{RaisePrice}{} comes from \cite{ahmadian2017better}, and will not be of importance, apart from the results that they give us for the overall algorithm (see Theorem \ref{thm:ahmadian_roundable}). We will make some important definitions and set up the overall algorithm, and then describe the method of \Call{GraphUpdate}{}, which is slightly modified from \cite{ahmadian2017better}.

First, we describe the initialization phase of Algorithm \ref{alg:main} to generate $\mathcal{S}^{(0)}$, which is almost identical to that of Ahmadian et al.~\cite[P. 22-23 of journal version]{ahmadian2017better}. The main difference is that we parameterized the procedure by $1+\kappa, 1/\kappa$ (instead of 2 and 6 in  Ahmadian et al.) Start by setting $\lambda = 0$, $z_i^{(0)} = 0$ for all $i \in \mathcal{F}$, and $\mathcal{F}_S^{(0)} = \emptyset$ (so $\mathcal{D}_S^{(0)}$ has empty domain). We then set $\alpha_j^{(0)} = 0$ for all $j \in \mathcal{D}$. 

Now, we increase all of the $\alpha_j^{(0)}$ values simultaneously at a uniform rate, and for each $j$, we stop increasing $\alpha_j^{(0)}$ once one of the following $2$ events occur:
\begin{enumerate}
    \item $\alpha_j^{(0)} = c(j, i)$ for some $i$.
    \item $(1+\kappa) \cdot \sqrt{\alpha_j} \ge d(j, j')+\frac{1}{\kappa} \cdot \sqrt{\alpha_{j'}}$ for some $j' \neq j$ (or $(1+\kappa) \cdot \alpha_j \ge d(j, j')+\frac{1}{\kappa} \cdot \alpha_{j'}$ in the $k$-median case). Here, $\kappa$ will be a fixed small constant (see Appendix \ref{app:why_am_i_doing_this} for details on how to set $\kappa$).
\end{enumerate}

In the initial solution, $\alpha_j^{(0)} \le \min_{i \in \mathcal{F}} c(j, i)$ for all $i \in \mathcal{F}$, which means that $N(i)$ is empty for all $i \in \mathcal{F}$, so $\tau_i = 0$. In addition, since every $z_i = 0$, every facility is tight. This means that the conflict graph $H(\delta)$ on the set of tight facilities for the initial solution we construct is just an empty graph on the full set of facilities, since $c(i, i') > 0 = \delta \cdot \min(\tau_i, \tau_{i'}).$ So, this means that if we apply Algorithm \ref{alg:lmp} to $V_1 = \mathcal{F}$, we will obtain that $I_1^{(0)} = \mathcal{F}$ and $I_2^{(0)} = I_3^{(0)} = \emptyset$.
\medskip

We now set up some definitions that will be important for the remainder of the algorithm and analysis.
Define two dual solutions $\alpha = \{\alpha_j\}$ and $\alpha' = \{\alpha_j'\}$ to be \emph{close} if $\max_{j \in \mathcal{D}} |\alpha_j-\alpha_j'| \le \frac{1}{n^2}$.
Consider two solutions $\mathcal{S}^{(\ell)} = (\alpha^{(\ell)}, z^{(\ell)}, \mathcal{F}_S^{(\ell)}, \mathcal{D}_S^{(\ell)})$ and $\mathcal{S}^{(\ell+1)} = (\alpha^{(\ell+1)}, z^{(\ell+1)}, \mathcal{F}_S^{(\ell+1)}, \mathcal{D}_S^{(\ell+1)})$ that are each 
$(\lambda, k')$-roundable for some choice of $\lambda$, such that $\alpha^{(\ell)}$ and $\alpha^{(\ell+1)}$ are close.
This means that $|\alpha_j^{(\ell)}-\alpha_j^{(\ell+1)}| \le \frac{1}{n^2}$ for all $j \in \mathcal{D},$ and that $\lambda \le z_i^{(\ell)}, z_i^{(\ell+1)} \le \lambda+\frac{1}{n}$ for all $i \in \mathcal{F}$, even if $i$ is not tight. Let $\mathcal{V}^{(\ell)}$ represent the set of tight or special facilities in $\mathcal{S}^{(\ell)}$ and define $\mathcal{V}^{(\ell+1)}$ likewise. 

Let $\sqcup$ denote the disjoint union, i.e., $S \sqcup T$ is a set consisting of a copy of each element in $S$ and a distinct copy of each element in $T$.
For each point $i \in \mathcal{V}^{(\ell)} \sqcup \mathcal{V}^{(\ell+1)}$, we define $\mathcal{D}_S(i)$ (if $i$ were a special facility), $\tau_i$, and $z_i$ based on whether $i$ came from $\mathcal{V}^{(\ell)}$ or from $\mathcal{V}^{(\ell+1)}$.
This means that for $i \in \mathcal{V}^{(\ell)}$, $\mathcal{D}_S(i) = \mathcal{D}^{(\ell)}_S(i)$, $z_i = z_i^{(\ell)}$, and $\tau_i = \tau_i^{(\ell)} = \max_{j: \alpha_j^{(\ell)} > c(j, i)} \alpha_j^{(\ell)}$ if $i$ is tight and $\tau_i = \tau_i^{(\ell)} = \max_{j: \alpha_j^{(\ell)} > c(j, i), j \in \mathcal{D}_S^{(\ell)}(i)} \alpha_j^{(\ell)}$ if $i$ is special (and likewise for $i \in \mathcal{V}^{(\ell+1)}$).
In addition, for each client $j \in \mathcal{D}$, we define $\alpha_j = \min(\alpha_j^{(\ell)}, \alpha_j^{(\ell+1)})$; for each $j \in \mathcal{D}$ and $i \in \mathcal{V}^{(\ell)} \sqcup \mathcal{V}^{(\ell+1)}$, we define $\beta_{ij} = \max(0, \alpha_j-c(j, i))$; and for each $i \in \mathcal{F}$, we let $N(i) = \{j \in \mathcal{D}: \alpha_j-c(j, i) > 0\}$. Since $\alpha_j = \min(\alpha_j^{(\ell)}, \alpha_j^{(\ell+1)}),$ this means $N(i) \subset N^{(\ell)}(i)$ if $i \in \mathcal{V}^{(\ell)}$ and $N(i) \subset N^{(\ell+1)}(i)$ if $i \in \mathcal{V}^{(\ell+1)}.$

We create a \emph{hybrid conflict graph} on a subset of the disjoint union $\mathcal{V}^{(\ell)} \sqcup \mathcal{V}^{(\ell+1)}$. First, we let $H^{(\ell, 0)}(\delta)$ represent the conflict graph on $\mathcal{V}^{(\ell, 0)} := \mathcal{V}^{(\ell)}$. The conflict graph on a set of facilities means there is an edge between two vertices $(i, i')$ if $c(i, i') \le \delta \cdot \min(\tau_i, \tau_{i'})$.
Next, we choose some ordering of the vertices in $\mathcal{V}^{(\ell)}$, and for each $1 \le r \le p_\ell := |\mathcal{V}^{(\ell)}|+1$, we let $V^{(\ell, r)}$ represent the vertices of the so-called \emph{merged vertex set} defined as $\mathcal{V}^{(\ell)} \sqcup \mathcal{V}^{(\ell+1)}$ after we removed the first $r-1$ vertices in $\mathcal{V}^{(\ell)}$, and let $H^{(\ell, r)}(\delta)$ represent the conflict graph on $\mathcal{V}^{(\ell, r)}$, where again the conflict graph means that $i, i' \in \mathcal{V}^{(\ell, r)}$ share an edge if $c(i, i') \le \delta \cdot \min(\tau_i, \tau_{i'})$. 
Note that $\mathcal{V}^{(\ell, 1)} = \mathcal{V}^{(\ell)} \sqcup \mathcal{V}^{(\ell+1)}$.
For simplicity of notation, we may abbreviate $\mathcal{V}^{(\ell, r)}$ as $\mathcal{V}$ and $H^{(\ell, r)}(\delta)$ as $H(\delta)$ if the context is clear.

In the context of a hybrid conflict graph $\mathcal{V}^{(\ell, r)}$, for any client $j$, let $\bar{N}(j)$ be the subset of $\mathcal{V}^{(\ell, r)}$ consisting of all tight facilities $i$ such that $j \in N(i)$ and all special facilities $i$ such that $j \in N(i) \cap \mathcal{D}_S(i)$.
We also define $w(j)$ as the witness for $j$ in the solution $\mathcal{S}^{(\ell+1)}$, and $\mathcal{D}_B$ as the set of bad clients from the solution $\mathcal{S}^{(\ell+1)}$.

Finally, to describe the actual \Call{GraphUpdate}{} procedure, it works as follows. First, we note that $I^{(0, 0)} = I^{(0)}$ which has already been decided, either by the first initialized solution or from the previous solution $I^{(q)}$ before we reset $I^{(0)} = I^{(q)}$. Otherwise, $I^{(\ell+1, 0)} = I^{(\ell, p_\ell)}$ since the set of vertices $\mathcal{V}^{(\ell+1, 0)} = \mathcal{V}^{(\ell+1)} = \mathcal{V}^{(\ell, p_\ell)}$. Finally, we note that $|\mathcal{V}^{(\ell, r)} \backslash \mathcal{V}^{(\ell, r+1)}| \le 1,$ since for $r = 0$ $\mathcal{V}^{(\ell, r)} \subset \mathcal{V}^{(\ell, r+1)}$ and otherwise, $\mathcal{V}^{(\ell, r+1)}$ is created by simply removing one element from $\mathcal{V}^{(\ell, r)}$. So, the maximal independent set $I_1^{(\ell, r)}$ of $H^{(\ell, r)}(\delta_1)$ can easily be extended to a maximal independent set $I_1^{(\ell, r+1)}$ of $H^{(\ell, r+1)}(\delta_1)$ by deleting at most $1$ element and then possibly extending the independent set. We then extend $I_1^{(\ell, r+1)}$ arbitrarily based on Steps \ref{step:V2} to \ref{step:I3} to create $I_2^{(\ell, r+1)}$ and $I_3^{(\ell, r+1)}$, where $I_2^{(\ell, r+1)}$ and $I_3^{(\ell, r+1)}$ may have no relation to $I_2^{(\ell, r)}$ and $I_3^{(\ell, r)}$. So, we inductively create $I^{(\ell, r+1)} = (I_1^{(\ell, r+1)}, I_2^{(\ell, r+1)}, I_3^{(\ell, r+1)})$ from $I^{(\ell, r)} = (I_1^{(\ell, r)}, I_2^{(\ell, r)}, I_3^{(\ell, r)})$, where importantly $|I_1^{(\ell, r)} \backslash I_1^{(\ell, r+1)}| \le 1$.

\subsection{Additional preliminaries} \label{subsec:more_prelims}

For our approximation guarantees, we require two additional preliminaries.
The first is to show a rough equivalence between solving $k$-means (resp., $k$-median) clustering and solving $k$-means (resp., $k$-median) clustering if allowed $O(1)$ additional clusters.
The second is to define the notion of negative-submodularity and its application for $k$-means and $k$-median.

First, we show that for any constant $C$ and
parameter $\eps$, if there exists a polynomial-time $\alpha$-approximation algorithm
for $k$-means or $k$-median in any metric space that
uses $k+C$ centers, then for any constant $\eps>0$, 
there exists a polynomial-time $\alpha(1+\eps)$-approximation algorithm that opens exactly $k$ centers.

More formally, the statement we prove is the following. Note that a similar statement was
proven in~\cite{LiS16}.
\begin{lemma}
\label{lem:center-reduction}
Let $C,\alpha$ be some absolute constants. 
Let $\mathcal{A}$ 
be an $\alpha$-approximation algorithm with running time $T$ for $k$-median (resp. $k$-means) that open $k+C$ centers. Then, for any $1/3 > \eps>0$, there exists an
$\alpha(1+\eps)$-approximation for $k$-median (resp. $k$-means) with running time 
$O(T+n^{\text{poly}(C/\eps)})$.
\end{lemma}
\begin{proof}
We give the proof of the $k$-median problem, the proof for
the $k$-means problem is identical, up to adjustment of the
constants.

To proceed, we need the following notion (due to~\cite{ORSS12}):
A $k$-median instance is said to be
\emph{$(1+\alpha)$-ORSS-separable} if the ratio of the 
cost of the optimum solution with $k-1$ centers to the cost of the optimum solution with
$k$ centers is at least $1+\alpha$.

We can now present our algorithm; For any $k$, the algorithm is as follows. 
    Compute an $\alpha$-approximate solution $S_C$ (with $k+C-C=k$ centers)
    to the $(k-C)$-median instance using $\mathcal{A}$.
    Then, for any $i=1,\ldots,C$, compute a solution $S_{i-1}$ with $k-i$ centers
    using the algorithms for $(1+\nf{\eps}{C})^{\eps/10}$-ORSS-separable
    instances of~\cite{awasthi2010stability,Cohen-AddadS17} 
    to obtain a $(1+\eps/3)$-approximate solution in time $n^{\text{poly}(C/\eps)}$.
    Output the solution $S^*$ of $S_0,\ldots, S_C$ with minimum $k$-median cost. 

We now turn to the analysis of the above algorithm. The running time follows immediately 
from its definition and the results of~\cite{awasthi2010stability,Cohen-AddadS17}. 
Let's then consider the approximation guarantee of the solution produced.
    For $0 \le i \le C$, let $\opt_i$ be the solution 
    to $(k-i)$-median, i.e.: the $k$-median problem with 
    $k-i$ centers.  Our goal is thus to show that the solution $S^*$ output by the above algorithm is an $\alpha(1+\eps)$-approximate
    solution to the $\opt_0$.

    If the cost of $\opt_C$ is within a $(1+\eps)$-factor
    of the cost of $\opt_0$, then the cost of the solution output by the algorithm is no 
    larger than the cost of solution $S_C$ whose total cost is thus
    at most $\alpha \text{cost}(\opt_C) \le \alpha(1+\eps)\text{cost}(\opt_0)$, as desired.

    Otherwise, since $\eps < 1/3$ we have that there exists a $i>0$ such that 
    $\text{cost}(\opt_i) \ge (1+\nf{\eps}{C})^{\eps/10} \text{cost}(\opt_{i-1})$.
    Let $i^*$ be the smallest $i$ such that the above holds. In which case, we have 
    both that $\text{cost}(\opt_{i^*-1}) \le (1+\eps/3)\text{cost}(\opt_{0})$
    and that the $(k-(i^* - 1))$-median instance is 
    $(1+\nf{\eps}{C})^{\eps/10}$-ORSS-separable. Therefore, by the results
    of~\cite{awasthi2010stability,Cohen-AddadS17}, the cost of the solution output by our 
    algorithm is no larger than $(1+\eps/3) \text{cost}(\opt_{i^* - 1})$ and so at most 
    $(1+\eps/3)^2  \text{cost}(\opt_{0}) \le (1+\eps)\text{cost}(\opt_0)$ by our choice of $\eps$, hence the lemma.
\end{proof}

Next, we describe the definition of submodular and negative-submodular set functions.

\begin{defn}
    Let $\Omega$ be a finite set, and let $f$ be a function from the set of subsets of $\Omega$, $\mathcal{P}(\Omega)$, to the real numbers $\BR$. Then, $f$ is \emph{submodular} over $\Omega$ if for any $X \subset Y \subsetneq \Omega$ and any $x \in \Omega \backslash Y$, we have that $f(X \cup \{x\}) - f(X) \ge f(Y \cup \{x\}) - f(Y)$. Likewise, $f$ is \emph{negative-submodular} over $\Omega$ if $-f$ is a submodular function: equivalently, if for any $X \subset Y \subsetneq \Omega$ and any $x \in \Omega \backslash Y$, we have that $f(X \cup \{x\}) - f(X) \le f(Y \cup \{x\}) - f(Y)$.
\end{defn}

The following claim, proven in \cite{cohenaddad2019fpt}, proves that the $k$-means and $k$-median objective functions are both negative-submodular functions.

\begin{proposition} \cite[Claim 10]{cohenaddad2019fpt}  \label{prop:submodular_1}
    Fix $\mathcal{F}$ and $\mathcal{D}$, and let $f: \mathcal{P}(\mathcal{F}) \to \BR$ be the function sending each subset $S \subset \mathcal{F}$ to $\text{cost}(\mathcal{D}, S)$. Then, $f$ is a negative-submodular function over $\mathcal{F}$, either if the cost is the $k$-median cost or if the cost is the $k$-means cost.
\end{proposition}

Given Proposition \ref{prop:submodular_1}, we can use standard properties of submodular functions to infer the following claims.

\begin{proposition}

\label{prop:submodular_2}
    Let $S_0 \subset S_1$ be sets of facilities, where $S_0$ has size $k_0$ and $S_1$ has size $k_0+k_1$. Then, for any $0 \le p \le 1$, if $S: S_0 \subset S \subset S_1$ is a set created including all of $S_0$ and then independently including each element in $S_1 \backslash S_0$ with probability $p$, then $\BE[\text{cost}(\mathcal{D}, S)] \le p \cdot \text{cost}(\mathcal{D}, S_1) + (1-p) \cdot \text{cost}(\mathcal{D}, S_0)$. 
\end{proposition}

\begin{proposition} \label{prop:submodular_3}
    Let $S_0 \subset S_1$ be sets of facilities, where $S_0$ has size $k_0$ and $S_1$ has size $k_0+k_1$. Then, if $S: S_0 \subset S \subset S_1$ is a set created by randomly adding exactly $k_2$ items from $S_1 \backslash S_0$ to $S_0$, for some fixed $0 \le k_2 \le k_1$, then we have $\BE[\text{cost}(\mathcal{D}, S)] \le \frac{k_2}{k_1} \cdot \text{cost}(\mathcal{D}, S_1) + (1-\frac{k_2}{k_1}) \cdot \text{cost}(\mathcal{D}, S_0)$.
\end{proposition}

\subsection{Analysis} \label{subsec:k_means_analysis}

The following theorem relating to the above algorithm was (essentially) proven by Ahmadian et al.~\cite{ahmadian2017better}, and will be very important in our analysis.

\begin{theorem} \label{thm:ahmadian_roundable}
    Algorithm \ref{alg:main} runs in $n^{O(1)}$ time (where the $O(1)$ may depend on $\eps$ and $\gamma$), and the following conditions hold.
\begin{enumerate}
    \item Let $k'$ be the minimum of $k$ and the sizes of all sets that become $I_1^{(0)}$ (i.e., the first part of each nested quasi-independent set $I^{(q)}$ that becomes $I^{(0)}$, as done in line 13 of the pseudocode).
    Then, every solution $\mathcal{S}^{(\ell)}$ that is generated when $\lambda$ is a certain value is $(\lambda, k')$-roundable. \label{cond:roundable}
    \item For any solution $\mathcal{S}$ that becomes $\mathcal{S}^{(0)}$, $\mathcal{F}_S^{(0)} = \emptyset$ (and so $\mathcal{D}_S^{(0)}$ is an empty function). In addition, $\mathcal{S} = \mathcal{S}^{(0)}$ has no corresponding bad clients, i.e., $\mathcal{D}_B = \emptyset$. \label{cond:no_special_or_bad}
    \item For any two consecutive solutions $\mathcal{S}^{(\ell)}$ and $\mathcal{S}^{(\ell+1)}$, we have that $\alpha^{(\ell)}$ and $\alpha^{(\ell+1)}$ are close. \label{cond:close}
    \item Every $I^{(\ell, r)}=(I_1^{(\ell, r)}, I_2^{(\ell, r)}, I_3^{(\ell, r)})$ is a nested quasi-independent set for the set of facilities $\mathcal{V}^{(\ell, r)}$. In addition, for every $1 \le r \le p_\ell$, $|I_1^{(\ell, r-1)} \backslash I_1^{(\ell, r)}| \le 1.$
\end{enumerate}
\end{theorem}

This theorem is technically stronger than what was proven in \cite{ahmadian2017better}, but follows from a nearly identical analysis to their paper, with a very minor tweak to the algorithm. We explain why Theorem \ref{thm:ahmadian_roundable} follows from their analysis in Appendix \ref{app:why_am_i_doing_this}.

\medskip

For the following lemmas (Lemmas \ref{lem:tau_bigger_than_alpha} until Proposition \ref{prop:duality_bound}), we consider a fixed family of conflict graphs $\{H(\delta)\}_{\delta > 0} = \{H^{(\ell, r)}(\delta)\}_{\delta > 0}$ on a hybrid set $\mathcal{V}=\mathcal{V}^{(\ell, r)}$, for some $r \ge 1$, where both $\mathcal{S}^{(\ell)}$ and $\mathcal{S}^{(\ell+1)}$ are $(\lambda, k')$-roundable. For some fixed $\delta_1 \ge \delta_2 \ge 2 \ge \delta_3,$ we let $(I_1, I_2, I_3)$ be a nested quasi-independent set of $\{H(\delta)\}$, i.e., the output of running all but the final step of Algorithm \ref{alg:lmp} with $V_1 = \mathcal{V}$, and treat it as fixed. 
However, we will let $p$ (required in the final step \ref{step:S}) be variable, though we may consider $p$ as initialized to some fixed $p_1$. 

Many of these results apply for both $k$-means and $k$-median. While we focus on $k$-means, we later explain how to make simple modifications to apply our results to the $k$-median problem. In addition, we will treat $\delta_1, \delta_2, \delta_3$ as fixed but $p$ as potentially variable. We let $\rho(p)$ represent the approximation constant from the LMP algorithm (i.e., in Lemma \ref{lem:main_lmp}) with probability $p$ (either for $k$-means or $k$-median, depending on constant).

We first show some crucial preliminary claims relating to the hybrid graph $\mathcal{V}^{(\ell, r)}$, where $r \ge 1.$

\begin{lemma} \label{lem:tau_bigger_than_alpha}
    For any client $j \in \mathcal{D}$ and any facility $i \in \bar{N}(j)$, $\tau_i \ge \alpha_j > c(j, i)$.
\end{lemma}

\begin{proof}
    Note that if $i \in \mathcal{V}^{(\ell)}$, then $\tau_i$ is the maximum $\alpha_{j'}^{(\ell)}$ over $j'$ such that $\alpha_{j'}^{(\ell)} > c(j', i)$ and $j' \in \mathcal{D}_S^{(\ell)}(i)$ if $i$ is special.
    Since $\alpha_{j'}^{(\ell)} > \alpha_{j'}$, this is at least the maximum $\alpha_{j'}$ over $j'$ such that $\alpha_{j'} > c(j', i)$ and $j' \in \mathcal{D}_S^{(\ell)}(i)$ if $i$ is special. But if $i \in \bar{N}(j)$, then indeed $\alpha_j > c(j, i)$ and $j \in \mathcal{D}_S^{(\ell)}(i)$ if $i$ is special (recall that $\mathcal{D}_S$ was defined based on whether $i$ is in $\mathcal{V}^{(\ell)}$ or $\mathcal{V}^{(\ell+1)}$). So, $\tau_i \ge \alpha_j$. By an identical argument, the same holds if $i \in \mathcal{V}^{(\ell+1)}$.
    
    Finally, note that we defined $\bar{N}(j)$ to precisely be the set of tight facilities $i$ in $\mathcal{V}$ with $\alpha_j > c(j, i)$, or special facilities $i$ in $\mathcal{V}$ with $\alpha_j > c(j, i)$ and $j \in \mathcal{D}_S(i)$. So, we always have $\alpha_j > c(j, i)$.
\end{proof}

\begin{lemma} \label{lem:RHS_positive_1/2} 
    Suppose that
    $S \subset I_1 \cup I_2 \cup I_3$ contains every point in $I_1$, and each point in $I_2$ and each point in $I_3$ with probability $p \le 0.5$ (not necessarily independently). Then, for any point $j$,
\[\BE\left[\alpha_j - \sum_{i \in \bar{N}(j) \cap S} (\alpha_j-c(j, i))\right] \ge 0.\]
\end{lemma}

\begin{remark}
    We note that this lemma holds even for bad clients $j \in \mathcal{D}_B$. In addition, we remark that we will be applying this lemma on $S$ as a nested quasi-independent set or something similar.
    
    Finally, we note that this lemma (and the following lemma, Lemma \ref{lem:cost_j_upper_bound}) are the only results where we directly use the fact that we are studying the $k$-means as opposed to the $k$-median problem.
\end{remark}

\begin{proof}[Proof of Lemma \ref{lem:RHS_positive_1/2}]
    Note that every point $i \in \bar{N}(j) \cap S$ satisfies $\tau_i \ge \alpha_j$ and $\alpha_j-c(j, i) > 0$, by Lemma \ref{lem:tau_bigger_than_alpha}.
    So, by linearity of expectation, it suffices to show that
\begin{equation} \label{eq:I2I3_bound}
    \alpha_j \ge \sum_{i \in \bar{N}(j) \cap I_1} (\alpha_j-c(j, i)) + \frac{1}{2} \cdot \sum_{i \in \bar{N}(j) \cap I_2} (\alpha_j-c(j, i)) + \frac{1}{2} \cdot \sum_{i \in \bar{N}(j) \cap I_3} (\alpha_j-c(j, i)).
\end{equation}

    Now, note that by the definition of $I_2$, every pair of points $(i, i')$ in $(I_1 \cup I_2) \cap \bar{N}(j)$ are separated by at least $\sqrt{\delta_2 \cdot \min(\tau_i, \tau_{i'})}$ distance. But since $i, i' \in \bar{N}(j)$, $\min(\tau_i, \tau_{i'}) \ge \alpha_j$ by Lemma \ref{lem:tau_bigger_than_alpha}. So, $d(i, i') \ge \sqrt{2 \cdot \alpha_j}$. Therefore,
\begin{align}
    \sum_{i \in \bar{N}(j) \cap (I_1 \cup I_2)} (\alpha_j-c(j, i)) \nonumber
    &\le |I_1 \cup I_2| \cdot \alpha_j - \frac{1}{2 \cdot |I_1 \cup I_2|} \cdot \sum_{i, i' \in \bar{N}(j) \cap (I_1 \cup I_2)} d(i, i')^2 \nonumber \\
    &\le |I_1 \cup I_2| \cdot \alpha_j - \frac{1}{2 \cdot |I_1 \cup I_2|} \cdot (|I_1 \cup I_2|) \cdot (|I_1 \cup I_2|-1) \cdot 2 \cdot \alpha_j \nonumber \\
    &= \alpha_j. \label{eq:I2_bound}
\end{align}
    Likewise, every pair of points $(i, i')$ in $(I_1 \cup I_3) \cap \bar{N}(j)$ are also separated by at least $\sqrt{\delta_2 \cdot \min(\tau_i, \tau_{i'})} \ge \sqrt{2 \cdot \alpha_j}$ distance. Therefore, the same calculations as in \eqref{eq:I2_bound} give us
\begin{equation} \label{eq:I3_bound}
    \sum_{i \in \bar{N}(j) \cap (I_1 \cup I_3)} (\alpha_j-c(j, i)) \le \alpha_j.
\end{equation}

    Averaging Equations \eqref{eq:I2_bound} and \eqref{eq:I3_bound} gives us Equation \eqref{eq:I2I3_bound}, which finishes the lemma.
\end{proof}

We next have the following lemma, which bounds the cost for clients that are not bad.

\begin{lemma} \label{lem:cost_j_upper_bound}
    Let $p < 0.5$
    and $S$ be generated by applying Step \ref{step:S} to $(I_1, I_2, I_3)$. Then, for every client $j \not\in \mathcal{D}_B$,
\[\BE[c(j, S)] \le \rho(p) \cdot (1+O(\eps)) \cdot \BE\left[\alpha_j - \sum_{i \in \bar{N}(j) \cap S} (\alpha_j-c(j, i))\right],\]
    where $\rho(p)$ represents the constant from Lemma \ref{lem:main_lmp}.
\end{lemma}

\begin{proof}
    By Lemma \ref{lem:tau_bigger_than_alpha}, we have that $\tau_i \ge \alpha_j$ for all $i \in \bar{N}(j)$.
    In addition, for every $j \not\in \mathcal{D}_B$, there exists a point $w(j)$ in $\mathcal{V}^{(\ell+1)}$ (so it has not been removed, i.e., it is still in $\mathcal{V} = \mathcal{V}^{(\ell, r)}$) such that $(1+\eps) \cdot \alpha_j^{(\ell+1)} \ge c(j, w(j))$ and $(1+\eps) \cdot \alpha_j^{(\ell+1)} \ge \tau_{w(j)}$. Since $\alpha_j^{(\ell+1)} \ge 1$ for all $j$ and $|\alpha_j-\alpha_j^{(\ell+1)}| \le \frac{1}{n^2},$ this means that $(1+O(\eps)) \cdot \alpha_j \ge c(j, w(j)), \tau_{w(j)}$. These pieces of information are sufficient for all of our calculations from Lemma \ref{lem:main_lmp} to go through. 
\end{proof}

We note that the expression $\alpha_j - \sum_{i \in \bar{N}(j) \cap S} (\alpha_j-c(j, i))$ may be somewhat unwieldy. Therefore, we provide an upper bound on its sum over $j \in \mathcal{D}.$

\begin{lemma} \label{lem:ugly_upper_bound}
    Let $S$ be any subset of $\mathcal{V} = \mathcal{V}^{(\ell, r)}$. Then,
\[\sum_{j \in \mathcal{D}} \left(\alpha_j-\sum_{i \in \bar{N}(j) \cap S} (\alpha_j-c(j, i))\right) \le \sum_{j \in \mathcal{D}} \alpha_j - \left(\lambda-\frac{1}{n}\right) \cdot |S| + 4 \gamma \cdot \text{OPT}_{k'}.\]
\end{lemma}

\begin{remark}
    We note that this lemma holds even when $\lambda < \frac{1}{n}$, i.e., $\lambda-\frac{1}{n} < 0$.
\end{remark}

\begin{proof}
    First, by splitting the sum based on tight and special facilities, we have that
\begin{align}
    &\hspace{0.5cm} \sum_{j \in \mathcal{D}}\left(\alpha_j - \sum_{i \in \bar{N}(j) \cap S} (\alpha_j-c(j, i))\right) \nonumber \\
    &= \sum_{j \in \mathcal{D}} \alpha_j - \sum_{\substack{i \in S \\ i \text{ tight}}} \sum_{j \in N(i)} (\alpha_j-c(j, i)) - \sum_{\substack{i \in S \\ i \text{ special}}} \sum_{j \in N(i) \cap \mathcal{D}_S(i)} (\alpha_j-c(j, i)). \label{eq:alpha_j_1}
\end{align}
    Now, we note that for any tight facility $i$, either $\sum_{j \in \mathcal{D}} \max(0, \alpha_j^{(\ell)}-c(j, i)) = z_i^{(\ell)} \in [\lambda, \lambda+\frac{1}{n}]$ or $\sum_{j \in \mathcal{D}} \max(0, \alpha_j^{(\ell+1)}-c(j, i)) = z_i^{(\ell+1)} \in [\lambda, \lambda+\frac{1}{n}]$. Since $\alpha^{(\ell)}$ and $\alpha^{(\ell+1)}$ are close, this means $0 \le \alpha_j^{(\ell)}-\alpha_j, \alpha_j^{(\ell+1)}-\alpha_j \le \frac{1}{n^2}$, and since there are $n$ clients in $\mathcal{D}$, this means that 
\begin{equation}
    \sum_{j \in N(i)} (\alpha_j-c(j, i)) = \sum_{j \in \mathcal{D}} \max(\alpha_j-c(j, i), 0) \ge \sum_{j \in \mathcal{D}} \max\left(\alpha_j^{(\ell')}-c(j, i)-\frac{1}{n^2}, 0\right) \ge \lambda-\frac{1}{n}, \label{eq:alpha_j_2}
\end{equation}
    for some choice of $\ell'$ in $\{\ell, \ell+1\}$.
    
    In addition, we know that both $\alpha^{(\ell)}$ and $\alpha^{(\ell+1)}$ are feasible solutions of $\text{DUAL}(\lambda+\frac{1}{n})$, and that $\alpha_j \le \alpha_j^{(\ell)}, \alpha_j^{(\ell+1)}$. Therefore, for any special facility $i \in \mathcal{F}_S^{(\ell)} \sqcup \mathcal{F}_S^{(\ell+1)}$, $\sum_{j \in N(i) \cap \mathcal{D}_S(i)} (\alpha_j-c(j, i)) \le \sum_{j \in \mathcal{D}} \max(0, \alpha_j-c(j, i)) \le \lambda+\frac{1}{n}$. But, we have that 
\begin{align*}
    \sum_{i \in \mathcal{F}_S^{(\ell')}} \sum_{j \in N(i) \cap \mathcal{D}_S^{(\ell')}(i)} (\alpha_j-c(j, i)) &\ge \sum_{i \in \mathcal{F}_S^{(\ell')}} \sum_{j \in \mathcal{D}_S^{(\ell')}(i)} \max\left(0, \alpha_j^{(\ell')}-\frac{1}{n^2}-c(j, i)\right) \\
    &\ge \lambda \cdot |\mathcal{F}_S^{(\ell')}| - \gamma \cdot \text{OPT}_{k'} - \frac{|\mathcal{F}_S^{(\ell')}|}{n},
\end{align*}
    for both $\ell' = \ell$ and $\ell' = \ell+1$ (the last inequality follows by Condition \ref{cond:4} of Definition \ref{def:roundable}). So, if we let $e_i$ represent $\lambda+\frac{1}{n} - \sum_{j \in N(i) \cap \mathcal{D}_S(i)} ( \alpha_j-c(j, i))$ for each special facility $i$, we have that $e_i \ge 0$ but $\sum_{i \in \mathcal{F}_S^{(\ell)}} e_i \le \frac{2}{n} \cdot |\mathcal{F}_S^{(\ell)}| + \gamma \cdot \text{OPT}_{k'} \le 2 \gamma \cdot \text{OPT}_{k'}$, since $|\mathcal{F}_S^{(\ell)}| \le n$ and $\text{OPT}_{k'} \ge n \ge \frac{2}{\gamma}$. 
    Similarly, $\sum_{i \in \mathcal{F}_S^{(\ell+1)}} e_i \le 2 \gamma \cdot \text{OPT}_{k'}.$ So, this means that
\[\sum_{\substack{i \in S \\ i \text{ special}}} e_i \le 4 \gamma \cdot \text{OPT}_{k'},\]
    which means that 
\begin{align}
    \sum_{\substack{i \in S \\ i \text{ special}}} \sum_{j \in N(i) \cap \mathcal{D}_S(i)} (\alpha_j-c(j, i)) &\ge \sum_{\substack{i \in S \\ i \text{ special}}} \left(\lambda + \frac{1}{n} - e_i\right) \nonumber \\
    &\ge \left(\lambda+\frac{1}{n}\right) \cdot |\{i \in S: i \text{ special}\}|-4 \gamma \cdot \text{OPT}_{k'}. \label{eq:alpha_j_3}
\end{align}

    Thus, by combining Equations \eqref{eq:alpha_j_1}, \eqref{eq:alpha_j_2}, and \eqref{eq:alpha_j_3}, we get
\begin{align*}
    \sum_{j \in \mathcal{D}} \left(\alpha_j - \sum_{i \in \bar{N}(j) \cap S} (\alpha_j-c(j, i))\right) 
    &\le \sum_{j \in \mathcal{D}} \alpha_j - \sum_{\substack{i \in S \\ i \text{ tight}}} \left(\lambda-\frac{1}{n}\right) - \left(\lambda+\frac{1}{n}\right) \cdot |\{i \in S: i \text{ special}\}|+4 \gamma \cdot \text{OPT}_{k'} \\
    &\le \sum_{j \in \mathcal{D}} \alpha_j - \left(\lambda-\frac{1}{n}\right) \cdot |S| + 4 \gamma \cdot \text{OPT}_{k'}. \qedhere
\end{align*}
\end{proof}

Next, we show that the bad clients do not contribute much to the total cost.

\begin{lemma} \label{lem:cost_bad_upper_bound}
    Let $S$ be a subset of $\mathcal{V}$ containing $I_1$. Then, we have that
\[\sum_{j \in \mathcal{D}_B} c(j, S) \le O(\gamma) \cdot \text{OPT}_{k'}.\]
\end{lemma}

\begin{proof}
    Note that for every point $j \in \mathcal{D}_B$, there exists a facility $w(j) \in \mathcal{V}^{(\ell+1)}$ such that
\[\sum_{j \in \mathcal{D}_B} \left(c(j, w(j)) + \tau_{w(j)}\right) \le O(\gamma) \cdot \text{OPT}_{k'},\]
    because $\mathcal{S}^{(\ell+1)}$ is $(\lambda, k')$-roundable. Now, note that $d(j, S) \le d(j, I_1) \le d(j, w(j)) + d(w(j), I_1)$, so $c(j, S) \le 2 [c(j, w(j)) + c(w(j), I_1)]$. But since $w(j) \in \mathcal{V}^{(\ell+1)} \subset \mathcal{V}$, $c(w(j), I_1) \le \delta_1 \cdot \tau_{w(j)}$, and therefore,
\[\sum_{j \in \mathcal{D}_B} c(j, S) \le 2 \cdot \sum_{j \in \mathcal{D}_B} \left(c(j, w(j)) + \delta_1 \cdot \tau_{w(j)}\right) \le O(\gamma) \cdot \text{OPT}_{k'},\]
    where the final inequality follows by Condition 3c) of Definition \ref{def:roundable}.
\end{proof}

We now combine our previous lemmas to obtain the following bound on the expected cost of $S$, giving a result that bounds the overall expected cost in terms of the dual solution.

\begin{lemma} \label{lem:expected_cost_S}
    Suppose that $S$ is generated by applying Step \ref{step:S} to $(I_1, I_2, I_3)$ with the probability set to $p < 0.5$.
    Then, the expected cost $\BE[\text{cost}(\mathcal{D}, S)]$ is at most
\[\rho(p) \cdot (1+O(\eps)) \cdot \left[\sum_{j \in \mathcal{D}} \alpha_j - \left(\lambda-\frac{1}{n}\right) \cdot \BE[|S|]\right] + O(\gamma) \cdot \text{OPT}_{k'},\]
    where $\rho(p)$ represents the constant from Lemma \ref{lem:main_lmp}.
\end{lemma}

\begin{proof}
    We will abbreviate $\rho := \rho(p)$.
    We can split up the cost based on good (i.e., not in $\mathcal{D}_B$) points $j$ and bad points $j$. Indeed, doing this, we get
\begin{align*}
    \BE[\text{cost}(\mathcal{D}, S)] &= \sum_{j \not\in \mathcal{D}_B} \BE[c(j, S)] + \sum_{j \in \mathcal{D}_B} \BE[c(j, S)] \\
    &\le \rho \cdot (1+O(\eps)) \cdot \sum_{j \not\in \mathcal{D}_B} \BE\left[\alpha_j - \sum_{i \in \bar{N}(j) \cap S} (\alpha_j-c(j, i))\right] + O(\gamma) \cdot \text{OPT}_{k'} \\
    &\le \rho \cdot (1+O(\eps)) \cdot \sum_{j \in \mathcal{D}} \BE\left[\alpha_j - \sum_{i \in \bar{N}(j) \cap S} (\alpha_j-c(j, i))\right] + O(\gamma) \cdot \text{OPT}_{k'} \\
    &\le \rho \cdot (1+O(\eps)) \cdot \left[\sum_{j \in \mathcal{D}} \alpha_j - \left(\lambda-\frac{1}{n}\right) \cdot \BE[|S|]\right] + O(\gamma) \cdot \text{OPT}_{k'}.
\end{align*}
    In the above equation, the second line follows from Lemmas \ref{lem:cost_j_upper_bound} and \ref{lem:cost_bad_upper_bound}. The third line is true since $\BE\left[\alpha_j-\sum_{i \in \bar{N}(j) \cap S} (\alpha_j-c(j, i))\right] \ge 0$ for all $j$, even in $\mathcal{D}_B$, by Lemma \ref{lem:RHS_positive_1/2}. Finally, the fourth line is true because of Lemma \ref{lem:ugly_upper_bound}.
\end{proof}

Next, we show that under certain circumstances, we can find a solution $S$ of size at most $k$ satisfying a similar condition to Lemma \ref{lem:expected_cost_S}, with high probability.

\begin{lemma} \label{lem:cost_bound}
    Suppose that $|I_1|+p \cdot |I_2 \cup I_3| = k$ for some integer $k$ (which may be larger than $n$) and some $p \in [0.01, 0.49]$. Then, for any sufficiently large integer $C$, if $|I_2 \cup I_3| \ge 100 \cdot C^{4}$,
    then there exists a polynomial-time randomized algorithm that outputs a set $S$ such that $I_1 \subset S \subset I_1 \cup I_2 \cup I_3$, and with probability at least $9/10$, $|S| \le k$ and $$\text{cost}(\mathcal{D}, S) \le \rho\left(p-\frac{2}{C}\right) \cdot \left(1+\frac{300}{C}\right) \cdot (1+O(\eps)) \cdot \left[\sum_{j \in \mathcal{D}} \alpha_j - \left(\lambda-\frac{1}{n}\right) \cdot k\right] + O(\gamma) \cdot \text{OPT}_{k'}.$$
\end{lemma}

\begin{remark}
    By repeating this randomized algorithm polynomially many times and outputting the lowest-cost solution $S$ with $|S| \le k$, we can make the failure probability exponentially small.
\end{remark}

\begin{proof}
Let $r = |I_2|$, and partition $I_2 \cup I_3$ into $T_1, \dots, T_r$, where each $T_\ell$ consists of a point $i_2 \in I_2$ and $q^{-1}(i_2)$, i.e., $i_2$'s preimage under the map $q$. We assume WLOG that the $T_\ell$'s are sorted in non-increasing order of size, and write $x_\ell$ as the unique point in $T_\ell \cap I_2$. Define $s$ such that $|T_1| \ge \dots \ge |T_s| \ge C > |T_{s+1}| \ge \dots \ge |T_r|$ (note that $s$ may be $0$ or $r$). Note that $s \le \frac{|I_2 \cup I_3|}{C}$ so $\frac{s}{|I_2 \cup I_3|} \le \frac{1}{C}$. Now, set $p' = p-\frac{2}{C}$, and consider creating the following set $S$:
\begin{itemize}
    \item For each $i \in I_1,$ include $i \in S$.
    \item For each $\ell \le s,$ include $x_\ell \in S$, and for each $i$ in $T_\ell \backslash \{x_\ell\}$, include $i \in S$ independently with probability $p'$. 
    \item For each $\ell > s,$ flip a fair coin. If it lands heads, include $x_\ell$ with probability $2p'$, and if it lands tails, include each point in $T_\ell \backslash \{x_\ell\}$ independently with probability $2p'$.
\end{itemize}

    The expected size of $S$ is $|I_1| + s + (|T_1| + \cdots + |T_s| - s) \cdot p' + (|T_{s+1}|+\cdots+|T_r|) \cdot p' = |I_1| + p' \cdot |I_2 \cup I_3| + (1-p') \cdot s \le |I_1|+p' \cdot |I_2 \cup I_3| + s.$ Therefore, since $p' = p-\frac{2}{C},$ the expected size of $S$ as at most $|I_1|+(p-\frac{2}{C}) \cdot |I_2 \cup I_3| + \frac{|I_2 \cup I_3|}{C} = |I_1| + (p-\frac{1}{C}) \cdot |I_2 \cup I_3|.$ To bound the variance of $S$, we note that each point in $I_1$ and each $x_\ell$ for $\ell \le s$ is deterministically in $S$, each point in $T_\ell \backslash \{x_\ell\}$ for $\ell \le s$ is independently selected with probability $p'$, and the number of points from each $T_\ell$ for $\ell > s$ is some independent random variable bounded by $|T_\ell| \le C$. So, the variance can be bounded by $(|T_1|+\cdots+|T_s|) + \sum_{\ell = s+1}^{r} |T_\ell|^2 \le \max_{s+1 \le \ell \le r} |T_\ell| \cdot \left(|T_1|+\cdots+|T_r|\right) \le C \cdot |I_2 \cup I_3|$. So, by Chebyshev's inequality, with probability at least $1-\frac{1}{10C}$, 
\begin{align*}
    |S| &\le |I_1|+(p-\frac{1}{C}) \cdot |I_2 \cup I_3| + \sqrt{10 C \cdot C \cdot |I_2 \cup I_3|} \\
    &\le |I_1|+p|I_2 \cup I_3| - \frac{1}{C} |I_2 \cup I_3| + \sqrt{10} C \cdot \sqrt{|I_2 \cup I_3|} \\
    &\le |I_1|+p|I_2 \cup I_3| = k,
\end{align*}
    where the final inequality is true since $|I_2 \cup I_3| \ge 100 C^4$.
    
    
    Next, we bound the expected cost of $S$. First, consider running the final step \ref{step:S} of the LMP algorithm on $(I_1, I_2, I_3)$ using probability $p'$. This would produce a set $S_0$ such that 
\begin{align}
    \BE\left[\text{cost}(\mathcal{D}, S_0)\right] &\le \rho(p') \cdot (1+O(\eps)) \cdot \BE\left[\sum_{j \in \mathcal{D}} \alpha_j - \left(\lambda-\frac{1}{n}\right) \cdot \BE[|S_0|]\right] + O(\gamma) \cdot \text{OPT}_{k'} \nonumber\\
    &= \rho\left(p-\frac{2}{C}\right) \cdot (1+O(\eps)) \cdot \left[\sum_{j \in \mathcal{D}} \alpha_j - \left(\lambda-\frac{1}{n}\right) \cdot \left(|I_1|+\left(p-\frac{2}{C}\right) \cdot |I_2 \cup I_3|\right)\right] + O(\gamma) \cdot \text{OPT}_{k'} \nonumber\\
    &= \rho\left(p-\frac{2}{C}\right) \cdot (1+O(\eps)) \cdot \left[\sum_{j \in \mathcal{D}} \alpha_j - \left(\lambda-\frac{1}{n}\right) \cdot k + \frac{2}{C} \cdot \left(\lambda-\frac{1}{n}\right) \cdot |I_2 \cup I_3|\right] + O(\gamma) \cdot \text{OPT}_{k'}. \label{eq:cost_D_S0}
\end{align}
    Above, the first line follows by Lemma \ref{lem:expected_cost_S}, the second line follows by definition of $p'$ and $S_0$, and the third line follows from the fact that $|I_1| + p|I_2 \cup I_3| = k$.
    
Now, note that if we had performed the final step of the LMP algorithm on $(I_1, I_2, I_3)$ using probability $1/2$ instead of $p'$, the set (call it $S_1$) would have satisfied 
\[\sum_{j \in \mathcal{D}} \alpha_j - \left(\lambda-\frac{1}{n}\right) \cdot \left(|I_1|+\frac{1}{2}(|I_2|+|I_3|)\right) + 4 \gamma \cdot \text{OPT}_{k'} \ge \BE\left[\sum_{j \in \mathcal{D}} \left(\alpha_j - \sum_{i \in \bar{N}(j) \cap S} (\alpha_j-c(j, i))\right)\right] \ge 0\]
by Lemmas \ref{lem:ugly_upper_bound} and \ref{lem:RHS_positive_1/2}. This means that $\sum_{j \in \mathcal{D}} \alpha_j \ge (\lambda-\frac{1}{n}) \cdot (|I_1|+1/2 \cdot |I_2 \cup I_3|) - 4 \gamma \cdot \text{OPT}_{k'}.$
Therefore, again using the fact that $|I_1|+p \cdot |I_2 \cup I_3| = k$, we have that 
\[\sum_{j \in \mathcal{D}} \alpha_j - \left(\lambda-\frac{1}{n}\right) \cdot k \ge \left(\lambda-\frac{1}{n}\right) \cdot \left(\frac{1}{2}-p\right) \cdot |I_2 \cup I_3| - 4\gamma \cdot \text{OPT}_{k'}.\]
We can rewrite this to obtain
\begin{equation} \label{eq:remainder_bound}
    \left(\lambda-\frac{1}{n}\right) \cdot |I_2 \cup I_3| \le \frac{1}{1/2-p} \cdot \left(\sum_{j \in \mathcal{D}} \alpha_j - \left(\lambda-\frac{1}{n}\right) \cdot k\right) + O(\gamma) \cdot \text{OPT}_{k'},
\end{equation}
since $p \in [0.01, 0.49]$.

In this and the next paragraph, we prove that $\BE[\text{cost}(\mathcal{D}, S)] \le \BE[\text{cost}(\mathcal{D}, S_0)]$.
To see why, we consider a coupling of the randomness to generate a sequence of sets $S_0, S_1, \dots, S_s = S$. To do so, for each point $i \in I_1$, we automatically include $i \in S = S_s$ and $i \in S_h$ for all $0 \le h < s$. Now, for each $1 \le \ell \le r$, we first create a temporary set $\tilde{T}_{\ell} \subset T_\ell \backslash \{x_{\ell}\}$ by including each point $i \in T_\ell \backslash \{x_{\ell}\}$ to be in $\tilde{T}_\ell$ independently with probability $2p$. Then, we create two sets $T_\ell^{(0)} \subset T_\ell$ and $T_{\ell}^{(1)} \subset T_\ell$ as follows. For $T_\ell^{(0)}$, we include each point in $\tilde{T}_\ell$ independently, with probability $1/2$, and always include $x_\ell \in T_\ell^{(0)}$. For $T_\ell^{(1)}$, we flip a fair coin: if the coin lands heads, we only include $x_\ell$, but if the coin lands tails, we do not include $x_\ell$ but include all of $\tilde{T}_\ell$. We remark that overall, $T_\ell^{(0)}$ includes each point in $T_\ell \backslash \{x_\ell\}$ independently with probability $p$.

Now, for each $0 \le h \le s$, we define $S_h := \left(\bigcup_{1 \le \ell \le h} T_\ell^{(0)}\right) \cup \left(\bigcup_{h < \ell \le r} T_\ell^{(1)}\right)$.
One can verify that $S$ and $S_0$ have the desired distribution, since $S_0 = \bigcup_{1 \le \ell \le r} T_\ell^{(1)}$ is precisely the distribution obtained after applying step \ref{step:S} of the LMP algorithm on $(I_1, I_2, I_3)$, but $S$ takes $T_\ell^{(0)}$ instead of $T_\ell^{(1)}$ for each $\ell \le s$, which is precisely the desired distribution for $S$ (as we defined $S$ at the beginning of this lemma's proof).
To show that $\BE[\text{cost}(\mathcal{D}, S)] \le \BE[\text{cost}(\mathcal{D}, S_0)]$, it suffices to show that $\BE[\text{cost}(\mathcal{D}, S_h)] \le \BE[\text{cost}(\mathcal{D}, S_{h-1})]$ for all $1 \le h \le s$.
However, note that because of our coupling, the only difference between $S_h$ and $S_{h-1}$ relates to points in $T_h$.
If we let $S_{h-1}' = S_{h-1} \cup \{x_\ell\}$ be the set that always includes $x_\ell$ but includes the entirety of $\tilde{T}_h$ with probability $1/2$, then clearly $\BE[\text{cost}(\mathcal{D}, S_{h-1}')] \le \BE[\text{cost}(\mathcal{D}, S_{h-1})]$.
In addition, if we condition on the set $\tilde{T}_h$, the only difference between $S_{h-1}'$ and $S_h$ is that $S_h$ includes each point in $\tilde{T}_h$ with $1/2$ probability, whereas $S_{h-1}$ either includes the entirety of $\tilde{T}_h$ with $1/2$ probability or includes none of $\tilde{T}_h$.
Therefore, by Proposition \ref{prop:submodular_2}, using the negative-submodularity of $k$-means~\cite{cohenaddad2019fpt}, we have that $\BE[\text{cost}(\mathcal{D}, S_h)] \le \BE[\text{cost}(\mathcal{D}, S_{h-1}')]$. So, we have that $\BE[\text{cost}(\mathcal{D}, S_h)] \le \BE[\text{cost}(\mathcal{D}, S_{h-1}')] \le \BE[\text{cost}(\mathcal{D}, S_{h-1})]$, which means that
\begin{equation} \label{eq:cost_S0_S_inequality}
    \BE[\text{cost}(\mathcal{D}, S)] = \BE[\text{cost}(\mathcal{D}, S_s)] \le \BE[\text{cost}(\mathcal{D}, S_{s-1})] \le \cdots \le \BE[\text{cost}(\mathcal{D}, S_1)] \le \BE[\text{cost}(\mathcal{D}, S_0)].
\end{equation}

In summary,
\begin{align*}
    \BE[\text{cost}(\mathcal{D}, S)] &\le \BE[\text{cost}(\mathcal{D}, S_0)] \\
    &\le \rho\left(p-\frac{2}{C}\right) \cdot (1+O(\eps)) \cdot \left[\sum_{j \in \mathcal{D}} \alpha_j - \left(\lambda-\frac{1}{n}\right) \cdot k + \frac{2}{C} \cdot \left(\lambda-\frac{1}{n}\right) \cdot |I_2 \cup I_3|\right] + O(\gamma) \cdot \text{OPT}_{k'} \\
    &\le \rho\left(p-\frac{2}{C}\right) \cdot (1+O(\eps)) \cdot \left(1 + \frac{2/C}{1/2 - p}\right) \cdot \left[\sum_{j \in \mathcal{D}} \alpha_j - \left(\lambda-\frac{1}{n}\right) \cdot k\right] + O(\gamma) \cdot \text{OPT}_{k'} \\
    &\le \rho\left(p-\frac{2}{C}\right) \cdot (1+O(\eps)) \cdot \left(1 + \frac{200}{C}\right) \cdot \left[\sum_{j \in \mathcal{D}} \alpha_j - \left(\lambda-\frac{1}{n}\right) \cdot k\right] + O(\gamma) \cdot \text{OPT}_{k'}.
\end{align*}
Above, the first line follows from Equation \eqref{eq:cost_S0_S_inequality}, the second line follows from Equation \eqref{eq:cost_D_S0}, the third line follows from Equation \eqref{eq:remainder_bound}, and the fourth line follows since $1/2-p \ge 0.01.$
So, with probability at least $1-\frac{1}{10C}$, $|S| \le k$, and by Markov's inequality,
\[\text{cost}(\mathcal{D}, S) \le \rho\left(p-\frac{2}{C}\right) \cdot (1+O(\eps)) \cdot \left(1 + \frac{300}{C}\right) \cdot \left[\sum_{j \in \mathcal{D}} \alpha_j - \left(\lambda-\frac{1}{n}\right) \cdot k\right] + O(\gamma) \cdot \text{OPT}_{k'}\]
with probability at least $\frac{10}{C}$.
So, both of these hold simultaneously with probability at least $\frac{9}{C},$ and by repeating the procedure $O(C)$ times, we will find our desired set $S$ with probability $9/10$.
\end{proof}

Our upper bound on the cost has so far been based on terms of the form $\sum_{j \in \mathcal{D}} \alpha_j - \left(\lambda-\frac{1}{n}\right) \cdot k.$ We note that this value is at most roughly $\text{OPT}_k$. Specifically, we note the following:

\begin{proposition} \label{prop:duality_bound}
    If $\alpha^{(\ell+1)}$ is a feasible solution to $\text{DUAL}(\lambda),$ then $\sum_{j \in \mathcal{D}} \alpha_j - \left(\lambda+\frac{1}{n}\right) \cdot k \le \text{OPT}_k.$
\end{proposition}

\begin{proof}
    Recall that $\alpha_j = \min(\alpha_j^{(\ell)}, \alpha_j^{(\ell+1)}),$ and that $\{\alpha_j^{(\ell+1)}\}$ is a feasible solution to $\text{DUAL}(\lambda+\frac{1}{n})$. Therefore, by duality, we have that 
\[\sum_{j \in \mathcal{D}} \alpha_j - \left(\lambda+\frac{1}{n}\right) \cdot k \le \sum_{j \in \mathcal{D}} \alpha_j^{(\ell+1)} - \left(\lambda+\frac{1}{n}\right) \cdot k \le \text{OPT}_k. \qedhere\]
\end{proof}

One potential issue is that if our goal is to obtain a good approximation to optimal $k$-means, the $\gamma \cdot \text{OPT}_{k'}$ error, which should be negligible, may appear too large if $k'$ is smaller than $k$. To fix this, we show that $\text{OPT}_{k'} = O(\text{OPT}_k)$ in certain cases, which we later show we will satisfy. For the following lemma, we return to considering a single roundable solution $\mathcal{S} = (\alpha, z, \mathcal{F}_S, \mathcal{D}_S)$, and let $\mathcal{V}$ represent the corresponding set of tight or special facilities corresponding to $\mathcal{S}.$

\begin{lemma} \label{lem:opt_k'_opt_k}
    Let $(\alpha, z, \mathcal{F}_S, \mathcal{D}_S)$ be $(\lambda', k')$-roundable for some $\lambda' \ge 0$, where $\mathcal{F}_S = \emptyset$ and the set of corresponding bad clients is $\mathcal{D}_B = \emptyset$. Define $\{H(\delta)\}$ as the corresponding family of conflict graphs, with some fixed nested quasi-independent set $(I_1, I_2, I_3)$. Suppose that $k \le \min\left(n-1, |I_1|+p \cdot |I_2 \cup I_3|\right)$ for some $p \le 0.49$, and that $k' = \min(k, |I_1|)$. Then, $\text{OPT}_{k'} = O\left(\text{OPT}_k\right)$.
\end{lemma}

\begin{proof}
    If $k' = k$, the result is trivial. So, we assume that $k' = |I_1|$.

    Since $\mathcal{D}_B$ is empty, we have that every client $j$ has a tight witness $w(j)$ (since there are no special facilities) such that $(1+\eps) \cdot \alpha_j \ge c(j, w(j))$ and $(1+\eps) \cdot \alpha_j \ge \tau_{w(j)}$. In addition, we have that for any $i \in I_1 \cup I_2 \cup I_3$, $i$ is tight which means $\sum_{j \in N(i)} (\alpha_j-c(j, i)) = z_i \in [\lambda', \lambda'+\frac{1}{n}]$. Therefore,
\[\sum_{j \in \mathcal{D}} \left[\alpha_j-\sum_{i \in N(j) \cap I_1} (\alpha_j-c(j, i))\right] = \sum_{j \in \mathcal{D}} \alpha_j - \sum_{i \in I_1} \sum_{j \in N(i)} (\alpha_j-c(j, i)) \le \sum_{j \in \mathcal{D}} \alpha_j - \lambda' \cdot |I_1|.\]
    Then, we can use the LMP approximation to get that
\[\sum_{j \in \mathcal{D}} c(j, I_1) \le \rho_1 \cdot (1+O(\eps)) \cdot \sum_{j \in \mathcal{D}} \left(\alpha_j - \sum_{i \in N(j) \cap I_1} (\alpha_j-c(j, i))\right) \le \rho_1 \cdot (1+O(\eps)) \cdot \left(\sum_{j \in \mathcal{D}} \alpha_j - \lambda' \cdot |I_1|\right),\]
    where $\rho_1 = \rho(0)$, i.e., assuming that no point in $I_2$ or $I_3$ is included as part of the set.
    In addition, we know that if $S$ is created by including all of $I_1$ and each point in $I_2 \cup I_3$ with probability $\frac{1}{2}$, then
\begin{align*}
    \sum_{j \in \mathcal{D}} \alpha_j - \lambda' \cdot \left(|I_1| + \frac{1}{2} |I_2 \cup I_3|\right) &\ge \BE\left[\sum_{j \in \mathcal{D}} \alpha_j - \sum_{i \in S} \sum_{j \in N(i)} (\alpha_j-c(j, i))\right] \\
    &= \sum_{j \in \mathcal{D}} \BE\left[\alpha_j-\sum_{i \in N(j) \cap S} (\alpha_j-c(j, i))\right] \\
    &\ge 0.
\end{align*}
    Above, the first inequality follows since $\sum_{j \in N(i)} (\alpha_j-c(j, i)) = z_i \ge \lambda'$ for any tight $i$, and the final inequality follows because of Lemma \ref{lem:RHS_positive_1/2}.
    
    To summarize, we have that there exists a constant $\theta = \frac{1}{\rho_1 \cdot (1+O(\eps))}$ such that
\begin{align}
    \theta \cdot \sum_{j \in \mathcal{D}} c(j, I_1) &\le \sum_{j \in \mathcal{D}} \alpha_j - |I_1| \cdot \lambda' \label{eq:IS_1}
\intertext{and}
    0 &\le \sum_{j \in \mathcal{D}} \alpha_j-\left(|I_1|+\frac{1}{2} \cdot |I_2 \cup I_3|\right) \cdot \lambda'. \label{eq:IS_2}
\end{align}
    Therefore, by taking a weighted average of Equations \eqref{eq:IS_1} and \eqref{eq:IS_2}, we get
\[(1-2p) \cdot \theta \cdot \text{OPT}_{k'} \le (1-2p) \cdot \theta \cdot \sum_{j \in \mathcal{D}} c(j, I_1) \le \sum_{j \in \mathcal{D}} \alpha_j - \left(|I_1| + p \cdot |I_2 \cup I_3|\right) \cdot \lambda' \le \sum_{j \in \mathcal{D}} \alpha_j - k \cdot \lambda',\]
    where the first inequality is true since $k' = |I_1|$ and the last inequality is true since $|I_1|+p \cdot |I_2 \cup I_3| \ge k.$
    Thus, since $p \le 0.49$ and since $\rho_1 = O(1)$, we have that $\text{OPT}_{k'} = O\left(\sum_{j \in \mathcal{D}} \alpha_j - k \cdot \lambda'\right)$.
    
    Finally, since $\{\alpha_j\}$ is a feasible solution to $\text{DUAL}(\lambda'+\frac{1}{n})$, this means that $\sum_{j \in \mathcal{D}} \alpha_j - k \cdot \lambda' = \sum_{j \in \mathcal{D}} \alpha_j - k \cdot (\lambda'+\frac{1}{n}) + \frac{k}{n} \le \text{OPT}_{k} + \frac{k}{n}.$ However, if $k \le n-1$, then $\frac{k}{n} \le 1 \le \text{OPT}_k$. So, $\text{OPT}_{k'} = O(\text{OPT}_k)$.
\end{proof}

Recall that $H(\delta)$ represents the conflict graph $H^{(\ell, r)}(\delta)$. We will also let $H'(\delta)$ represent the conflict graph $H^{(\ell, r+1)}(\delta)$. In that case, $H'(\delta)$ is the same as $H(\delta)$ except with one vertex removed. Recall $(I_1, I_2, I_3)$ was a nested quasi-independent set of $\{H(\delta)\}$, and let $(I_1', I_2', I_3')$ be a nested quasi-independent set of $\{H'(\delta)\}$, such that $|I_1 \backslash I_1'| \le 1$ and $|I_1|+p_1|I_2 \cup I_3| \ge k > |I_1'|+p_1|I_2' \cup I_3'|$.


\begin{theorem} \label{thm:main}
    Let $C > 0$ be an arbitrarily large constant, and $\eps < 0$ be an arbitrarily small constant. Given the sets $(I, I')$ obtained in Algorithm \ref{alg:main}, in polynomial time we can obtain a solution for Euclidean $k$-means with approximation factor at most
\begin{equation} \label{eq:approximation}
    \left(1+\eps\right) \cdot \max_{r \ge 1} \min\left(\rho\left(\frac{p_1}{r}\right), \rho(p_1) \cdot \left(1+\frac{1}{4r \cdot \left(\frac{r}{2p_1}-1\right)}\right)\right).
\end{equation}
\end{theorem}

\begin{proof}
First, we remark that it suffices to obtain a set of facilities of size at most $k+c_0 \cdot C^4$ with cost at most $K \cdot \text{OPT}_{k-c_0 \cdot C^4}$ for any fixed constant $c_0$ (for all $1 \le k \le n-1$), where $K$ is the value in Equation \eqref{eq:approximation}. Indeed, we can apply Lemma \ref{lem:center-reduction} to obtain a solution of size $k-c_0 \cdot C^4$ and cost $(1+\eps) \cdot K \cdot \text{OPT}_{k-c_0 \cdot C^4}$ in polynomial time, hence obtaining a $(1+\eps) \cdot K$-approximate solution for $(k-c_0 \cdot C^4)$ means clustering for all $1 \le k \le n-1$. Thus, we get the desired approximation ratio for all $k \le n-c_0 C^4-1$, but for $k \ge n-c_0 C^4-1$, we can enumerate
    all the $(2n)^{c_0 C^{4}-1}$ different clusterings
    of the input that have at most $c_0 C^{4}-1$ non-singleton parts and solve $k$-clustering exactly in $n^{O(C^4)}$ time.

    Algorithm \ref{alg:main} stops once we have found the first hybrid conflict graph $\mathcal{V}^{(\ell, r)}$ for some $r \ge 1$ where the corresponding nested quasi-independent set $(I_1^{(\ell, r)}, I_2^{(\ell, r)}, I_3^{(\ell, r)})$ satisfies $|I_1^{(\ell, r)}|+p_1 \cdot |I_2^{(\ell, r)} \cup I_3^{(\ell, r)}| < k$. Let $\mathcal{V}' := \mathcal{V}^{(\ell, r)}$ and let $(I_1', I_2', I_3') := (I_1^{(\ell, r)}, I_2^{(\ell, r)}, I_3^{(\ell, r)})$. In addition, let $\mathcal{V} := \mathcal{V}^{(\ell, r-1)}$ and $(I_1, I_2, I_3) := (I_1^{(\ell, r-1)}, I_2^{(\ell, r-1)}, I_3^{(\ell, r-1)})$. If $r = 1$ then $r-1 = 0$, which may be problematic since our previous lemmas can only be used for $r \ge 1$. However, we note that $I^{(\ell, 0)} = I^{(\ell)} = I^{(\ell-1, p_{\ell-1})}$ if $\ell \ge 1$, and that $I^{(0, 0)} = I^{(0)}$ was previously labeled as $I^{(q)} = I^{(q-1, p_{q-1})}.$ The only exception to this is the case when $\ell = 0, r = 1$, and $I^{(0, 0)}$ is the initialized solution created in the first line of the algorithm. In this case, however, recall from our initialization that $I_1^{(0, 0)} = \mathcal{F}$ is the full set of facilities, and $I_1^{(0, 1)}$ will just be an extension of this set, so $|I_1^{(0, 1)}| \ge |\mathcal{F}| \ge k$. Therefore, $I = (I_1, I_2, I_3)$ and $I' = (I_1', I_2', I_3')$ are both expressible as nested quasi-independent sets of merged conflict graphs. However, if $I = I^{(0, 1)}$ and $I' = I^{(0, 0)}$, then we may need to express $I' = I^{(q-1, p_{q-1}+1)}$ based on a previous labeling, so it is possible that $I$ comes from a $(\lambda+\frac{1}{n}, k')$-roundable solution and $I'$ comes from a $(\lambda, k')$-roundable solution, rather than both nested quasi-independent sets coming from $(\lambda, k')$-roundable solutions.

    First, we show that for the value of $k'$ at the end of the algorithm (which means all solutions found are $(\lambda', k')$-roundable for some $\lambda'$), that $\text{OPT}_{k'} = O(\text{OPT}_k)$. To see why this is the case, note that either $k' = k$, so the claim is trivial, or $k' = |I_1^{(0)}|$ for some $I_1^{(0)}$ that corresponds to a solution that was at some point labeled as $\mathcal{S}^{(0)}$. Note that the corresponding nested quasi-independent set $(I_1^{(0)}, I_2^{(0)}, I_3^{(0)})$ is not the final set $(I_1', I_2', I_3')$, because if so we would have stopped the algorithm before we decided to label $\mathcal{S}^{(0)}$ as such. Therefore, $k \le |I_1^{(0)}|+p_1 \cdot |I_2^{(0)} \cup I_3^{(0)}|$ and we are assuming that $k \le n-1.$ Finally, since $(I_1^{(0)}, I_2^{(0)}, I_3^{(0)})$ arises from a $(\lambda', k')$-roundable solution with no special facilities or bad clients (by Condition \ref{cond:no_special_or_bad}), we may apply Lemma \ref{lem:opt_k'_opt_k} to obtain that $\text{OPT}_{k'} = O(\text{OPT}_k)$.

    Note that $|I_1'| \ge |I_1|-1$, and that $|I_1'| \le |I_1'| + p_1|I_2' \cup I_3'| < k$, which means that $|I_1| < k+1$ so $|I_1| \le k$.
    First, suppose that $|I_1| \ge k - 100C^4$, where we recall that $C$ is an arbitrarily large but fixed constant. In this case, this means that $|I_1'| \ge k-100 C^4 - 1$ and $p_1 \cdot |I_2' \cup I_3'| \le 100C^4+1$ so $|I_2' \cup I_3'| = O(C^4).$ In this case, we can apply Lemma \ref{lem:expected_cost_S} to find a randomized set $I_1' \subset S \subset I_1' \cup I_2' \cup I_3'$ such that
\[\BE[\text{cost}(\mathcal{D}, S)] \le \rho(p_1) \cdot (1+O(\eps)) \cdot \left[\sum_{j \in \mathcal{D}} \alpha_j' - \left(\lambda-\frac{1}{n}\right) \cdot \BE[|S|]\right] + O(\gamma) \cdot \text{OPT}_{k'},\]
    where $\{\alpha_j'\}_{j \in \mathcal{D}}$ corresponds to the merged solution that produces $(I_1', I_2', I_3')$.
    Since $|I_2' \cup I_3'| = O(C^4)$, we can try every possible $I_1' \subset S \subset I_1' \cup I_2' \cup I_3'$ to get a deterministic set $S$ of size at most $|I_1'|+|I_2' \cup I_3'| \le k + O(C^4)$ and size at least $|I_1'| \ge k-O(C^4)$ such that
\begin{align*}
    \text{cost}(\mathcal{D}, S) &\le \rho(p_1) \cdot (1+O(\eps)) \cdot \left[\sum_{j \in \mathcal{D}} \alpha_j' - \left(\lambda-\frac{1}{n}\right) \cdot |S|\right] + O(\gamma) \cdot \text{OPT}_{k'} \\
    &\le \rho(p_1) \cdot (1+O(\eps)) \cdot \left[\sum_{j \in \mathcal{D}} \alpha_j' - \left(\lambda+\frac{2}{n}\right) \cdot |S|\right] + O(1/n) \cdot |S| + O(\gamma) \cdot \text{OPT}_{k'} \\
    &\le \rho(p_1) \cdot (1+O(\eps)) \cdot \text{OPT}_{|S|} + O(\gamma) \cdot \text{OPT}_{k}.
\end{align*}
    The final line follows by Proposition \ref{prop:duality_bound}, since $\text{OPT}_{k'} = O(\text{OPT}_k)$, and since $|S| \le k + O(C^4) \le O(k)$ so $O(1/n) \cdot |S| = O(1) \le O(\gamma) \cdot \text{OPT}_{k}$. Therefore, there exists an absolute constant $c_0$ such that we have a set of size at most $|I_1' \cup I_2' \cup I_3'| \le k+c_0 \cdot C^4$ with cost at most $\rho(p_1) \cdot (1+O(\eps)) \cdot \text{OPT}_{|S|} + O(\gamma) \cdot \text{OPT}_{k} \le \rho(p_1) \cdot (1+O(\eps)) \cdot \text{OPT}_{k-c_0 \cdot C^4}$. 
    As argued in the first paragraph of this proof, 
    this is sufficient since we can apply Lemma~\ref{lem:center-reduction}.
    

    Otherwise, namely when $|I_1| \le k-100 C^4$, we have $|I_2 \cup I_3| \ge 100 C^4$ since $|I_1|+p_1|I_2 \cup I_3| \ge k$.
    Then, recall that $|I_1 \backslash I_1'| \le 1,$ so let $\kappa = |I_1 \backslash I_1'| \in \{0, 1\}$.
    Set $t > 0$ and $c \ge 0$ such that $|I_1|+p_1|I_2 \cup I_3| = k + c \cdot t$ and $|I_1'|+p_1|I_2' \cup I_3'| = k-t.$ Then, $|I_1 \cap I_1'| = k-(1+d) t$ for some $d \ge 0$, so $|I_1| = k-(1+d)t + \kappa$. In this case, $p_1|I_2 \cup I_3| = (1+c+d) t-\kappa$, so if we set $p = p_1 \cdot \frac{(1+d)t-\kappa}{(1+c+d)t-\kappa}$, then $|I_1|+p |I_2 \cup I_3| = k$. Also, since $|I_2 \cup I_3| \ge 100 C^4$, we have that $(1+c+d) t \ge p_1 \cdot 100 C^4 \ge C$, so $p = p_1 \cdot \frac{1+d}{1+c+d} - \eta$ for some $\eta \le 1/C$.
    In this case, assuming that $p > 0.01,$ we can use Lemma \ref{lem:cost_bound} to obtain a solution of size at most $k$ with cost at most
\[\rho\left(p_1 \cdot \frac{1+d}{1+c+d}-\eta-\frac{2}{C}\right) \cdot \left(1+\frac{300}{C}\right) \cdot (1+O(\eps)) \cdot \left[\sum_{j \in \mathcal{D}} \alpha_j - \left(\lambda-\frac{1}{n}\right) \cdot k\right] + O(\gamma) \cdot \text{OPT}_{k'}.\]
    Now, since $p_1 \cdot \frac{1+d}{1+c+d} \le p_1 =.402$, it is straightforward to verify that $\rho$ has bounded derivative. (Indeed, each case produces a function that is continuously differentiable on $[0, 0.5)$, so has bounded derivative on $[0, 0.402]$.) Therefore, since $\eta \le \frac{1}{C}$, the solution in fact has cost at most
\[\rho\left(p_1 \cdot \frac{1+d}{1+c+d}\right) \cdot \left(1+O\left(\frac{1}{C}+\eps\right)\right) \cdot \left[\sum_{j \in \mathcal{D}} \alpha_j - \left(\lambda-\frac{1}{n}\right) \cdot k\right] + O(\gamma) \cdot \text{OPT}_{k'}.\]

    But since $\text{OPT}_{k'} = O(\text{OPT}_k)$, and since $$\sum_{j \in \mathcal{D}} \alpha_j - \left(\lambda-\frac{1}{n}\right) \cdot k = \sum_{j \in \mathcal{D}} \alpha_j - \left(\lambda+\frac{2}{n}\right) \cdot k + \frac{3k}{n} \le \text{OPT}_k + 3 \le (1+O(\gamma)) \cdot \text{OPT}_k,$$ we obtain a solution of cost at most
\begin{equation} \label{eq:main_bound_1}
    \rho\left(p_1 \cdot \frac{1+d}{1+c+d}\right) \cdot \left(1+O\left(\frac{1}{C}+\eps+\gamma\right)\right) \cdot \text{OPT}_k.
\end{equation}

    In addition, note that we can use Lemma \ref{lem:cost_bound} to obtain a solution $S$ for $(k+ct)$-means and $S'$ for $(k-t)$-means. Also, $|S \cup S'| = |S|+|S'|-|S \cap S'| \le |S|+|S'|-|I_1 \cap I_1'| = k+(c+d)t.$ So, if we define $\rho' := \max_{0 \le \eta \le 1/C} \rho(p_1-\eta-\frac{2}{C}) \cdot (1+\frac{300}{C}) \cdot (1+O(\eps)) = \rho(p_1) \cdot (1+O(\frac{1}{C}+\eps))$, then
\begin{equation} \label{eq:Sprime_bound}
    \text{cost}(\mathcal{D},S') \le \rho' \cdot \left(\sum_{j \in D} \alpha_j - \left(\lambda - \frac{1}{n}\right) \cdot (k-t)\right) + O(\gamma) \cdot \text{OPT}_{k'}
\end{equation}
    and
\begin{equation} \label{eq:S_Sprime_bound}
    \text{cost}(\mathcal{D}, S \cup S') \le \text{cost}(\mathcal{D}, S) \le \rho' \cdot \left(\sum_{j \in D} \alpha_j - \left(\lambda - \frac{1}{n}\right) \cdot (k+ct)\right) + O(\gamma) \cdot \text{OPT}_{k'}.
\end{equation}

    Therefore, by Proposition \ref{prop:submodular_3},
    if we randomly add $t$ of the items in $S \backslash S'$, we will get a set $S''$ of size $k$ with expected cost
\begin{align}
    \BE[\text{cost}(\mathcal{D}, S'')] &\le \rho' \cdot \Bigg(\frac{c+d}{1+c+d} \cdot \Bigg(\sum_{j \in D} \alpha_j - \left(\lambda - \frac{1}{n}\right) \cdot (k-t)\Bigg) \nonumber \\
    &\hspace{3.4cm}+ \frac{1}{1+c+d} \cdot \Bigg(\sum_{j \in D} \alpha_j - \left(\lambda - \frac{1}{n}\right) \cdot (k+ct)\Bigg)\Bigg) + O(\gamma) \cdot \text{OPT}_{k'} \nonumber \\
    &= \rho' \cdot \Bigg(\sum_{j \in D} \alpha_j - \left(\lambda - \frac{1}{n}\right) \cdot k + \frac{d}{1+c+d} \cdot \left(\lambda-\frac{1}{n}\right) \cdot t \Bigg) + O(\gamma) \cdot \text{OPT}_k. \label{eq:i_ran_out_of_names}
\end{align}

    Note that if $\lambda-\frac{1}{n} \le 0$, then this means that $\BE[\text{cost}(\mathcal{D}, S'')] \le \rho' \cdot \left(\sum_{j \in \mathcal{D}} \alpha_j\right) + O(\gamma) \cdot \text{OPT}_k,$ and Proposition \ref{prop:duality_bound} tells us that $\sum_{j \in \mathcal{D}} \alpha_j \le \text{OPT}_k + O(k/n) \le (1+O(\gamma)) \cdot \text{OPT}_k$, so the expected cost is at most $\rho(p_1) \cdot (1+O(\frac{1}{C}+\eps+\gamma)) \cdot \text{OPT}_k.$ Alternatively, we may assume that $\lambda-\frac{1}{n} > 0$.

    In this case, note that by Lemmas \ref{lem:RHS_positive_1/2} and \ref{lem:ugly_upper_bound}, we have that
\[\sum_{j \in \mathcal{D}} \alpha_j - \left(\lambda-\frac{1}{n}\right) \cdot \left(|I_1|+\frac{1}{2}|I_2 \cup I_3|\right) + 4\gamma \cdot \text{OPT}_{k'} \ge 0.\]
    We can rewrite $|I_1|+\frac{1}{2}|I_2 \cup I_3|$ as $(k-(1+d)t+\kappa)+\frac{1}{2p_1} \cdot ((1+c+d)t-\kappa) = k +t \cdot \left(\frac{1}{2p_1}(1+c+d) - (1+d)\right) - O(1) \ge k + t \cdot \left(\frac{1}{2p_1}(1+c+d)-(1+d)\right) \cdot \left(1-\frac{O(1)}{C^4}\right)$, where the last inequality is true since $t \cdot \left[\frac{1}{2p_1}(1+c+d)-(1+d)\right] \ge \Omega(t \cdot (1+c+d)) \ge \Omega(|I_2 \cup I_3|) \ge 100 C^4$. Thus, we have that
\[\sum_{j \in \mathcal{D}} \alpha_j - \left(\lambda-\frac{1}{n}\right) \cdot k + 4 \gamma \cdot \text{OPT}_{k'} \ge \left(\lambda-\frac{1}{n}\right) \cdot t \cdot \left(\frac{1}{2p_1} (1+c+d)-(1+d)\right) \cdot \left(1-\frac{O(1)}{C^4}\right).\]
    Therefore, combining the above equation with \eqref{eq:i_ran_out_of_names}, we have that
\begin{align}
    \BE[\text{cost}(\mathcal{D}, S'')] &\le \rho' \cdot \left(\sum_{j \in \mathcal{D}} \alpha_j - \left(\lambda-\frac{1}{n}\right) \cdot k + O(\gamma) \cdot \text{OPT}_{k'}\right) \cdot \left(1 + \frac{\frac{d}{1+c+d}}{\frac{1}{2p_1}(1+c+d)-(1+d)} \cdot \left(1+\frac{O(1)}{C^4}\right)\right) \nonumber \\
    &\le \rho' \cdot \left(\sum_{j \in \mathcal{D}} \alpha_j - \left(\lambda+\frac{2}{n}\right) \cdot k + O(\gamma) \cdot \text{OPT}_k\right) \cdot \left(1 + \frac{\frac{d}{1+c+d}}{\frac{1}{2p_1}(1+c+d)-(1+d)} \cdot \left(1+\frac{O(1)}{C^4}\right)\right) \nonumber \\
    &\le \rho' \cdot \left(1+O(\frac{1}{C}+\gamma)\right) \cdot \left(1 + \frac{\frac{d}{1+c+d}}{\frac{1}{2p_1}(1+c+d)-(1+d)}\right) \cdot \text{OPT}_k. \label{eq:main_bound_2}
\end{align}
    
    If we set $r \ge 1$ such that $c = (r-1)(1+d)$, then $\frac{1+d}{1+c+d} = \frac{1}{r}$ and $$\frac{\frac{d}{1+c+d}}{\frac{1}{2p_1}(1+c+d)-(1+d)} = \frac{d}{r(1+d)} \cdot \frac{1}{(\frac{r}{2p_1}-1) \cdot (1+d)} = \frac{d}{(1+d)^2} \cdot \frac{1}{r \cdot (\frac{r}{2p_1}-1)} \le \frac{1}{4r \cdot (\frac{r}{2p_1}-1)}.$$ So, 
    by combining Equations \eqref{eq:main_bound_1} and \eqref{eq:main_bound_2},
    we can always guarantee an approximation factor of at most
\begin{equation} \label{eq:main_bound}
    \left(1+O(\frac{1}{C}+\eps+\gamma)\right) \cdot \max_{r \ge 1} \min\left(\rho\left(\frac{p_1}{r}\right), \rho(p_1) \cdot \left(1+\frac{1}{4r \cdot \left(\frac{r}{2p_1}-1\right)}\right)\right).
\end{equation}
    This approximation factor also holds in the case when $\lambda-\frac{1}{n} \le 0$, by setting $r = 1$. So, by letting $C$ be an arbitrarily large constant and $\gamma \ll \eps \ll 1$ be arbitrarily small constants, the result follows.
\end{proof}

    Since we have set $p_1 = 0.402$ and $\delta_1 = \frac{4+8\sqrt{2}}{7}, \delta_2 = 2, \delta_3 = 0.265$, by Proposition \ref{prop:more_bash} we have $\rho(p_1) = 3+2\sqrt{2}$. If $r \ge 3.221,$ one can verify that $(3+2\sqrt{2}) \cdot \left(1+\frac{1}{4r \cdot (\frac{r}{2p_1} - 1)}\right) \le 5.979$. Alternatively, if $1 \le r \le 3.221,$ then $\frac{p_1}{r} \ge \frac{.402}{3.221} \ge .1248,$ and it is straightfoward to verify that $\rho(p) \le 5.979$ for all $p \in [.1248, .402]$ using Proposition \ref{prop:more_bash}.\footnote{We remark that while Proposition \ref{prop:more_bash} follows from Lemma \ref{lem:more_bash}, Lemma \ref{lem:more_bash} only depends on Lemma \ref{lem:main_lmp}, so there is no circular reasoning.} Hence, we have a $\boxed{5.979}$-approximation to $k$-means.

\subsection{Improving the approximation further} \label{subsec:k_means_improved}

First, we define some important quantities. For any client $j$, we define $A_j := \alpha_j - \sum_{i \in \bar{N}(j) \cap I_1} (\alpha_j-c(j, i))$, and we define $B_j := \sum_{i \in \bar{N}(j) \cap (I_2 \cup I_3)} (\alpha_j-c(j, i))$. Note that $B_j \ge 0$ always.

We split the clients into $5$ groups.
Let $\mathcal{D}_1$ be the set of all clients $j \not\in \mathcal{D}_B$ corresponding to subcases \ref{eq:1.a}, \ref{eq:1.c}, \ref{eq:1.d}, \ref{eq:1.g.ii}, 1.h, \ref{eq:2.a}, \ref{eq:3.a}, \ref{eq:4.a.i}, \ref{eq:4.b.i}, and \ref{eq:4.c}, as well as the clients in \ref{eq:5.a} where there do not exist $i_2 \in \bar{N}(j) \cap I_2$ and $i_3 \in \bar{N}(j) \cap I_3$ such that $q(i_3) = i_2$. (In the casework, our choice of $a$ is $|\bar{N}(j) \cap S|$ rather than $N(j) \cap S$, similar for $b$ and $c$.) Let $Q_1$ be the sum of $A_j$ for these clients, and $R_1$ be the sum of $B_j$ for these clients. Next, let $\mathcal{D}_2$ be the set of all clients $j \not\in \mathcal{D}_B$ corresponding to subcases \ref{eq:1.b}, \ref{eq:1.e}, \ref{eq:4.a.ii}, and \ref{eq:4.b.ii}. Let $\mathcal{D}_3$ be the set of all clients $j \not\in \mathcal{D}_B$ corresponding to subcase \ref{eq:1.g.i}. Let $\mathcal{D}_4$ be the set of all clients $j \not\in \mathcal{D}_B$ corresponding to subcase \ref{eq:2.d}, further restricted to $c(j, i_2) + c(j, i_3) \ge 0.25 \cdot \alpha_j$ (or equivalently in the language of Case \ref{eq:2.d} in Lemma \ref{lem:main_lmp}, $\beta^2+\gamma^2 \ge 0.25$). Finally, let $\mathcal{D}_5$ be the set of all bad clients $j \in \mathcal{D}_B$, as well as all remaining subcases (\ref{eq:2.b}, \ref{eq:2.c}, \ref{eq:2.d} when $\beta^2+\gamma^2 < 0.25$, \ref{eq:3.b}, \ref{eq:3.c}, and the clients in \ref{eq:5.a} not covered by $\mathcal{D}_1$). Note these cover all cases (recall that 1.f is a non-existent case). Finally, we define $Q_2, Q_3, Q_4, Q_5, R_2, R_3, R_4, R_5$ similarly to how we defined $Q_1$ and $R_1.$


Now, we have the following result, which improves over Lemma \ref{lem:RHS_positive_1/2}.

\begin{lemma} \label{lem:RHS_positive_1/2_improved}
    For any client $j$, $A_j-\frac{1}{2} B_j \ge 0$. In addition, if the client $j$ corresponds to any of the subcases in case $1$ or case $4$, or to subcases $2.a$ or $3.a$, then $A_j-B_j \ge 0$. Also, if the client $j$ corresponds to subcase $2.d$ where $\beta^2+\gamma^2 \ge 0.25$, then $A_j-\frac{4}{7}B_j \ge 0$.
\end{lemma}

\begin{remark}
    As in Lemma \ref{lem:RHS_positive_1/2}, this lemma holds even for bad clients $j \in \mathcal{D}_B$.
\end{remark}

\begin{proof}
    The proof that $A_j-\frac{1}{2} B_j \ge 0$ for any client $j$ is identical to that of Lemma \ref{lem:RHS_positive_1/2}. So, we focus on the next two claims.
    For the subcases in Case 1, note that $\bar{N}(j) \cap I_3 = \emptyset$, so we just need to show that $\alpha_j - \sum_{i \in \bar{N}(j) \cap I_1} (\alpha_j-c(j, i)) \ge \sum_{i \in \bar{N}(j) \cap I_2} (\alpha_j-c(j, i))$, which is implied by Equation \eqref{eq:I2_bound}.
    Likewise, for the subcases in Case 4, note that $\bar{N}(j) \cap I_2 = \emptyset$, so we just need to show that $\alpha_j - \sum_{i \in \bar{N}(j) \cap I_1} (\alpha_j-c(j, i)) \ge \sum_{i \in \bar{N}(j) \cap I_3} (\alpha_j-c(j, i))$, which is implied by Equation \eqref{eq:I3_bound}.
    
    We recall that in subcases 2.a and 3.a, we noted in both cases that the points in $\bar{N}(j) \cap I_2$ and $\bar{N}(j) \cap I_3$ were all separated in $H(\delta_2)$, and that $\bar{N}(j) \cap I_1 = \emptyset$ since $a = 0$.
    So, we have that $\alpha_j-\sum_{i \in \bar{N}(j) \cap I_1} (\alpha_j-c(j, i)) = \alpha_j \ge \sum_{i \in \bar{N}(j) \cap (I_2 \cup I_3)} (\alpha_j-c(j, i))$ in both cases.
    
    Finally, we consider subcase 2.d when $\beta^2+\gamma^2 \ge 0.25.$ In this case, we have (when $\alpha_j = 1$) that $1-\sum_{i \in \bar{N}(j) \cap I_1} = 1$, and $\sum_{i \in \bar{N}(j) \cap (I_2 \cup I_3)} (1-c(j, i)) = (1-c(j, i_2))+(1-c(j, i_3)) = 2-\beta^2-\gamma^2$, which is at most $1.75$ if $\beta^2+\gamma^2 \ge 0.25.$ So, for general $\alpha_j$, we have that $\alpha_j-\sum_{i \in \bar{N}(j) \cap I_1} (\alpha_j-c(j, i)) = \alpha_j \ge \frac{4}{7} \cdot \sum_{i \in \bar{N}(j) \cap (I_2 \cup I_3)} (\alpha_j-c(j, i)).$
\end{proof}

Therefore, we have that
\begin{equation} \label{eq:S_i_bound}
    R_1 \le Q_1, \hspace{1cm} R_2 \le Q_2, \hspace{1cm} R_3 \le Q_3, \hspace{1cm} R_4 \le 1.75 \cdot Q_4, \hspace{0.5cm} \text{and} \hspace{0.5cm} R_5 \le 2 Q_5.
\end{equation}

Next, we define $\rho^{(i)}(p)$ to be the maximum fraction $\rho(p)$ obtained from the casework corresponding to the (not bad) clients in $\mathcal{D}_i$. (Note that $\rho(p) = \max\left(\rho^{(1)}(p), \rho^{(2)}(p), \rho^{(3)}(p), \rho^{(4)}(p), \rho^{(5)}(p)\right)$.) We have the following result:

\begin{lemma} \label{lem:more_bash}
    Let $\delta_1 = \frac{4+8\sqrt{2}}{7},$ $\delta_2 = 2$, and $\delta_3 = 0.265$. Then, for all $p \in [0.096, 0.402],$ we have that $\rho^{(1)}(p) \le 3+2\sqrt{2}$, $\rho^{(2)}(p) \le 1+2p+(1-p)\cdot \delta_1+2\sqrt{2p^2 +(1-p)\cdot \delta_1}$, and $\rho^{(5)}(p) \le 5.68$.
\end{lemma}

\begin{proof}
    We start by considering $\rho^{(1)}(p)$, covered by subcases \ref{eq:1.a}, \ref{eq:1.c}, \ref{eq:1.d}, \ref{eq:1.g.ii}, 1.h, \ref{eq:2.a}, \ref{eq:3.a}, \ref{eq:4.a.i}, \ref{eq:4.b.i}, and \ref{eq:4.c}, and certain subcases of \ref{eq:5.a}. All subcases except \ref{eq:2.a}, \ref{eq:4.c}, and \ref{eq:5.a} can easily be verified (see our Desmos file for K-means Case 1, the link is in Appendix \ref{app:files}). For subcase \ref{eq:2.a}, we have to verify it for all choices of $c \ge 1$. However, it is simple to see that the numerator of the fraction decreases as $c_2$ increases whenever $p \in [0, 0.5]$, so in fact we just have to verify it for $c = c_2 = 1$, which is straightforward. For subcase \ref{eq:4.c}, we have to verify it for all choices of $c \ge 2$. For $c = 2$ it is straightforward to verify. For $c \ge 3,$ since $2.5+\sqrt{2} \le 3+2\sqrt{2},$ it suffices to show
\[\frac{\left(\frac{1}{2}(1-2p)+\frac{1}{2}(1-2p)^c\right) \cdot (1+\sqrt{\delta_1})^2 - \left(\frac{1}{2} + \frac{1}{2}(1-2p)^c\right) +p(1+\sqrt{2})^2+1}{1-p} \le 3+2\sqrt{2},\]
    where we have taken the fraction from \ref{eq:4.c} and added back a $\frac{p}{c}$ term to the numerator. Now, this fraction is decreasing as $c$ increases, so it suffices to verify it for $c = 3$, which is straightforward.
    
    The last case for $\rho^{(1)}(p)$ is Case \ref{eq:5.a}. We show that in all cases the fraction is bounded by $3+2\sqrt{2}$ for $p \in [0.096, 0.402]$, and if $h \ge 1$ then the fraction can further be bounded by $5.68$. This is clearly sufficient for bounding $\rho^{(1)}(p)$. It will also be important in bounding $\rho^{(5)}(p)$ - indeed, if there exist $i_2 \in \bar{N}(j) \cap I_2$ and $i_3 \in \bar{N}(j) \cap I_3$ such that $q(i_3) = i_2,$ then regardless of the outcomes of the initial fair coins, $h \ge 1$ since exactly one of $i_3$ or $q(i_3) = i_2$ will contribute to the value of $h$.
    
    First, we note that $T_1-T_1^2+T_3$ can be rewritten as
\[a+2ph-(a^2+4aph+4p^2h^2) + \frac{\delta_1}{2} \cdot (a^2-a) + 4pah + 4p^2 \cdot h(h-1) = \left(\frac{\delta_1}{2}-1\right) \cdot a(a-1) + 2p(1-2p) \cdot h.\]
    In the case where $a=1$ and $h \ge 1$, we can therefore simplify the fraction in \eqref{eq:5.a} to $\frac{1}{a+2ph} + \frac{2ph}{2p(1-2p) \cdot h} = \frac{1}{a+2ph} + \frac{1}{1-2p} \le \frac{1}{1+2p}+\frac{1}{1-2p} = \frac{2}{1-4p^2}.$ This is at most $5.68$ for any $p \le 0.402$.
    When $a \ge 2$, we can write the fraction as 
\begin{equation} \label{eq:5a_simplified}
    \frac{1}{a+2ph} + \frac{(a-1)+(2p)h}{\left(\frac{\delta_1}{2}-1\right) \cdot a(a-1) + (2p)h \cdot (1-2p)}.
\end{equation}
    When $a \ge 2$ and $h = 0$, \eqref{eq:5a_simplified} can be simplified as
\[\frac{1}{a} + \frac{1}{a \cdot \left(\frac{\delta_1}{2}-1\right)} \le  \frac{1}{2} \left(1 + \frac{1}{\frac{\delta_1}{2}-1}\right) = 3+2\sqrt{2}.\]
\begin{align*}
    \frac{1}{a+2ph} + \frac{(a-1)+(2p)h}{\left(\frac{\delta_1}{2}-1\right) \cdot a(a-1) + (2p)h \cdot (1-2p)} &= \frac{1}{2+2ph} + \frac{[1+2p]+2p(h-1)}{\left[\left(\frac{\delta_1}{2}-1\right) \cdot 2+2p\right] + 2p(h-1) \cdot (1-2p)}\\
    &\le \frac{1}{a+2p} + \max\left(\frac{1}{\left(\frac{\delta_1}{2}-1\right) \cdot a}, \frac{1}{1-2p}\right).
\end{align*}
    When $a = 2$ and $h \ge 1$, we can rewrite \eqref{eq:5a_simplified} as
\begin{align*}
    \frac{1}{2+2ph} + \frac{1+2ph}{\delta_1-2 + 2ph(1-2p)} &= \frac{1}{2+2ph} + \frac{1+2p + 2p(h-1)}{\delta_1-2 + 2p(1-2p) + 2p(1-2p)(h-1)} \\
    &\le \frac{1}{2+2p} + \max\left(\frac{1+2p}{\delta_1-2+2p(1-2p)}, \frac{1}{1-2p}\right),
\end{align*}
    which is easily verifiable to be at most $5.68$ for $p\in [0.096, 0.402]$.
    When $a \ge 3$ and $h \ge 1$, \eqref{eq:5a_simplified} is at most
\[\frac{1}{3} + \max\left(\frac{1}{3 \cdot \left(\frac{\delta_1}{2}-1\right)}, \frac{1}{1-2p}\right),\]
    which is easily verifiable to be at most $5.68$ for $p\in [0.096, 0.402]$.
    The final case is when $a = 1, h = 0$, but here we saw in our analysis of \ref{eq:5.a} that the fraction is at most $1$, or that the numerator and denominator are both $0$.
    
    Next, consider $\rho^{(2)}(p)$, which is covered by subcases \ref{eq:1.b}, \ref{eq:1.e}, \ref{eq:4.a.ii}, and \ref{eq:4.b.ii}. Indeed, since $\delta_2 = 2$, these all have the exact same bound of $1+2p+(1-p)\delta_1 + 2\sqrt{2p^2+(1-p)\delta_1}$.
    
    Finally, we deal with $\rho^{(5)}(p)$, which deals with subcases \ref{eq:2.b}, \ref{eq:2.c}, \ref{eq:3.b}, \ref{eq:3.c}, and \ref{eq:5.a}, along with \ref{eq:2.d} when $\beta^2+\gamma^2 < 0.25$. 
    
    Subcase \ref{eq:2.b} can be easily verified to be at most $5.664$ in the range $p \in [0.096, 0.402]$ when $c_2 = 0$ and $c = c_1$ is between $2$ and $5$. Beyond this, we assume that $c_1 \ge 6$, so we can apply the crude bound that the fraction is at most
\[\frac{\frac{1}{2}+\left(1-p-\frac{1}{2}\cdot\left(1-(1-2p)^{6}\right)\right)\cdot\left(1+\sqrt{\delta_1}\right)^{2}}{1-2p},\]
    which is at most $5.68$ for $p \in [0.096, 0.402]$. It is easy to verify that the fraction in Subcase \ref{eq:2.c} is at most $5.68$ for $p \in [0.096, 0.402]$. 
    
    Subcase \ref{eq:3.b} is easy to verify for $2 \le b \le 5.$ For $b \ge 6$, we can apply the crude bound that the fraction is at most
\[\frac{(1-p)^{5}\cdot(1-2p)\cdot\left(1+\sqrt{\delta_1}\right)^2+1}{1-\left(1+\frac{2-\delta_3}{7}\right)p},\]
    which trivially satisfies the desired bounds. Finally, we note that in Subcase \ref{eq:3.c}, the fraction decreases as $c_1$ and $c_2$ increase, so we may assume that either $c_1=c_2 = 1$ or $c_1 = 2$ and $c_2 = 0$. These are easy to verify for $2 \le b \le 5$, and for $b \ge 6$, we may apply a crude bound to say the fraction is at most
\[\frac{1+(1-p)^5\cdot (1-2p) \cdot \left((1+\sqrt{\delta_1})^2-1\right)}{1-2p}\]
    as long as $c_1 \ge 1$ and $c_2 \ge 0.$ This is at most $5.68$ in the range $[0.096, 0.402]$.
    
    Subcase \ref{eq:5.a} was dealt with previously (as we only have to consider when $h \ge 1$), so the final case is \ref{eq:2.d} when $\beta^2+\gamma^2 < 0.25$. In this case, we recall the fraction is
\[\frac{(1-2p) \cdot \min(1+\sqrt{\delta_1}, \max(\beta, \gamma)+\sqrt{\delta_1 \cdot t})^2 + p \cdot (\beta^2+\gamma^2)}{1-(2-\beta^2-\gamma^2) \cdot p},\]
    where $t \ge 1$, $\beta+\gamma \ge \sqrt{\delta_3 \cdot t}$, and also $\beta^2+\gamma^2 \le 0.25$. By the symmetry of $\beta$ and $\gamma$, we may replace $\max(\beta, \gamma)$ with $\beta$. So, by defining $\zeta = \beta^2+\gamma^2,$ we can upperbound the above expression by
\[\frac{(1-2p) \cdot (\beta+\sqrt{\delta_1 \cdot t})^2 + p \cdot \zeta}{1-2p + p \cdot \zeta} \le \frac{(1-2p) \cdot (\beta+\sqrt{\delta_1/\delta_3} \cdot(\beta+\gamma))^2 + p \cdot \zeta}{1-2p + p \cdot \zeta},\]
    since $\zeta \le 0.25$ and since $\beta+\gamma \ge \sqrt{\delta_3 \cdot t}$. By Cauchy-Schwarz, $\left(\beta \cdot x + \gamma \cdot y\right)^2 \le (\beta^2+\gamma^2) \cdot(x^2+y^2) \le \zeta \cdot (x^2+y^2)$. So, we can bound the above expression by
\[\frac{(1-2p) \cdot \zeta \cdot \left((1+\sqrt{\delta_1/\delta_3})^2 + \delta_1/\delta_3\right) + p \cdot \zeta}{1-2p + p \cdot \zeta} = \frac{(1-2p) \cdot \left((1+\sqrt{\delta_1/\delta_3})^2 + \delta_1/\delta_3\right) + p}{\frac{1-2p}{\zeta} + p}.\]
    For $p \le 0.5,$ this fraction clearly increases with $\zeta$, so we maximize this when $\zeta = 0.25$. When setting $\zeta = 0.25$, this can easily be verified to be at most $5.68$ for all $p \in [0.096, 0.5]$.
    
    This concludes all cases, thus proving the proposition.
\end{proof}


Next, we recall Lemma \ref{lem:ugly_upper_bound}. First, by setting $S$ to be $I_1$ in Lemma \ref{lem:ugly_upper_bound}, we obtain that 
\begin{equation} \label{eq:Qi_sum_upper_bound}
    \sum_{i = 1}^{5} Q_i = \sum_{j \in \mathcal{D}} \left(\alpha_j - \sum_{i \in \bar{N}(j) \cap I_1} (\alpha_j-c(j, i))\right) \le \sum_{j \in \mathcal{D}} \alpha_j - \left(\lambda-\frac{1}{n}\right) \cdot |I_1| + 4\gamma \cdot \text{OPT}_{k'}.
\end{equation}
Next, by setting $S$ to be $I_2 \cup I_3$ in Lemma \ref{lem:ugly_upper_bound}, we obtain that
\begin{equation} \label{eq:Ri_sum_lower_bound}
    \sum_{i = 1}^{5} R_i = \sum_{j \in \mathcal{D}} \sum_{i \in \bar{N}(j) \cap (I_2 \cup I_3)} (\alpha_j-c(j, i)) \ge \left(\lambda-\frac{1}{n}\right) \cdot |I_2 \cup I_3| - 4\gamma \cdot \text{OPT}_{k'}.
\end{equation}

Next, we recall Lemma \ref{lem:expected_cost_S}. By splitting $\BE[\text{cost}(\mathcal{D}, S)]$ based on whether $j$ is in $\mathcal{D}_1$, $\mathcal{D}_2$, $\mathcal{D}_3$, $\mathcal{D}_4$, $\mathcal{D}_5 \backslash \mathcal{D}_B$, or $\mathcal{D}_B$, we obtain that
\[\BE[\text{cost}(\mathcal{D}, S)] \le (1+O(\eps)) \cdot \left[\sum_{i = 1}^{5} \rho^{(i)}(p) \cdot (Q_i - p \cdot R_i)\right] + O(\gamma) \cdot \text{OPT}_{k'}.\]
Therefore, the argument of Lemma \ref{lem:cost_bound}
implies that if $|I_1|+p \cdot |I_2 \cup I_3| = k$, if $p \in [0.01, 0.49]$, and if $|I_2 \cup I_3| \ge 100 C^4$, then we can choose a set $I_1 \subset S \subset I_1 \cup I_2 \cup I_3$ such that $|S| \le k$ and 
\begin{align}
    \text{cost}(\mathcal{D}, S) &\le \left(1+O(\eps+\frac{1}{C})\right) \cdot \left[\sum_{i = 1}^{5} \rho^{(i)}\left(p-\frac{2}{C}\right) \cdot \left(Q_i-\left(p-\frac{2}{C}\right) \cdot R_i\right)\right] + O(\gamma) \cdot \text{OPT}_{k'} \nonumber \\
    &\le \left(1+O(\eps+\frac{1}{C})\right) \cdot \left[\sum_{i = 1}^{5} \rho^{(i)}\left(p\right) \cdot(Q_i-p \cdot R_i)\right] + O(\gamma) \cdot \text{OPT}_{k}. \label{eq:cost_bound_improved}
\end{align}
    To explain the second line, note that $\rho^{(i)}$ has bounded derivative on $[0.01, 0.49]$ and that $Q_i \ge 0.5 \cdot R_i$. Therefore, since $p \in [0.01, 0.49]$, $\rho^{(i)}\left(p-\frac{2}{C}\right) = \rho^{(i)}(p) \cdot \left(1+O(1/C)\right)$, and $Q_i-p \cdot R_i = \Omega(R_i)$ which means $Q_i-\left(p-\frac{2}{C}\right) \cdot R_i = (Q_i-p \cdot R_i) \cdot \left(1+O(1/C)\right)$. In addition, we still have that $\text{OPT}_{k'} = O(\text{OPT}_k)$, as in our proof of Theorem \ref{thm:main}.

We now return to the setup of Theorem \ref{thm:main}, where we have $(I_1, I_2, I_3)$ and $(I_1', I_2', I_3')$. Suppose that $|I_1|+p_1|I_2 \cup I_3| = k+ct$, $|I_1'|+p_1|I_2' \cup I_3'| = k-t$, $|I_1 \cap I_1'| = k-(1+d) t$, and $|I_1 \backslash I_1'| = \kappa \in \{0, 1\}$. In addition, suppose that $|I_1| \ge k-100 C^4,$ which means that $|I_1'| \ge k-100 C^4-1$. In this case, we may follow the same approach as in our Theorem \ref{thm:main} to obtain a $\rho(p_1) \cdot (1+O(\eps))$-approximation to $k$-means.

Alternatively, we may suppose that $|I_1| \le k-100 C^4$, which implies that $|I_2 \cup I_3| \ge 100 C^4$. Then, defining $r \ge 1$ such that $c=(r-1) \cdot (1+d)$, we can use Equation \eqref{eq:cost_bound_improved} to find a solution of size at most $k$ with cost at most
\begin{equation} \label{eq:cost1_improved}
    \left(1+O(\eps+\frac{1}{C})\right) \cdot \left[\sum_{i = 1}^{5} \rho^{(i)}\left(\frac{p_1}{r}\right) \cdot \left(Q_i - \frac{p_1}{r} \cdot R_i\right)\right] + O(\gamma) \cdot \text{OPT}_{k},
\end{equation}
in the same manner as \eqref{eq:cost_bound_improved}, by setting $p = p_1 \cdot \frac{(1+d)t-\kappa}{(1+c+d)t-\kappa} = \frac{p_1}{r} - O(1/C)$.
Alternatively, we can obtain two separate solutions $I_1 \subset S \subset I_1 \cup I_2 \cup I_3$ of size $k+ct$, and a solution $I_1' \subset S' \subset I_1' \cup I_2' \cup I_3'$ of size $k-t$, such that $|S \cup S'| = k+(c+d) t$. We have that
\[\text{cost}(\mathcal{D}, S \cup S') \le \text{cost}(\mathcal{D}, S) \le \left(1+O(\eps+\frac{1}{C})\right) \cdot \left[\sum_{i = 1}^{5} \rho^{(i)}\left(p_1\right) \cdot (Q_i-p_1 \cdot R_i)\right]+O(\gamma) \cdot \text{OPT}_{k}.\]
Finally, using the bound \eqref{eq:Sprime_bound} for the cost of $S'$, we have
\[\text{cost}(\mathcal{D}, S') \le \left(1+O(\eps+\frac{1}{C})\right) \cdot \rho(p_1) \cdot \left(\sum_{j \in \mathcal{D}} \alpha_j - \left(\lambda-\frac{1}{n}\right) \cdot (k-t)\right) + O(\gamma) \cdot \text{OPT}_{k}.\]
Note that we are not able to use a more sophisticated bound for $\text{cost}(\mathcal{D}, S')$, because our values of $\{Q_i\}$ and $\{R_i\}$ only apply to $(I_1, I_2, I_3)$ and not to $(I_1', I_2', I_3')$. By combining the solutions $S \cup S'$ and $S'$, by adding $t$ random points from $S \backslash S'$ to $S'$, and using Proposition \ref{prop:submodular_3}, we obtain a solution $S''$ with expected cost
\begin{multline}
    \BE[\text{cost}(\mathcal{D}, S'')] \le \left(1+O(\eps+\frac{1}{C})\right) \cdot \Biggr[\frac{1}{r(1+d)} \cdot \sum_{i = 1}^5 \rho^{(i)}(p_1) \cdot (Q_i-p_1 \cdot R_i) \\
    + \left(1-\frac{1}{r(1+d)}\right) \cdot \rho(p_1) \cdot \Biggr(\sum_{j \in \mathcal{D}} \alpha_j - \left(\lambda-\frac{1}{n}\right) \cdot (k-t)\Biggr)\Biggr] + O(\gamma) \cdot \text{OPT}_k. \label{eq:cost2_improved}
\end{multline}
    This is because we combine the solution $S \cup S'$, which has size $k+(c+d)t$, with the solution $S'$, which has size $k-t,$ so we assign the first solution relative weight $\frac{1}{1+c+d} = \frac{1}{r(1+d)}$ and the second solution relative weight $\frac{c+d}{1+c+d} = 1-\frac{1}{r(1+d)}$.

Now, let $\mathfrak{D}$ equal $\sum_{j \in \mathcal{D}} \alpha_j - \left(\lambda-\frac{1}{n}\right) \cdot k$. Then, since $|I_1|+\frac{p_1}{r} |I_2 \cup I_3| \ge k$, we can combine Equations \eqref{eq:Qi_sum_upper_bound} and \eqref{eq:Ri_sum_lower_bound} to get that
\begin{align}
    \sum_{i = 1}^{5} \left(Q_i-\frac{p_1}{r} R_i\right) &\le \sum_{j \in \mathcal{D}} \alpha_j - \left(\lambda-\frac{1}{n}\right) \cdot \left(|I_1|+\frac{p_1}{r} |I_2 \cup I_3|\right) + O(\gamma) \cdot \text{OPT}_k \nonumber \\
    &\le \mathfrak{D} + O(\gamma) \cdot \text{OPT}_k. \label{eq:bound_1}
\end{align}
Next, recall (by Equation \eqref{eq:cost1_improved}) that we have a solution of size at most $k$ with cost at most
\begin{equation} \label{eq:bound_1.5}
    \left(1+O(\eps+\frac{1}{C})\right) \cdot \left[\sum_{i = 1}^{5} \rho^{(i)}\left(\frac{p_1}{r}\right) \cdot \left(Q_i - \frac{p_1}{r} \cdot R_i\right)\right] + O(\gamma) \cdot \text{OPT}_k.
\end{equation}
Finally, we note that since $|I_2 \cup I_3| = \frac{r(1+d)t-\kappa}{p_1} \ge (1-O(1/C)) \cdot \frac{r(1+d) t}{p_1},$ we have that 
\[\sum_{i = 1}^{5} R_i + O(\gamma) \cdot \text{OPT}_k \ge \left(\lambda-\frac{1}{n}\right) \cdot |I_2 \cup I_3| \ge \left(1-O(\frac{1}{C})\right) \cdot \frac{r(1+d)}{p_1} \cdot \left(\lambda-\frac{1}{n}\right) \cdot t.\]
Therefore, we can bound the expected cost of $S''$ by
\begin{multline}
    \left(1+O(\eps+\frac{1}{C})\right) \cdot \Biggr[\frac{1}{r(1+d)} \cdot \sum_{i = 1}^{5} \rho^{(i)}(p_1) \cdot (Q_i-p_1 \cdot R_i) \\
    + \left(1-\frac{1}{r(1+d)}\right) \cdot \rho(p_1) \cdot \left(\mathfrak{D}+\frac{p_1}{r(1+d)} \cdot \sum_{i = 1}^{5} R_i\right)\Biggr] + O(\gamma) \cdot \text{OPT}_k. \label{eq:bound_2}
\end{multline}


Now, we have that $r \ge 1$, and if we let $\theta = \frac{1}{1+d}$, we have that $\theta \in [0, 1]$. Hence, to show that we obtain an approximation $\rho + O(\eps + \gamma + 1/C)$, it suffices to show that for all choices of $\theta \in [0, 1]$ and $r \ge 1,$ that if we let $\mathfrak{D}' = \mathfrak{D}+O(\gamma) \cdot \text{OPT}_k$, one cannot simultaneously satisfy
\begin{align}
    \mathfrak{D}' &\ge \sum_{i = 1}^{5} \left(Q_i-\frac{p_1}{r} R_i\right) \label{eq:bound_3}\\
    \rho \cdot \mathfrak{D}' &< \frac{\theta}{r} \sum_{i = 1}^{5} \rho^{(i)}(p_1) \cdot (Q_i-p_1 \cdot R_i) + \left(1-\frac{\theta}{r}\right) \cdot \rho(p_1) \cdot \left(\mathfrak{D}' + p_1 \cdot \frac{\theta}{r} \sum_{i = 1}^{5} R_i\right) \label{eq:bound_4} \\
    \rho \cdot \mathfrak{D}' &< \sum_{i = 1}^{5} \rho^{(i)}\left(\frac{p_1}{r}\right) \cdot \left(Q_i-\frac{p_1}{r} \cdot R_i\right) \label{eq:bound_5}
\end{align}
and
\begin{equation}
    R_1 \le Q_1, \hspace{1cm} R_2 \le Q_2, \hspace{1cm} R_3 \le Q_3, \hspace{1cm} R_4 \le 1.75 Q_4, \hspace{1cm} R_5 \le 2 Q_5. \label{eq:bound_6}
\end{equation}
    Indeed, we already know that \eqref{eq:bound_3} is true (same as \eqref{eq:bound_1}) and that \eqref{eq:bound_6} is true (same as \eqref{eq:S_i_bound}). So if we can show we can't simultaneously satisfy all of \eqref{eq:bound_3}, \eqref{eq:bound_4}, \eqref{eq:bound_5}, and \eqref{eq:bound_6}, then either \eqref{eq:bound_4} or \eqref{eq:bound_5} is false. But we have a clustering with at most $k$ centers and cost at most the right hand sides of each of \eqref{eq:bound_4} and \eqref{eq:bound_5} up to a $1+O(1/C + \eps + \gamma)$ multiplicative factor, due to \eqref{eq:bound_2} and \eqref{eq:bound_1.5}, respectively. Therefore, we successfully obtain a solution of cost at most $\rho \cdot \left(1 + O(1/C + \eps + \gamma)\right) \cdot \mathfrak{D}'$. Moreover, $\mathfrak{D}' \le \sum_{j \in \mathcal{D}} \alpha_j - \left(\lambda-\frac{1}{n}\right) \cdot k + O(\gamma) \cdot \text{OPT}_k \le (1+O(\gamma)) \cdot \text{OPT}_k$, since $\sum_{j \in \mathcal{D}} \alpha_j - \left(\lambda+\frac{2}{n}\right) \cdot k \le \text{OPT}_k$ by Proposition \ref{prop:duality_bound} as both $\alpha^{(\ell)}, \alpha^{(\ell+1)}$ are solutions to $\text{DUAL}(\lambda+\frac{1}{n})$, and since $\frac{3k}{n} \le 3 \le O(\gamma) \cdot \text{OPT}_k$. Therefore, $\mathfrak{D}' \le (1+O(\gamma)) \cdot \text{OPT}_k$, which means that we have found a $\rho \cdot (1+O(1/C+\eps+\gamma))$ approximation to $k$-means clustering.

    Indeed, by numerical analysis of these linear constraints and based on the functions $\rho^{(i)}$, we obtain a $\boxed{\kmeansratio}$-approximation algorithm for Euclidean $k$-means clustering. We defer the details to Appendix \ref{app:k_means_numerical_analysis}.
    
\section{Improved Approximation Algorithm for $k$-median} \label{sec:k_median}

\subsection{Improvement to $1+\sqrt{2}$-approximation} \label{subsec:lmp_k_median_easy}

In this subsection, we show that a $1+\sqrt{2}+\eps$-approximation can be obtained by a simple modification of the Ahmadian et al. \cite{ahmadian2017better} analysis. Because we use the same algorithm as \cite{ahmadian2017better}, the reduction from an LMP algorithm to a full polynomial-time algorithm is identical, so it suffices to improve the analysis of their LMP algorithm to a $1+\sqrt{2}$-approximation. The main difficulty in this subsection will be obtaining a tight bound on the norms (as opposed to squared norms) of points that are pairwise separated, which we prove in Lemma \ref{lem:geometric_median}. In the next subsection, we show how to break the $1+\sqrt{2}$ barrier that this algorithm runs into, which will follow a similar approach to our improved $k$-means algorithm.

We first recall the setup of the LMP approximation of \cite{ahmadian2017better}.
Let $c(j, i) = d(j, i)$ be the distance between a client $j \in \mathcal{D}$ and a facility $i \in \mathcal{F}$. Suppose we have a solution $\alpha$ to $\text{DUAL}(\lambda)$, such that every client $j$ has a tight witness $w(j) \in \mathcal{F}$ with $\alpha_j \ge t_{w(j)}$ and $\alpha_j \ge c(j, w(j)).$ Recall that $t_i = \max_{j \in N(i)} \alpha_j$, where $N(i) = \{j \in \mathcal{D}: \alpha_j > c(j, i)\}$, and likewise, $N(j) = \{i \in \mathcal{F}: \alpha_j > c(j, i)\}$.
Now, we let the conflict graph $H(\delta)$ on tight facilities (i.e., facilities $i$ with $\sum_{j \in N(i)} (\alpha_j-c(j, i)) = \lambda$) have an edge $(i, i')$ if $c(i, i') \le \delta \cdot \min(t_i, t_{i'})$.

We let $\delta = \sqrt{2}$ and return a maximal independent set $I$ of $H(\delta)$ as our LMP-approximation. It suffices to show that for each client $j \in \mathcal{D},$ that $c(j, I) \le (1+\sqrt{2}) \cdot \left(\alpha_j-\sum_{i \in N(j) \cap I} (\alpha_j-c(j, i))\right).$ To see why, by adding over all clients $j$, we obtain that
\[\text{cost}(\mathcal{D}, I) \le (1+\sqrt{2}) \cdot \left(\sum_{j \in \mathcal{D}} \alpha_j - \sum_{i \in I} \sum_{j \in N(i)} (\alpha_j-c(j, i))\right) = (1+\sqrt{2}) \cdot \left(\sum_{j \in \mathcal{D}} \alpha_j - \lambda \cdot |I|\right).\]
Finally, since $\alpha$ is a feasible solution to $\text{DUAL}(\lambda),$ this implies that $\text{cost}(\mathcal{D}, I) \le (1+\sqrt{2}) \cdot \text{OPT}_{|I|}.$

Before we verify the LMP approximation, we need the following lemma about points in Euclidean space.

\begin{lemma} \label{lem:geometric_median}
    Let $h \ge 2$ and suppose that $x_1, \dots, x_h$ are points in Euclidean space $\BR^d$ (for some $d$) such that $\|x_i-x_j\|_2^2 \ge 2$ for all $i \neq j$. Then, $\sum_{i=1}^{h} \|x_i\|_2 \ge \sqrt{h \cdot (h-1)}$.
\end{lemma}

\begin{proof}
    Note that for any positive real numbers $t_1, t_2, \dots, t_h$ that add to $1$, we have that
\[\sum_{i = 1}^{h} t_i \cdot \|x_i\|_2^2 \ge \sum_{i=1}^{h} t_i \cdot \|x_i\|_2^2 - \left\|\sum t_i x_i\right\|_2^2 = \sum_{i < j} t_i t_j \|x_i-x_j\|_2^2 \ge 2 \cdot \sum_{i < j} t_i t_j.\]
    Then, by setting $a_i = \|x_i\|_2$ for each $i$ and scaling by $t_1+\cdots+t_h$ accordingly to remove the assumption that $t_1+\cdots+t_h = 1$, we have that
\[\left(\sum_{i = 1}^{h} t_i \cdot a_i^2\right) \cdot \left(\sum_{i=1}^{h} t_i\right) \ge 2 \cdot \sum_{i < j} t_it_j\]
    for all $t_1, \dots, t_h \ge 0$. Now, if some $a_i = 0$, then $\|x_j\|_2 = \|x_i-x_j\|_2 \ge \sqrt{2}$, which means that $\sum_{i=1}^{h} \|x_j\|_2 \ge (h-1) \cdot \sqrt{2} \ge \sqrt{h(h-1)}$ for all $h \ge 2$. Alternatively, $a_i \neq 0$ for any $i$, so we can set $t_i = \frac{1}{a_i}$, to obtain that
\begin{equation} \label{eq:sym_ineq_1}
    \left(\sum_{i=1}^{h} a_i\right) \cdot \left(\sum_{i=1}^{h} \frac{1}{a_i}\right) \ge 2 \cdot \sum_{i < j} \frac{1}{a_ia_j}.
\end{equation}

    From now on, for any polynomial $P(a_1, \dots, a_h)$, we denote $\sum_{\text{sym}} P(a_1, \dots, a_h)$ to be the sum of all distinct terms of the form $P(a_{\pi(1)}, \dots, a_{\pi(h)})$ over all permutations of $[h]$. For instance, $\sum_{\text{sym}} a_1a_2a_3 = \sum_{1 \le i < j < k \le h} a_ia_ja_k$ and $\sum_{\text{sym}} a_1^2 a_2 = \sum_{1 \le i < j \le h} a_i^2 a_j + \sum_{1 \le j < i \le h} a_i^2 a_j$.

    In the case when $h = 2$, this means that $(a_1+a_2) \cdot \frac{a_1+a_2}{a_1a_2} \ge \frac{2}{a_1a_2},$ so $a_1+a_2 \ge \sqrt{2} = \sqrt{h(h-1)}.$ Alternatively, we assume that $h \ge 3$. When $h \ge 3$, note that
\begin{align}
\left(\sum_{\text{sym}} a_1 \cdots a_{h-2}\right) \cdot \left(\sum a_i\right) &= (h-1) \cdot \sum_{\text{sym}} a_1 \cdots a_{h-1} + \sum_{\text{sym}} a_1^2 a_2 \cdots a_{h-2} \nonumber \\
&\ge \left[(h-1) + \frac{h(h-1)(h-2)/2}{h}\right] \cdot \sum_{\text{sym}} a_1 \cdots a_{h-1} \nonumber \\
&= \frac{h(h-1)}{2} \cdot \sum_{\text{sym}} a_1 \cdots a_{h-1}, \label{eq:sym_ineq_2}
\end{align}
    where the second line above follows by Muirhead's inequality. Therefore, we have that
\begin{align*}
    \left(\sum_{i=1}^{h} a_i\right)^2 &\ge \left(\sum_{i=1}^{h} a_i\right) \cdot \frac{h(h-1)}{2} \cdot \frac{\sum_{\text{sym}} a_1 \cdots a_{h-1}}{\sum_{\text{sym}} a_1 \cdots a_{h-2}} \\
    &= \left(\sum_{i=1}^{h} a_i\right) \cdot \frac{h(h-1)}{2} \cdot \frac{\sum_{\text{sym}} \frac{1}{a_i}}{\sum_{\text{sym}} \frac{1}{a_ia_j}} \\
    &\ge 2 \cdot \frac{h(h-1)}{2} \\
    &= h(h-1),
\end{align*}
    where the first line follows from \eqref{eq:sym_ineq_2} and the third line follows from \eqref{eq:sym_ineq_1}. Therefore, we indeed have that $\sum_{i=1}^{h} \|x_i\|_2 \ge \sqrt{h(h-1)}$.
\end{proof}

To verify the LMP approximation, it suffices to show that for every $j$, $c(j, I) \le (1+\sqrt{2}) \cdot \left(\alpha_j-\sum_{i \in N(j) \cap I} (\alpha_j-c(j, i))\right).$ We split this up into $3$ cases.

\paragraph{Case 1: $\boldsymbol{|I \cap N(j)| = 0}$.} In this case, $d(j, I) \le d(j, w(j))+d(w(j), I)$ by the Triangle Inequality. But we know that $d(j, w(j)) \le \alpha_j$, and that $d(w(j), I) \le \sqrt{2} \cdot t_{w(j)} \le \sqrt{2} \cdot \alpha_j$, using the fact that $I$ is a maximal independent set so $w(j)$ has some neighbor of $I$ in the conflict graph. Thus, $d(j, I) \le (1+\sqrt{2}) \cdot \alpha_j$. However, since $N(j) \cap I = \emptyset$, this means that $\left(\alpha_j-\sum_{i \in N(j) \cap I} (\alpha_j-c(j, i))\right) = \alpha_j$. So, the desired inequality holds.

\paragraph{Case 2: $\boldsymbol{|I \cap N(j)| = 1}$.} In this case, let $i_1$ be the unique point in $N(j) \cap I$. Then, $d(j, I) \le d(j, i_1)$. In addition, $\left(\alpha_j-\sum_{i \in N(j) \cap I} (\alpha_j-c(j, i))\right) = \alpha_j-(\alpha_j-c(j, i_1)) = c(j, i_1) = d(j, i_1)$. Since $d(j, i_1) \ge 0$, the desired inequality holds (even with a ratio of $1 < 1+\sqrt{2}$).

\paragraph{Case 3: $\boldsymbol{|I \cap N(j)| = s \ge 2}$.} In this case, let $i_1, \dots, i_s$ be the set of points in $N(j) \cap I$. Then, we know that $d(i_r, i_{r'}) \ge \delta \cdot \min(t_{i_r}, t_{i_{r'}})$ for any $r \neq r'$. But $t_{i_r}, t_{i_{r'}} \ge \alpha_j$ by definition of $t_i$ (since $i_r, i_{r'} \in N(j)$), so this means that $d(i_r, i_{r'}) \ge \sqrt{2} \cdot \alpha_j$ for every $r \neq r'$.

Now, by applying Lemma \ref{lem:geometric_median}, we have that $\sum_{r=1}^{s} d(j, i_r) \ge \sqrt{s \cdot (s-1)} \cdot \alpha_j$. Now, let $t = \frac{1}{\alpha_j} \cdot \sum_{r=1}^{s} d(j, i_r)$, so $t \ge \sqrt{s \cdot (s-1)}$. Then, $d(j, I) \le \min_{1 \le r \le s} d(j, i_r) \le \frac{1}{s} \cdot \sum_{r=1}^{s} d(j, i_r) = \frac{T}{s} \cdot \alpha_j$. Ina ddition, we have that $\alpha_j-\sum_{i \in N(j) \cap I} (\alpha_j-c(j, i)) = \alpha_j-s \cdot \alpha_j + T \cdot \alpha_j = (T-(s-1)) \cdot \alpha_j$. So, the ratio is
\[\frac{T/s}{T-(s-1)} \le \frac{\sqrt{s \cdot s-1}/s}{\sqrt{s(s-1)}-(s-1)} = \frac{1}{s-\sqrt{s(s-1)}} \le 2.\]
Above, the first inequality follows because as $T$ increases, the numerator increases at a slower rate than the denominator, so assuming that the fraction is at least $1$, we wish for $T$ to be as small as possible to maximize the fraction. The final inequality holds because $s-\sqrt{s(s-1)} \ge \frac{1}{2}$ for all $s \ge 2$. Therefore, the desired inequality holds (even with a ratio of $2 < 1+\sqrt{2}$).

\medskip

So in fact, there is a simple improvement from the $1+\sqrt{8/3} \approx 2.633$ approximation algorithm to a $1+\sqrt{2} \approx 2.414$ algorithm. A natural question is whether this can be improved further without any significant changes to the algorithm or analysis. Indeed, there only seems to be one bottleneck, when $|I \cap N(j)| = 0$, so naturally one may assume that by slightly reducing $\delta=\sqrt{2}$, the approximation from Case 1 should improve below $1+\sqrt{2}$ and the approximation from Case 3 should become worse than $2$, but can still be below $1+\sqrt{2}$.

Unfortunately, such a hope cannot be realized. Indeed, if we replace $\delta=\sqrt{2}$ with some $\delta < \sqrt{2}$, we may have that $d(j, i_1) = d(j, i_2) = \cdots = d(j, i_s) = \delta \cdot \sqrt{\frac{s-1}{2s}} \cdot \alpha_j$ and the pairwise distances are all exactly $\delta \cdot \alpha_j$ between each $i_r, i_{r'}$. However, in this case, $\alpha_j-\sum_{i \in N(j) \cap I} (\alpha_j-c(j,i)) = \alpha_j \cdot \left(1-s+\delta \cdot \sqrt{s(s-1)/2}\right),$ which for $\delta < \sqrt{2}$ is in fact negative for sufficiently large $s$. Hence, even for $\delta = \sqrt{2}-\eps$ for a very small choice of $\eps > 0$, we cannot even guarantee a constant factor approximation with this analysis approach. So, this approach gets stuck at a $1+\sqrt{2}$ approximation.

In the following subsection, we show how an improved LMP approximation algorithm for Euclidean $k$-median, breaking the $1+\sqrt{2}$ approximation barrier. We will then show that we can also break this barrier for a polynomial-time $k$-median algorithm as well.

\subsection{An improved LMP algorithm for Euclidean $k$-median}

Recall the conflict graph $H := H(\delta)$, where we define two tight facilities $(i, i')$ to be connected if $c(i, i') \le \delta \cdot \min(t_i, t_{i'}).$ We set parameters $\delta_1 \ge \delta_2 \ge \delta_3$ and $0 < p < 1$, and define $V_1$ to be the set of all tight facilities. Given the set of tight facilities $V_1$ and conflict graphs $H(\delta)$ for all $\delta > 0$, our algorithm works by applying the procedure described in Algorithm \ref{alg:lmp_k_median} to $V_1$.

\begin{figure}
\centering
\begin{algorithm}[H]
    \caption{
    Generate a Nested Quasi-Independent Set of $V_1$, as well as a set of centers $S$ providing an LMP approximation for Euclidean $k$-median
    }\label{alg:lmp_k_median}
\Call{LMPMedian}{$V_1, \{H(\delta)\}, \delta_1, \delta_2, \delta_3, p$}:
    \begin{algorithmic}[1] 
        \item Create a maximal independent set $I_1$ of $H(\delta_1)$. \label{step:I1_median}
        \item Let $V_2$ be the set of points in $V_1 \backslash I_1$ that are not adjacent to $I_1$ in $H(\delta_2)$. \label{step:V2_median}
        \item Create a maximal independent set $I_2$ of the induced subgraph $H(\delta_1)[V_2]$. \label{step:I2_median}
        \item Let $V_3$ be the set of points $i$ in $V_2 \backslash I_2$ such that there is exactly one point in $I_2$ that is a neighbor of $i$ in $H(\delta_1)$, there are no points in $I_1$ that are neighbors of $i$ in $H(\delta_2)$, and there are no points in $I_2$ that are neighbors of $i$ in $H(\delta_3)$. \label{step:V3_median}
        \item Create a maximal independent set $I_3$ of the induced subgraph $H(\delta_1)[V_3]$. \label{step:I3_median}
        \item Note that every point $i \in I_3$ has a unique adjacent neighbor $q(i) \in I_2$ in $H(\delta_1)$. We create the final set $S$ as follows: \label{step:S_median}
        \begin{itemize}
            \item Include every point $i \in I_1$.
            \item For each point $i \in I_2$, flip a fair coin. If the coin lands heads, include $i$ with probability $2p$. Otherwise, include each point in $q^{-1}(i)$ independently with probability $2p$.
        \end{itemize}
    \end{algorithmic}
\end{algorithm}
\end{figure}

As in the $k$-means case, we consider a more general setup, so that we can convert the LMP approximation to a full polynomial-time algorithm. Instead of $V_1,$ let $\mathcal{V} \subset \mathcal{F}$ be a subset of facilities and let $\mathcal{D}$ be the full set of clients. For each $j \in \mathcal{D},$ let $\alpha_j \ge 0$ be some real number, and for each $i \in \mathcal{V}$, let $t_i \ge 0$ be some real number. In addition, for each client $j \in \mathcal{D}$, we associate with it a set $N(j) \subset \mathcal{V}$ and a ``witness'' facility $w(j) \in \mathcal{V}$.
Finally, suppose that we have the following assumptions:

\begin{enumerate}
    \item For any client $j \in \mathcal{D}$, the witness $w(j) \in \mathcal{V}$ satisfies $\alpha_j \ge t_{w(j)}$ and $\alpha_j \ge c(j, w(j))$.
    \item For any client $j \in \mathcal{D}$ and any facility $i \in N(j)$, $t_i \ge \alpha_{j} > c(j, i)$.
\end{enumerate}

Then, for the graph $H(\delta)$ on $\mathcal{V}$ where $i, i' \in \mathcal{V}$ are connected if and only if $c(i, i') \le \delta \cdot \min(t_i, t_{i'})$ (recall that now, $c(i, i') = d(i, i')$ instead of $d(i, i')^2$), we have the following main lemma.

\begin{lemma} \label{lem:main_lmp_median}
    Fix $\delta_1 = \sqrt{2}$, $\delta_2 = 1.395$, and $\delta_3 = 2-\sqrt{2} \approx 0.5858,$ and let $p < 0.337$ be variable.  Now, let $S$ be the randomized set created by applying Algorithm \ref{alg:lmp_k_median} on $V_1 = \mathcal{V}$.
    Then, for any $j \in \mathcal{D}$,
\[\BE[c(j, S)] \le \rho(p) \cdot \BE\left[\alpha_j - \sum_{i \in N(j) \cap S} (\alpha_j-c(j, i))\right],\]
    where $\rho(p)$ is some constant that only depends on $p$ (since $\delta_1, \delta_2, \delta_3$ are fixed).
\end{lemma}

\begin{proof}
As in the $k$-means case, we fix $j \in \mathcal{D}$, and we do casework based on the sizes of $a=|I_1 \cap N(j)|$, $b=|I_2 \cap N(j)|$, and $c=|I_3 \cap N(j)|$.

\paragraph{Case 1: $\boldsymbol{a=0, b=1, c=0}$.}
Let $i_2$ be the unique point in $I_2 \cap N(j),$ and let $i^* = w(j)$ be the witness of $j$. We have the following subcases:

\begin{enumerate}[label=\alph*)]
    \item $\boldsymbol{i^* \not\in V_2}$. In this case, either $i^* \in I_1$ so $d(i^*, I_1) = 0$, or there exists $i_1 \in I_1$ such that $d(i^*, i_1) \le \delta_2 \cdot \min(t_{i^*}, t_{i_1}) \le \delta_2$. So, $d(j, I_1) \le 1+\delta_2$. In addition, we have that $i_2 \in S$ with probability $p$. So, if we let $t := d(j, i_2)$, we can bound the ratio by
\begin{equation} \tag{1.a'}\label{eq:1.a'}
    \frac{p \cdot t + (1-p) \cdot (1+\delta_2)}{1-p(1-t)} = \frac{p \cdot t + (1-p) \cdot (1+\delta_2)}{p \cdot t + (1-p)} \le 1+\delta_2,
\end{equation}
    since $t \ge 0$.

    \item $\boldsymbol{i^* \in V_3}.$ In this case, there exists $i_3 \in I_3$ (possibly $i_3 = i^*$) such that $d(i^*, i_3) \le \delta_1 \cdot \min(t_{i^*}, t_{i_3}).$ In addition, there exists $i_1 \in I_1$ such that $d(i^*, i_1) \le \delta_1 \cdot \min(t_{i^*}, t_{i_1})$. In addition, we have that $t_{i^*} \le \alpha_j = 1$. Finally, since $I_3 \subset V_2$, we must have that $d(i_1, i_3) \ge \delta_2 \cdot \min(t_{i_1}, t_{i_3})$. If we condition on $i_2 \in S$, then the numerator and denominator both equal $c(j, i_2)$, so the fraction is $1$ (or $0/0$). Else, if we condition on $i_2 \not\in S$, then the denominator is $1$, and $i_3 \in S$ with probability either $p$ or $\frac{p}{1-p} > p$. Therefore, $\BE[d(j, S)|i_2 \not\in S] \le p \cdot \|i_3-j\|_2 + (1-p) \cdot \|i_1-j\|_2$. We can bound this (we defer the details to Appendix \ref{app:bash}) by 
\begin{equation} \tag{1.b'}\label{eq:1.b'}
    \inf_{T > 0} \sqrt{3\left(X+Y\right)+2\sqrt{2\left(X+Y\right)^{2}-\delta_2^2 \cdot XY}},
\end{equation}
where $X = p^2+p(1-p) \cdot T$ and $Y = (1-p)^2 + \frac{p(1-p)}{T}.$
\end{enumerate}

    In the remaining cases, we may assume that $i^*\in V_2 \backslash V_3$. Then, one of the following must occur:
\begin{enumerate}[label=\alph*)] \setcounter{enumi}{2}
    \item $\boldsymbol{i^* = i_2}$. In this case, define $t = d(j, i^*) \in [0, 1]$, and note that $d(j, I_1) \le d(j, i^*)+d(i^*, I_1) \le t + \delta_1$. So, with probability $p$, we have that $d(j, S) \le d(j, i^*) = t$, and otherwise, we have that $d(j, S) \le d(j, I_1) = t + \delta_1$. So, we can bound the ratio by
\[\max_{0 \le t \le 1} \frac{p \cdot t + (1-p) \cdot (t + \delta_1)}{1-p \cdot (1-t)} = \max_{0 \le t \le 1} \frac{t + (1-p)\delta_1}{p \cdot t + (1-p)}.\]
    For $p$ such that $1/p > \delta_1,$ it is clear that this function increases as $t$ increases, so it is maximized when $t = 1$, which means we can bound the ratio by
\begin{equation} \tag{1.c'}\label{eq:1.c'}
    1 + (1-p) \cdot \delta_1.
\end{equation}

    \item $\boldsymbol{i^* \in I_2}$ \textbf{but} $\boldsymbol{i^* \neq i_2}$. First, we recall that $d(j, i^*) \le 1$. Now, let $t = d(j, i_2)$. In this case, with probability $p$, $d(j, S) = t$ (if we select $i_2$ to be in $S$), with probability $p(1-p)$, $d(j, S) \le 1$ (if we select $i^*$ but not $i_2$ to be in $S$), and in the remaining event of $(1-p)^2$ probability, we still have that $d(j, S) \le d(j, I_1) \le d(j, i^*) + d(i^*, I_1) \le 1 + \delta_1.$ So, we can bound the ratio by
\[\max_{0 \le t \le 1} \frac{p \cdot t + p(1-p) \cdot 1 + (1-p)^2 \cdot (1+\delta_1)}{1-p \cdot (1-t)}.\]
    Note that this is maximized when $t = 0$ (since the numerator and denominator increase at the same rate when $t$ increases), so we can bound the ratio by
\begin{equation} \tag{1.d'} \label{eq:1.d'}
    \frac{p(1-p) + (1-p)^2 \cdot (1+\delta_1)}{1-p} = 1 + (1-p) \cdot \delta_1.
\end{equation}
    
    \item \textbf{There is more than one neighbor of} $\boldsymbol{i^*}$ \textbf{in} $\boldsymbol{H(\delta_1)}$ \textbf{that is in} $\boldsymbol{I_2}$. In this case, there is some other point $i_2' \in I_2$ not in $N(j)$ such that $d(i^*, i_2') \le \delta_1 \cdot \min(t_{i^*}, t_{i_2'}).$ So, we have four points $j, i^*, i_1 \in I_1, i_2' \in I_2$ such that $d(j, i^*) \le 1,$ $d(i^*, i_2') \le \delta_1 \cdot \min(t_{i^*}, t_{i_2'}),$ $d(i^*, i_1) \le \delta_1 \cdot \min(t_{i^*}, t_{i_1}),$ and $d(i_1, i_2') \ge \delta_2 \cdot \min(t_{i_1}, t_{i_2'}).$
    
    If we condition on $i_2 \in S$, then the denominator equals $c(j, i_2)$ and the numerator is at most $c(j, i_2)$, so the fraction is $1$ (or $0/0$). Else, if we condition on $i_2 \not\in S$, then the denominator is $1$, and the numerator is at most $p \cdot \|i_2'-j\|_2^2 + (1-p) \cdot \|i_1-j\|_2^2$. Note that $d(j, i^*) \le 1$, that $t_{i^*} \le 1$, and that $\delta_2, \delta_1 \ge \delta_2$. So, as in case b), the overall fraction is at most
\begin{equation} \tag{1.e'} \label{eq:1.e'}
    \inf_{T > 0} \sqrt{3\left(X+Y\right)+2\sqrt{2\left(X+Y\right)^{2}-\delta_2^2 \cdot XY}},
\end{equation}
where $X = p^2+p(1-p) \cdot T$ and $Y = (1-p)^2 + \frac{p(1-p)}{T}.$    
    
    \item \textbf{There are no neighbors of} $\boldsymbol{i^*}$ \textbf{in} $\boldsymbol{H(\delta_1)}$ \textbf{that are in} $\boldsymbol{I_2}$. In this case, $d(i^*, i_2) \ge \delta_1 \cdot \min(t_{i^*}, t_{i_2}).$ Define $t = \min(t_{i^*}, t_{i_2})$. Since $d(j, i^*) \le 1,$ by the triangle inequality we have that $d(j, i_2) \ge \max\left(0, \delta_1 \cdot t-1\right)$. In addition, we still have that $d(j, I_1) \le d(j, i^*) + d(i^*, I_1) \le 1 + \delta_1 \cdot t_{i^*},$ and $d(j, I_1) \le d(j, i_2) + d(i_2, I_1) \le 1 + \delta_1 \cdot t_{i_2}$, so together we have that $d(j, I_1) \le 1+\delta_1 \cdot t$.
    Since $i_2 \in S$ with probability $p$, the ratio is at most
\begin{align*}
    &\hspace{0.5cm}\max_{0 \le t \le 1} \max_{d(j, i_2) \ge \delta_1 \cdot t - 1} \frac{p \cdot d(j, i_2) + (1-p) \cdot (1+\delta_1 \cdot t)}{1-p(1-d(j, i_2))} \\
    &= \max_{0 \le t \le 1} \max_{d(j, i_2) \ge \delta_1 \cdot t - 1} \frac{p \cdot d(j, i_2) + (1-p) \cdot (1+\delta_1 \cdot t)}{p \cdot d(j, i_2) + (1-p)}.
\end{align*}
    It is clear that this function is decreasing as $d(j, i_2)$ is increasing (and nonnegative). So, we may assume WLOG that $d(j, i_2) = \max(0, \sqrt{\delta_2 \cdot t}-1)$ to bound this ratio by 
\begin{equation*}
    \max_{0 \le t \le 1} \frac{p \cdot \max(0, \delta_1 \cdot t-1) + (1-p) \cdot (1+\delta_1 \cdot t)}{p \cdot \max(0, \delta_1 \cdot t-1) + (1-p)}    
\end{equation*}
    If $\delta_1 \cdot t - 1 \le 0$, then $\delta_1 \cdot t + 1 \le 2,$ so we can bound the above equation by $2$. Otherwise, the above fraction can be rewritten as $\frac{\delta_1 \cdot t + (1-2p)}{p \cdot \delta_1 \cdot t + (1-2p)}$. For $p < 0.5,$ this is maximized when $t = 1$ over the range $t \in [0, 1]$, so we can bound the ratio by
\begin{equation} \tag{1.f'} \label{eq:1.f'}
    \frac{1-2p + \delta_1}{1-2p + p \cdot \delta_1}.
\end{equation}

    \item \textbf{There is a neighbor of} $\boldsymbol{i^*}$ \textbf{in} $\boldsymbol{H(\delta_3)}$ \textbf{that is also in} $\boldsymbol{I_2}$. In this case, either $d(i^*, i_2) \le \delta_3 \cdot t_{i^*}$ so $d(i_2, j) \ge \max(0, d(j, i^*)-\delta_3 \cdot t_{i^*})$, or there is some other point $i_2' \in I_2$ not in $N(j)$ such that $d(i^*, i_2') \le \delta_3 \cdot \min(t_{i^*}, t_{i_2'}).$ If $d(i^*, i_2) \le \delta_3 \cdot t_{i^*}$, then define $t = t_{i^*}$ and $u = d(j, i^*)$. In this case, $d(j, I_1) \le u + \delta_1 \cdot t,$ and $d(j, i_2) \ge \max(0, u-\delta_3 \cdot t).$ So, the fraction is at most
\[\frac{(1-p) \cdot (u+\delta_1 \cdot t)+p \cdot d(j, i_2)}{1-p+p \cdot d(j, i_2)} \le \frac{(1-p) \cdot (u+\delta_1 \cdot t)+p \cdot \max(0, u-t \cdot \delta_3)}{1-p+p \cdot \max(0, u-t \cdot \delta_3)}.\]
    Since $t = t_{i^*} \le 1$ and $d(j, i^*) \le 1$, we can bound the overall fraction as at most
\begin{align} 
    &\hspace{0.5cm}\max_{0 \le t \le 1} \max_{0 \le u \le 1} \frac{(1-p) \cdot (u+\delta_1 \cdot t)+p \cdot \max(0, u-t \cdot \delta_3)}{1-p+p \cdot \max(0, u-t \cdot \delta_3)} \nonumber \\
    &\le \max\left(\delta_1+\delta_3, \frac{1 + \delta_1 - p(\delta_1+\delta_3)}{1-p \cdot \delta_3}\right) \tag{1.g.i'} \label{eq:1.g.i'}
\end{align}
    We derive the final equality in Appendix \ref{app:bash}.
    
    Alternatively, if $d(i^*, i_2') \le \delta_3 \cdot \min(t_{i^*}, t_{i_2'}),$ then if we condition on $i_2 \in S,$ the fraction is $1$ (or $0/0$), and if we condition on $i_2 \not\in S$, the denominator is $1$ and the numerator is at most $p \cdot d(j, i_2') + (1-p) \cdot d(j, i_1) \le p \cdot (1+\delta_3)+(1-p) \cdot (1+\delta_1).$ (Note that $i_2 \in S$ and $i_2' \in S$ are independent.) Therefore, we can also bound the overall fraction by 
\begin{equation} \tag{1.g.ii'} \label{eq:1.g.ii'}
    p \cdot (1+\delta_3)+(1-p) \cdot (1+\delta_1).
\end{equation}

    \item \textbf{There is a neighbor of} $\boldsymbol{i^*}$ \textbf{in} $\boldsymbol{H(\delta_2)}$ \textbf{that is also in} $\boldsymbol{I_1}$. In this case, $i^*$ would not be in $V_2$, so we are back to sub-case 1.a'.
\end{enumerate}

\paragraph{Case 2: $\boldsymbol{a = 0, b = 0, c \le 1}$.}
We again let $i^*$ be the witness of $j$. In this case, if $i \not\in V_2$, then there exists $i_1 \in I_1$ such that $d(i^*, i_1) \le \delta_2 \cdot \min(t_{i^*}, t_{i_1}) \le \delta_2$, in which case $d(j, I_1) \le 1+\delta_2.$ Otherwise, there exists $i_1 \in I_1$ such that $d(i^*, i_1) \le \delta_1 \cdot \min(t_{i^*}, t_{i_1})$, and there exists $i_2 \in I_2$ such that $d(i^*, i_2) \le \delta_1 \cdot \min(t_{i^*}, t_{i_2})$. Finally, in this case we also have that $d(i_1, i_2) \ge \delta_2 \cdot \min(t_{i_1}, t_{i_2})$. Now, we consider two subcases, either $c = 0$ or $c = 1$.

\begin{enumerate}[label=\alph*)]
    \item $\boldsymbol{c=0.}$ In this case, we have that the denominator is $1$, and the numerator is either at most $1+\delta_2$, or is at most $p \cdot \|j-i_2\|_2 + (1-p) \cdot \|j-i_1\|_2$, where $d(j, i^*) \le 1$, $d(i^*, i_1) \le \delta_1 \cdot \min(t_{i^*}, t_{i_1})$, $d(i^*, i_2) \le \delta_1 \cdot \min(t_{i^*}, t_{i_2})$, and $d(i_1, i_2) \ge \delta_2 \cdot \min(t_{i_1}, t_{i_2})$. Hence, we can bound the overall fraction, by the same computation as in the $k$-median subcase 1.b), as
\begin{equation} \tag{2.a'} \label{eq:2.a'}
    \max\left(1+\delta_2, 
    \inf_{T > 0} \sqrt{3\left(X+Y\right)+2\sqrt{2\left(X+Y\right)^{2}-\delta_2^2 \cdot XY}}\right),
\end{equation}
where $X = p^2+p(1-p) \cdot T$ and $Y = (1-p)^2 + \frac{p(1-p)}{T}.$

    \item $\boldsymbol{c=1.}$ In this case, let $i_3$ be the unique point in $N(j) \cap I_3.$ Then, conditioned on $i_3$ being in $S$, the numerator and denominator both equal $d(j, i_3)$. Otherwise, the denominator is $1$ and we can bound the numerator the same way as in subcase 2a), since the probability of $i_2 \in S$ is either $p$ (if $q(i_3) \neq i_2$) or $\frac{p}{1-p} \ge p$ (if $q(i_3) = i_2$). So, we can bound the overall fraction again as 
\begin{equation} \tag{2.b'} \label{eq:2.b'}
    \max\left(1+\delta_2, 
    \inf_{T > 0} \sqrt{3\left(X+Y\right)+2\sqrt{2\left(X+Y\right)^{2}-\delta_2^2 \cdot XY}}\right),
\end{equation}
where $X = p^2+p(1-p) \cdot T$ and $Y = (1-p)^2 + \frac{p(1-p)}{T}.$    

\end{enumerate}

\paragraph{Case 3: $\boldsymbol{a = 0}$, all other cases.}
Note that in this case, we may assume $b+c = |N(j) \cap (I_2 \cup I_3)| \ge 2$, since we already took care of all cases when $a = 0$ and $b+c \le 1$. We split into two main subcases.

\begin{enumerate}[label=\alph*)]
    \item \textbf{Every point} $\boldsymbol{i}$ \textbf{in} $\boldsymbol{N(j) \cap (I_2 \cup I_3)}$ \textbf{satisfies} $\boldsymbol{d(j, i) \ge \delta_1-1.}$
    In this case, let $\bar{I} \subset I_2 \cup I_3$ represent the set of points selected to be in $S$. Note that $\bar{I}$ is a random set.
    
    Note that with probability at least $2p-2p^2$, $|N(j) \cap \bar{I}| \ge 1$. (Since $N(j) \cap (I_2 \cup I_3)$ has size at least $2$, the probability of $\bar{I} \cap N(j)$ being nonempty is minimized when $b = 0, c = 2$, and the two points in $N(j) \cap I_3$ map to the same point under $q$.) In this event, let $h = |N(j) \cap \bar{I}|$, and let $r_1, \dots, r_h$ represent the distances from $j$ to each of the points in $N(j) \cap \bar{I}$. Then, by Lemma \ref{lem:geometric_median}, $\frac{r_1+\cdots+r_h}{h} \ge \sqrt{\frac{h-1}{h}}.$ So, if we set $r= \frac{r_1+\cdots+r_h}{h},$ then $r \le 2\left[r \cdot h - (h-1)\right]$ for any $h \ge 1$ and $r \ge \sqrt{\frac{h-1}{h}},$ which means that $\min r_i \le  \frac{r_1+\cdots+r_h}{h} \le 2 \cdot \left(1 - \sum_{i=1}^{h} (1-r_i)\right)$.
    
    In addition, if $|\bar{I}| = 1$, then $1-\sum (1-r_i) = r_1 \ge \delta_1-1 = \sqrt{2}-1$, and otherwise, because every point in $\bar{I}$ is separated by at least $\delta_1 = \sqrt{2}$ distance, $1 - \sum (1-r_i) \ge \sqrt{h(h-1)}-(h-1) \ge \sqrt{2}-1$ by Lemma \ref{lem:geometric_median}. Overall, this means that whenever $h = |N(j) \cap \bar{I}| \ge 1,$ $\min r_i \le 2 \cdot \left(1 - \sum (1-r_i)\right)$ and $1 - \sum (1-r_i) \ge \sqrt{2}-1$.
    
    In addition, if $\bar{I} = \emptyset$, then the denominator is $1$ and the numerator is at most $d(j, w(j)) + d(w(j), I_1) \le 1+\sqrt{2}$. Therefore, if we let $q$ be the probability that $|\bar{I}| \ge 1$ and $t$ be the expectation of $1-\sum(1-r_i)$ conditioned on $|\bar{I}| \ge 1$, the overall fraction is at most
\begin{align}
    \frac{(1+\sqrt{2}) \cdot (1-q) + 2 \cdot t \cdot q}{(1-q) + t \cdot q} &\le \frac{(1+\sqrt{2}) \cdot (1-q) + 2 \cdot (\sqrt{2}-1) \cdot q}{(1-q) + (\sqrt{2}-1) q} \nonumber \\
    &\le \frac{(1+\sqrt{2})-(3-\sqrt{2}) \cdot (2p-2p^2)}{1-(2-\sqrt{2}) \cdot (2p-2p^2)}. \tag{3.a'} \label{eq:3.a'}
\end{align}
    
    \item \textbf{There exists a point} $\boldsymbol{i \in N(j) \cap (I_2 \cup I_3)}$ \textbf{such that} $\boldsymbol{d(j, i) < \delta_1-1.}$ 
    In this case, note that $d(i, i') < \delta_1$ for all points $i' \in N(j) \cap (I_2 \cup I_3)$. Assuming $b+c \ge 2$, this is only possible if either: 
    \begin{enumerate}[i)]
        \item $b=1,c=1$ and the unique points $i_2 \in N(j) \cap I_2$ and $i_3 \in N(j) \cap I_3$ satisfy $q(i_3) = i_2$, or \label{case:i}
        \item $b=1, c\ge 2$, the unique point $i_2 \in N(j) \cap I_2$ is the only point with $d(j, i) < \delta_1-1$, and every point in $N(j) \cap I_3$ maps to $i_2$ under $q$. \label{case:ii}
    \end{enumerate}
    
    First, assume Case \ref{case:i}.
    Let $r=d(j, i_2)$ and $s=d(j, i_3)$. Then, $\BE \left[1-\sum_{i \in N(j) \cap S} (1-d(j, i))\right]$ $= (1-2p) \cdot 1 + p \cdot r + p \cdot s$, and the expected distance $d(j, S)$ is at most $p \cdot r + p \cdot s + (1-2p) \cdot (1+\sqrt{2})$. Since $d(r, s) \ge \delta_3$, this means that $r+s \ge \delta_3,$ so the overall fraction is at most
\begin{equation} \label{eq:3.b.i'} \tag{3.b.i'}
    \frac{(1+\sqrt{2}) \cdot (1-2p) + \delta_3 \cdot p}{(1-2p) + \delta_3 \cdot p}.
\end{equation}

    Next, assume Case \ref{case:ii}.
    Let $r=d(j, i_2)$, and let $s_1, \dots, s_c$ be the distances from $j$ to each of the $c$ points in $N(j) \cap I_3$. Let $s = \frac{s_ 1+\cdots+s_c}{c}$ Then, $\BE\left[1-\sum_{i \in N(j) \cap S} (1-d(j,i))\right]=1-p(1-r)-\sum_{i=1}^{c}p(1-s_i) \ge 1-p(1+c)+p \cdot (r+s \cdot c).$ In addition, $\BE[c(j, S)]$ is at most $(1+\sqrt{2}) \cdot (\frac{1}{2}-p) + (1+\sqrt{2}) \cdot \frac{1}{2}(1-2p)^c + p \cdot r + \frac{1}{2}\left(1-(1-2p)^c\right) \cdot s$. Since the numerator and denominator grow at the same rate with respect to $r$, and the numerator grows slower with respect to $s$ than the denominator, we wish to minimize $r$ and $s$ to maximize the fraction. So, we set $r=0$, and $s=\sqrt{\frac{c-1}{c}}$ by Lemma \ref{lem:geometric_median}. Therefore, the fraction is at most
\begin{equation} \label{eq:3.b.ii'} \tag{3.b.ii'}
    \frac{(1+\sqrt{2}) \cdot \left(\frac{1}{2}(1-2p)+\frac{1}{2}(1-2p)^c\right) + \frac{1}{2}\left(1-(1-2p)^c\right) \cdot \sqrt{\frac{c-1}{c}}}{1-p(1+c-\sqrt{c(c-1)})}
\end{equation}
\end{enumerate}

\paragraph{Case 4: $\boldsymbol{a \ge 1}$.} First, we will condition on the fair coin flips, and let $\bar{I} \subset N(j) \cap (I_2 \cup I_3)$ be the set of ``surviving'' points, i.e., the points that will be included in $S$ with probability $2p$. Note all points in $\bar{I}$ have pairwise distance at least $\delta_1 = \sqrt{2}$ from each other, and all points in $N(j) \cap I_1$ have pairwise distance at least $\delta_1$ from each other also. However, the points in $N(j) \cap I_1$ and $\bar{I}$ are only guaranteed to have pairwise distance at least $\delta_2$ from each other. Let $h$ represent the size $|\bar{I}|$.

We consider several subcases.

\begin{enumerate}[label=\alph*)]
    \item $\boldsymbol{h = 0.}$
    In this case, we can use the same bounds as Cases 2 and 3 of the simpler $1+\sqrt{2}$-approximation, since we only have to worry about points in $I_1 \cap N(j).$ Indeed, the same bounds on the numerator and denominator still hold, so the ratio is at most
\begin{equation} \label{eq:4.a'} \tag{4.a'}
    2.
\end{equation}
    
    \item $\boldsymbol{a=1,h=1.}$
    In this case, let $i_1$ be the unique point in $N(j) \cap I_1$, and let $i_2$ be the unique point in $\bar{I}$. Then, $d(i_1, i_2) \ge \delta_2,$ so if $t = d(j, i_1)$ and $u = d(j, i_2)$, then the denominator in expectation is $1-(1-t)-2p(1-u) = t-2p(1-u) \ge t\cdot (1-2p)$, since $1-u \le t$. But, the numerator $\BE[c(j, S)]$ is at most $t$, so the overall fraction is at most
\begin{equation} \label{eq:4.b'} \tag{4.b'}
    \frac{1}{1-2p}.
\end{equation}
    
    \item $\boldsymbol{a = 1, h \ge 2.}$
    Let $i_1$ be the unique point in $N(j) \cap I_1.$ Then, we must have that $d(j, i_1) \ge \delta_2-1.$ Letting $t = d(j, i_1)$, we have that $d(j, S) \le t$, but $\BE\left[\alpha_j-\sum_{i \in N(j) \cap S} (\alpha_j-c(j, i))\right] \ge t - 2p \cdot \sum_{i \in \bar{I}} (1-c(j, i)).$ However, we know that $\sum_{i \in \bar{I}} (1-c(j, i)) \ge h-\sqrt{h(h-1)} \ge 2-\sqrt{2}$ by Lemma \ref{lem:geometric_median}, so the denominator is at least $t-(2-\sqrt{2}) \cdot 2p$. So, the ratio is at most $\frac{t}{t-2(2-\sqrt{2})p},$ which is maximized when $t$ is as small as possible, namely $t = \delta_2-1$. So, the ratio is at most
\begin{equation} \label{eq:4.c'} \tag{4.c'}
    \frac{\delta_2-1}{(\delta_2-1)-2(2-\sqrt{2})p}.
\end{equation}

    \item $\boldsymbol{a \ge 2, h = 1.}$
    In this case, let $i_2$ be the unique point in $\bar{I}$, and let $t=d(j, i_2)$. Note that $d(j, i_2) \ge \delta_2-1$, so $1-d(j, i_2) \le 2-\delta_2$. In addition, if the distances from $j$ to the points in $N(j) \cap I_1$ are $r_1, \dots, r_a$, then $d(j, S) \le \frac{r_1+\cdots+r_a}{a}.$ If we let $r = \frac{r_1+\cdots+r_a}{a}$, then $\BE\left[\alpha_j-\sum_{i \in N(j) \cap S} (\alpha_j-c(j, i))\right] = 1-a(1-r)-2p(1-t) \ge 1-a(1-r)-(2-\delta_2) \cdot 2p.$ So, the overall fraction is at most $\frac{r}{1-a(1-r)-(2-\delta_2) \cdot 2p}.$ It is clear that this function decreases as $r$ increases, so we want to set $r$ as small as possible. However, we know that $r \ge \sqrt{\frac{a-1}{a}}$ by Lemma \ref{lem:geometric_median}, so the overall fraction is at most
$$\frac{\sqrt{\frac{a-1}{a}}}{\sqrt{a(a-1)}-(a-1)-(2-\delta_2) \cdot 2p} = \frac{1}{(a-\sqrt{a(a-1)})-(2-\delta_2) \cdot 2p \cdot \sqrt{\frac{a-1}{a}}}.$$
    The denominator clearly decreases as $a \to \infty$, so the overall fraction is at most the limit of the above as $a \to \infty$, which is
\begin{equation} \label{eq:4.d'} \tag{4.d'}
    \frac{1}{\frac{1}{2}-(2-\delta_2) \cdot 2p}.
\end{equation}

    \item $\boldsymbol{a, h \ge 2.}$
    In this case, let the distances from $j$ to the points in $N(j) \cap I_1$ be $r_1, \dots, r_a$, and let the distances from the points from $j$ to the points in $\bar{I}$ be $s_1, \dots, s_h$. Also, let $r = \frac{r_1+\cdots+r_a}{a},$ and let $s = \frac{s_1+\cdots+s_h}{h}.$ Then, we have that the numerator is at most $r$, and the denominator is at least $1-a \cdot (1-r) - 2p \cdot h \cdot (1-s)$. Next, note that by Lemma \ref{lem:geometric_median}, $s \ge \sqrt{\frac{h-1}{h}}$, so $h \cdot (1-s) \ge h \cdot \left(1-\sqrt{\frac{h-1}{h}}\right) \ge h-\sqrt{h-1} \ge 2-\sqrt{2}$. So, the fraction is at most $\frac{r}{1-a(1-r)-2p \cdot (2-\sqrt{2})}$. This is exactly the same as in subcase 4d), except there the denominator was $1-a(1-r)-(2-\delta_2) \cdot 2p,$ i.e., we just replaced $\delta_2$ with $\sqrt{2}$. So, the same calculations give us that we can bound the overall fraction by at most
\begin{equation} \label{eq:4.e'} \tag{4.e'}
    \frac{1}{\frac{1}{2}-(2-\sqrt{2}) \cdot 2p}.
\end{equation}
\end{enumerate}
\end{proof}

Finally, we bound the actual LMP approximation constant, similar to Proposition \ref{prop:more_bash} for the $k$-means case. We have the following proposition,
which will immediately follow from analyzing all subases carefully (see Lemma \ref{lem:more_bash_median}).

\begin{proposition} \label{prop:more_bash_2}
    For $p = 0.068,$ $\rho(p) \le 2.395$. Hence, we can obtain a $2.395$-LMP approximation.
\end{proposition}

\subsection{Improved $k$-median approximation}

In this section, we explain how our LMP approximation for $k$-median implies an improved polynomial time $k$-median approximation for any fixed $k$. We set $p_1 = 0.068$ and $\delta_1 = \sqrt{2}, \delta_2 = 1.395$, and $\delta_3 = 2-\sqrt{2}$. In this case, we have that $\rho(\delta_1) \le 2.395$ by Proposition \ref{prop:more_bash_2}.

Next, we have that all of the results in Subsections \ref{subsec:k_means_alg_prelim} and \ref{subsec:k_means_analysis} hold in the $k$-median context, with two changes. The first, more obvious, change is that Lemma \ref{lem:cost_j_upper_bound} (and all subsequent results in Section \ref{subsec:k_means_analysis}) needs to use the function $\rho$ associated with $k$-median as opposed to the function associated with $k$-means. 

The second change is that Lemma \ref{lem:RHS_positive_1/2} no longer holds for $p \le 0.5$, but still holds for $p \le p_0$ for some fixed choice $p_0$. Indeed, for Cases 1, 2, and 3 (i.e., when $a = 0$), we have that $\BE\left[\alpha_j-\sum_{i \in N(j) \cap S} (\alpha_j-c(j, i))\right] \ge 0$ for any $p \le 0.5$, since $I_1 = \emptyset$, and $\sum_{i \in N(j) \cap I_2} (\alpha_j - c(j, i)) \le (b-\sqrt{b(b-1)}) \cdot \alpha_j \le \alpha_j$ if $|N(j) \cap I_2| = b$, and likewise $\sum_{i \in N(j) \cap I_3} (\alpha_j - c(j, i)) \le (c-\sqrt{c(c-1)}) \cdot \alpha_j \le \alpha_j$ if $|N(j) \cap I_3| = c$. So, $\BE\left[\alpha_j-\sum_{i \in N(j) \cap S} (\alpha_j-c(j, i))\right] \ge (1-2p) \cdot \alpha_j \ge 0$ for $p \le 0.5$. 
For case $4$ of $k$-median, we verify that $\BE\left[\alpha_j-\sum_{i \in N(j) \cap S} (\alpha_j-c(j, i))\right] \ge 0$ after conditioning on $\bar{I}$. Indeed, if $\bar{I} = 0$ (i.e., subcase 4.a), then this just equals $\alpha_j-\sum_{i \in N(j) \cap I_1} (\alpha_j-c(j, i)) \ge \alpha_j \cdot (\sqrt{a(a-1)}-(a-1)) \ge 0$. In the remaining subcases, the value $\BE\left[\alpha_j-\sum_{i \in N(j) \cap S} (\alpha_j-c(j, i))\right] \ge 0$ is always nonnegative as long as the denominator of our final fractions are also nonnegative. So, we just need that $1-2p \ge 0$, $(\delta_2-1)-2(2-\sqrt{2})p \ge 0$, $\frac{1}{2}-(2-\delta_2) \cdot 2p \ge 0$, and $\frac{1}{2}-(2-\sqrt{2}) \cdot 2p \ge 0$. These are all true as long as $p \le \frac{\delta_2-1}{2(2-\sqrt{2})} \le 0.337$. Thus, we replace $0.5$ with $p_0 = 0.337$.

Overall, the rest of Subsection \ref{subsec:k_means_analysis} goes through, except that our final bound will be
\begin{equation} \label{eq:main_bound_median}
    \left(1+O(\frac{1}{C}+\eps+\gamma)\right) \cdot \max_{r \ge 1} \min\left(\rho\left(\frac{p_1}{r}\right), \rho(p_1) \cdot \left(1+\frac{1}{4r \cdot \left(\frac{p_0 \cdot r}{p_1}-1\right)}\right)\right),
\end{equation}
    where $p_1 = 0.068$ and $p_0 = 0.337$. The main replacement here is that we replaced $\frac{r}{2p_1} = \frac{0.5 \cdot r}{p_1}$ with $\frac{0.337 \cdot r}{p_1}.$ We can use this to obtain a $2.408$-approximation, improving over $1+\sqrt{2}$. We will not elaborate on this, however, as we will see that using the method in Subsection \ref{subsec:k_means_improved}, we can further improve this to $\kmedianratio$.


\medskip

    We split the clients this time into $3$ groups. We let $\mathcal{D}_1$ be the set of clients $j \not\in \mathcal{D}_B$ corresponding to all subcases in Cases 1 and 2, $\mathcal{D}_2$ be the set of clients $j \not\in \mathcal{D}_B$ corresponding to all subcases in Case 3, and $\mathcal{D}_3$ be the set of clients corresponding to all subcases in Case 4, and all bad clients $j \in \mathcal{D}_B.$ For any client $j$, as in Subsection \ref{subsec:k_means_improved}, we define $A_j := \alpha_j - \sum_{i \in \bar{N}(j) \cap I_1} (\alpha_j-c(j, i))$ and $B_j := \sum_{i \in \bar{N}(j) \cap (I_2 \cup I_3)} (\alpha_j-c(j, i))$. We also define $Q_1, Q_2, Q_3, R_1, R_2, R_3$ similar to how we did for the $k$-means case.
    
    Similar to Lemma \ref{lem:RHS_positive_1/2_improved} in the $k$-means case, we have the following result for the $k$-median case. 
    
\begin{lemma}
    Let $\delta_1 = \sqrt{2}$, $\delta_2 = 1.395$, and $\delta_3 = 2-\sqrt{2}$. For any client $j \in \mathcal{D}_1$, $A_j \ge B_j$. For any client $j \in \mathcal{D}_2$, $A_j \ge \frac{1}{2} B_j$. Finally, for any client $j \in \mathcal{D}_3,$ $A_j \ge \frac{\delta_2-1}{2(2-\sqrt{2})} \cdot B_j.$
\end{lemma}

\begin{proof}
    Recall that $a = |\bar{N}(j) \cap I_1|$, $b = |\bar{N}(j) \cap I_2|$, and $c = |\bar{N}(j) \cap I_3|$. In case $1$ or $2$, we have that $a = 0$ and $b+c \le 1$, so $\alpha_j - \sum_{i \in \bar{N}(j) \cap I_1} = \alpha_j$, and the sum of $\sum_{i \in \bar{N}(j) \cap (I_1 \cup I_2)} (\alpha_j - c(j, i))$ is merely over a single point, so is at most $\alpha_j$. Thus, if $j \in \mathcal{D}_1,$ $A_j \ge B_j$. (Note that this even holds for bad clients $j \in \mathcal{D}_B$.)
    
    In case $3$, we have that $a = 0$, so $\alpha_j - \sum_{i \in \bar{N}(j) \cap I_1} = \alpha_j$. In addition, all points $i, i'$ in $\bar{N}(j) \cap I_2$ are separated by at least $\sqrt{2} \cdot \min(\tau_i, \tau_{i'}) \ge \sqrt{2} \cdot \alpha_j$. Hence, by Lemma \ref{lem:geometric_median}, $\sum_{i \in \bar{N}(j) \cap I_2} (\alpha_j - c(j, i)) \le \alpha_j \cdot (b-\sqrt{b(b-1)}) \le \alpha_j$. Likewise, $\sum_{i \in \bar{N}(j) \cap I_3} (\alpha_j - c(j, i)) \le \alpha_j \cdot (c-\sqrt{c(c-1)}) \le \alpha_j$. So,  $\alpha_j - \sum_{i \in \bar{N}(j) \cap I_1} = \alpha_j \ge \frac{1}{2} \sum_{i \in \bar{N}(j) \cap (I_2 \cup I_3)} (\alpha_j - c(j, i)).$ Thus, if $j \in \mathcal{D}_2,$ $A_j \ge B_j$.
    
    In case $4$, we have that $a \ge 1$. We claim that in this case, $\alpha_j - \sum_{i \in \bar{N}(j) \cap I_1} (\alpha_j - c(j, i)) \ge \frac{\delta_2-1}{2-\sqrt{2}} \cdot \sum_{i \in \bar{N}(j) \cap I_2} (\alpha_j - c(j, i))$. This will follow from the fact that all points in $\bar{N}(j) \cap I_1$ are separated by at least $\sqrt{2} \cdot \alpha_j$, all points in $\bar{N}(j) \cap I_2$ are also separated by at least $\sqrt{2} \cdot \alpha_j$, and all points in $\bar{N}(j) \cap I_1$ are separated by at least $\delta_2 \cdot \alpha_j$ from all points in $\bar{N}(j) \cap I_2$. In fact, this immediately follows from the bounding of the denominators in subcases \ref{eq:4.a'}, \ref{eq:4.b'}, \ref{eq:4.c'}, \ref{eq:4.d'}, and \ref{eq:4.e'}, where we replace $h$ with $b$. Likewise, $\alpha_j - \sum_{i \in \bar{N}(j) \cap I_1} (\alpha_j - c(j, i)) \ge \frac{\delta_2-1}{2-\sqrt{2}} \cdot \sum_{i \in \bar{N}(j) \cap I_3} (\alpha_j - c(j, i))$. Overall, we have that for all clients in case $4$, $A_j \ge \frac{\delta_2-1}{2(2-\sqrt{2})} \cdot B_j$. Since $A_j \ge \frac{\delta_2-1}{2(2-\sqrt{2})} \cdot B_j$ in all cases, it also holds for the bad clients $j \in \mathcal{D}_B$ as well.
\end{proof}

    As a direct corollary, we have that
\begin{equation} \label{eq:S_i_bound_median}
    R_1 \le Q_1, \hspace{1cm} R_2 \le 2Q_2, \hspace{0.5cm} \text{and} \hspace{0.5cm} R_3 \le \frac{2(2-\sqrt{2})}{\delta_2-1} Q_3.
\end{equation}

Next, similar to Lemma \ref{lem:more_bash} in the $k$-means case, we have the following result.

\begin{lemma} \label{lem:more_bash_median}
    Let $\delta_1 = \sqrt{2},$ $\delta_2 = 1.395$, and $\delta_3 = 2-\sqrt{2}$. Then, for all $p \in [0.01, 0.068],$ we have that $\rho^{(1)}(p) \le \max\left(1+\delta_2, \sqrt{3\left(X+Y\right)+2\sqrt{2\left(X+Y\right)^{2}-\delta_2^2 \cdot XY}}\right),$
    where $X = p^2+p(1-p) \cdot 1.1$ and $Y = (1-p)^2 + \frac{p(1-p)}{1.1}.$ 
    In addition, for all $p \in [0.01, 0.068]$, we have that $\rho^{(2)}(p) \le \frac{(1+\sqrt{2})-(3-\sqrt{2})\cdot(2p-2p^2)}{1-(2-\sqrt{2})\cdot(2p-2p^2)},$ and that $\rho^{(3)}(p) \le \frac{1}{\frac{1}{2}-2\cdot(2-d_2)\cdot p}$.
\end{lemma}

\begin{proof}
    To bound $\rho^{(1)}(p)$, we simply analyze all subcases in Case 1 and Case 2 and set $T = 1.1$. This is straightfoward to verify (see, for instance, our Desmos files on $k$-median in Appendix \ref{app:files}).
    
    To bound $\rho^{(2)}(p)$, we analyze all subcases in Case 3. Subcase \ref{eq:3.a'} and \ref{eq:3.b.i'} are straightfoward to verify. For subcase \ref{eq:3.b.ii'}, we have to verify for all choices of $c \ge 2$. For $c = 2$, we can verify manually. For $c \ge 3$, it is easy to see that the numerator of \ref{eq:3.b.ii'} is at most
\begin{align*}
    \frac{1}{2}\left[(1+\sqrt{2}) \cdot \left((1-2p)+(1-2p)^c\right) + 1-(1-2p)^c\right] &= \frac{1}{2}\left[1 + (1+\sqrt{2}) \cdot (1-2p) + \sqrt{2} \cdot (1-2p)^c\right] \\
    &\le \frac{1}{2}\left[1 + (1+\sqrt{2}) \cdot (1-2p) + \sqrt{2} \cdot (1-2p)^3\right],
\end{align*}
    and the denominator is at least $1-2p$. So, the fraction is at most
\[\frac{\frac{1}{2} \cdot \left[1 + (1+\sqrt{2}) \cdot (1-2p) + \sqrt{2} \cdot (1-2p)^3\right]}{1-2p},\]
    which is at most $\frac{(1+\sqrt{2})-(3-\sqrt{2})\cdot(2p-2p^{2})}{1-(2-\sqrt{2})\cdot(2p-2p^{2})}$ for all $p \in [0.01, 0.068]$.
    
    Finally, it is straightfoward to check that all $5$ subcases in Case $4$ are at most $\frac{1}{\frac{1}{2}-2\cdot(2-d_2)\cdot p}$ for all $p \in [0.1, 0.068]$. 
\end{proof}

One can modify the remainder of the proof analogously to as in the $k$-means case in Section \ref{subsec:k_means_improved}.
Hence, to show that we obtain an approximation $\rho + O(\eps + \gamma + 1/C)$, it suffices to show that for all choices of $\theta \in [0, 1]$ and $r \ge 1,$ that if we let $\mathfrak{D}' = \mathfrak{D}+O(\gamma) \cdot \text{OPT}_k$, one cannot simultaneously satisfy
\begin{align}
    \mathfrak{D}' &\ge \sum_{i = 1}^{3} \left(Q_i-\frac{p_1}{r} R_i\right) \label{eq:bound_3_kmedian}\\
    \rho \cdot \mathfrak{D}' &< \frac{\theta}{r} \sum_{i = 1}^{3} \rho^{(i)}(p_1) \cdot (Q_i-p_1 \cdot R_i) + \left(1-\frac{\theta}{r}\right) \cdot \rho(p_1) \cdot \left(\mathfrak{D}' + p_1 \cdot \frac{\theta}{r} \sum_{i = 1}^{4} R_i\right) \label{eq:bound_4_kmedian} \\
    \rho \cdot \mathfrak{D}' &< \sum_{i = 1}^{3} \rho^{(i)}\left(\frac{p_1}{r}\right) \cdot \left(Q_i-\frac{p_1}{r} \cdot R_i\right) \label{eq:bound_5_kmedian}
\end{align}
and
\begin{equation}
    R_1 \le Q_1, \hspace{1cm} R_2 \le 2 Q_2, \hspace{1cm} R_3 \le \frac{2(2-\sqrt{2})}{\delta_2-1} Q_3. \label{eq:bound_6_kmedian}
\end{equation}
    By numerical analysis of these linear constraints and based on the functions $\rho^{(i)}$, we obtain a $\boxed{\kmedianratio}$-approximation algorithm for Euclidean $k$-median clustering. We defer the details to Appendix \ref{app:k_means_numerical_analysis}.

\section*{Acknowledgments}

The authors thank Ashkan Norouzi-Fard for helpful discussions relating to modifying the previous results on roundable solutions.
The authors also thank Fabrizio Grandoni, Piotr Indyk, Euiwoong Lee, and Chris Schwiegelshohn for helpful conversations.
Finally, we would like to thank an anonymous reviewer for providing a useful suggestion in removing one of the cases for $k$-means.


\appendix

\section{Desmos Graphs and Code} \label{app:files}

Here, we provide links for the Desmos files used to visualize the LMP approximations for both $k$-means and $k$-median, and the Python code used to improve the approximation factor for $k$-means.

We provide graphs on Desmos for the LMP Approximation bounds for $k$-means and $k$-medians, as functions of the probability $p$. We remark that in some of these cases, there may be parameters (such as $a, b, c, c_1, c_2, h$) that need to be set properly (which can be done via toggles on the respective Desmos link) to see the actual approximation ratio as a function of $k$.

\medskip

Case $1$ of $k$-means is available here: 
\url{https://www.desmos.com/calculator/jd8ud6h2e9}

Case $2$ of $k$-means is available here: 
\url{https://www.desmos.com/calculator/pgtylk9eui}

Case $3$ of $k$-means is available here: 
\url{https://www.desmos.com/calculator/zjshynypsh}

Case $4$ of $k$-means is available here: 
\url{https://www.desmos.com/calculator/ibwult8qzs}

Case $5$ of $k$-means is available here: 
\url{https://www.desmos.com/calculator/pgtylk9eui}

\medskip

Case $1$ of $k$-median is available here: 
\url{https://www.desmos.com/calculator/9qmscsfvrr}

Case $2$ of $k$-median is available here: 
\url{https://www.desmos.com/calculator/rdidyxhs2o}

Case $3$ of $k$-median is available here: 
\url{https://www.desmos.com/calculator/zoeswetvyz}

Case $4$ of $k$-median is available here: 
\url{https://www.desmos.com/calculator/mpwrmz7mhe}


\medskip
The Python source code for $k$-means is available at
\newline
\url{https://drive.google.com/file/d/1mzKPr4ZbXe7FPDtx8JBkz2ReCK4OriQz/view?usp=sharing}, and the Python source code for $k$-median is available at
\newline
\url{https://drive.google.com/file/d/1SEfUvHaOd78QgxCFaFA8OjuDsKb3Qg0I/view?usp=sharing}. One can also view the code in a more readable PDF format for $k$-means at
\newline
\url{https://drive.google.com/file/d/1Ujcd6znbwxOkG-72zBZGX3otXPPAScud/view?usp=sharing},
and for $k$-median at
\newline
\url{https://drive.google.com/file/d/15HP3wBN20tCanwAc1dAA3rvjI23drdcO/view?usp=sharing}.

\section{Omitted Details for the LMP Approximations} \label{app:bash}

First, we prove Proposition \ref{prop:triangle}.

\begin{proof}[Proof of Proposition \ref{prop:triangle}]
    Let $v_1 = B-A$, $v_2 = C-B$, and $v_3 = D-B.$ Then, 
\begin{align*}
    p \cdot \|C-A\|_2^2 + (1-p) \cdot \|D-A\|_2^2 &= p \cdot \|v_1+v_2\|_2^2 + (1-p) \cdot \|v_1+v_3\|_2^2 \\
    &= \|v_1\|_2^2 + 2 \cdot \langle v_1, p v_2 + (1-p) v_3\rangle + p \cdot \|v_2\|_2^2  +(1-p) \cdot \|v_3\|_2^2 \\
    &\le 1 + 2 \cdot \|p v_2 + (1-p) v_3\|_2 + p \cdot \|v_2\|_2^2  +(1-p) \cdot \|v_3\|_2^2,
\end{align*}
    since $\|v_1\|_2 \le 1$.
Now, we can write 
\[\|p v_2 + (1-p) v_3\|_2 = \sqrt{\|p v_2 + (1-p) v_3\|_2^2} = \sqrt{p \cdot \|v_2\|_2^2 + (1-p) \cdot \|v_3\|_2^2 - p(1-p) \cdot \|v_2 -v_3\|_2^2}.\]
    So, we have that $p \cdot \|c-a\|_2^2 + (1-p) \cdot \|d-a\|_2^2$ is at most
\begin{multline} \label{eq:triangle}
    1 + 2 \cdot \sqrt{p \cdot \nu_1 \cdot \min(\sigma_1, \sigma_2) + (1-p) \cdot \nu_2 \cdot \min(\sigma_1, \sigma_3) - p(1-p) \cdot \nu_3 \cdot \min(\sigma_2, \sigma_3)} \\
    + p \cdot \nu_1 \cdot \min(\sigma_1, \sigma_2) + (1-p) \cdot \nu_2 \cdot \min(\sigma_1, \sigma_3).
\end{multline}
    It is simple to see that \eqref{eq:triangle} is nondecreasing in $\sigma_1$ for a fixed $\sigma_2, \sigma_3$, so \eqref{eq:triangle} is maximized when $\sigma_1 = 1$. Next, when $\sigma_1 = 1$, it is clear that \eqref{eq:triangle} is non-increasing in $\sigma_2$ if $\sigma_2 \ge 1$ and likewise for $\sigma_3$, so \eqref{eq:triangle} is maximized for some $\sigma_2, \sigma_3 \le 1$. In this case, \eqref{eq:triangle} simplifies to
\[1 + 2\sqrt{p \cdot \nu_1 \cdot \sigma_2 + (1-p) \cdot \nu_2 \cdot \sigma_3 - p(1-p) \cdot \nu_3 \cdot \min(\sigma_2, \sigma_3)} + p \cdot \nu_1 \cdot \sigma_2 + (1-p) \cdot \nu_2 \cdot \sigma_3.\]
    Now, using the fact that $\nu_1, \nu_2 \ge \nu_3$ and that $p, (1-p) \ge p(1-p)$, we have that this expression is nondecreasing in both $\sigma_2, \sigma_3$ as long as $\sigma_2, \sigma_3\le 1$. So, we may upper bound \eqref{eq:triangle}, and thus $p \cdot \|C-A\|_2^2 + (1-p) \cdot \|D-A\|_2^2$, by
\[1 + 2 \cdot \sqrt{p \cdot \nu_1 + (1-p) \cdot \nu_2 - p(1-p) \cdot \nu_3} + p \cdot \nu_1 + (1-p) \cdot \nu_2,\]
    by setting $\sigma_2 = \sigma_3 = 1.$
\end{proof}

Next, we complete the details in Lemmas \ref{lem:main_lmp} and \ref{lem:main_lmp_median} that we did not complete in the main body of the paper.

\paragraph{K-means: Case \ref{eq:1.c}:} We wish to maximize
\[\frac{(1-p) \cdot (t+\sqrt{\delta_1})^2 + p \cdot t^2}{1-p(1-t^2)} = \frac{(1-p) \cdot (t+\sqrt{\delta_1})^2 + p \cdot t^2}{(1-p) + p \cdot t^2},\]
over $0 \le t \le 1.$
First, note that if $t \le \sqrt{0.75}$, then we can bound this fraction by $(t+\sqrt{\delta_1})^2 \le (\sqrt{0.75}+\sqrt{\delta_1})^2$. Alternatively, if $t \ge \sqrt{0.75}$, then we can bound this fraction by at most
\[\frac{(1-p) \cdot (1+\sqrt{\delta_1})^2 + p \cdot t^2}{(1-p) + p \cdot t^2} \le \frac{(1-p) \cdot (1+\sqrt{\delta_1})^2 + 3p/4}{1-p/4},\]
where the left-hand side in the above equation has the numerator and denominator increasing at the same rate in terms of $t$, so it is maximized when $t$ is minimized, i.e., $t = \sqrt{0.75}$. Thus, we can bound the overall fraction as at most
\[\max\left((\sqrt{0.75}+\sqrt{\delta_1})^2, \frac{(1-p) \cdot (1+\sqrt{\delta_1})^2 + 3p/4}{1-p/4}\right).\]

\paragraph{K-Means: Case \ref{eq:1.g.i}:} We wish to maximize 
$$\frac{(1-p) \cdot (u+\sqrt{\delta_1 \cdot t})^2+p \cdot d(j, i_2)^2}{1-p+p \cdot d(j, i_2)^2}$$
over $t, u \in [0, 1]$ and $d(j, i_2) \ge u-\sqrt{\delta_3 \cdot t}$. First, note that if we treat $d(j, i_2)$ as a variable, the numerator and denominator increase at the same rate as $d(j, i_2)^2$ increases, so this fraction is maximized when $d(j, i_2) = \max(0, u-\sqrt{\delta_3 \cdot t})$. If $u-\sqrt{\delta_3 \cdot t} \le 0$, then this fraction equals $(u+\sqrt{\delta_1 \cdot t})^2$, but $u \le \sqrt{\delta_3 \cdot t} \le \sqrt{\delta_3}$ since $t \le 1,$ and this means that $(u+\sqrt{\delta_1 \cdot t})^2 \le (\sqrt{\delta_3} + \sqrt{\delta_1})^2$. Alternatively, we are maximizing
$$\frac{(1-p) \cdot (u+\sqrt{\delta_1 \cdot t})^2 + p \cdot (u-\sqrt{\delta_3 \cdot t})^2}{1-p + p \cdot (u-\sqrt{\delta_3 \cdot t})^2}$$
over $t, u \in [0, 1]$.
Next, note that if $u > \sqrt{\delta_3 \cdot t}$ and $t < 1,$ then increasing $t$ will decrease $(u-\sqrt{\delta_3 \cdot t})^2$ and increase $(u+\sqrt{\delta_1 \cdot t})^2$. So, the denominator decreases and the numerator either increases or decreases at a slower rate. Thus, we may assume that either $u-\sqrt{\delta_3 \cdot t} \le 0$ or that $t = 1$. In the case where $t = 1$, we wish to maximize
\[\frac{(1-p) \cdot (u+\sqrt{\delta_1})^2+p \cdot (u-\sqrt{\delta_3})^2}{1-p+p \cdot (u-\sqrt{\delta_3})^2}=\frac{p \cdot \left[(u+\sqrt{\delta_1})^2+(u-\sqrt{\delta_3})^2\right]+(1-2p) \cdot (u+\sqrt{\delta_1})^2}{p \cdot \left[(u-\sqrt{\delta_3})^2+1\right]+(1-2p)}.\]
Writing $A(u)=(u+\sqrt{\delta_1})^2+(u-\sqrt{\delta_3})^2$, $B(u) = (u+\sqrt{\delta_1})^2,$ and $C(u)=(u-\sqrt{\delta_3})^2+1,$ we can verify that $\frac{A(u)}{C(u)}$ and $B(u)$ are both increasing functions in $u$ over $[0, 1]$, which means so is $\frac{p \cdot A(u)+(1-2p) \cdot B(u)}{p \cdot C(u)+(1-2p)}.$ Therefore, the overall maximum is at most
\[\max\left((\sqrt{\delta_1}+\sqrt{\delta_3})^2, \frac{(1-p) \cdot (1+\sqrt{\delta_1})^2+p \cdot (1-\sqrt{\delta_3})^2}{1-p+p \cdot (1-\sqrt{\delta_3})^2}\right)\]

\paragraph{K-Means: Case \ref{eq:2.d}:} 
Our goal is to maximize
\[\frac{(1-2p) \cdot \min(1+\sqrt{\delta_1}, \max(\beta, \gamma)+\sqrt{\delta_1 \cdot t})^2 + p \cdot \beta^2 + p \cdot \gamma^2}{1-p (1-\beta^2) - p (1-\gamma^2)}\]
    over $t \ge 1$ and $\beta+\gamma \ge \sqrt{\delta_3 \cdot t}$ (and where $\beta, \gamma \ge 0$). By symmetry, we may assume WLOG that $\beta \ge \gamma$, and replace $\max(\beta, \gamma)$ with $\beta$. Next, note that increasing $t$ only increases the overall fraction, so we may increase $t$ until we have that $\beta+\gamma = \sqrt{\delta_3 \cdot t}$. So, we now wish to maximize
\[\frac{(1-2p) \cdot \min\left(1+\sqrt{\delta_1}, \beta+\sqrt{\delta_1/\delta_3} \cdot (\beta+\gamma)\right)^2+p \cdot (\beta^2+\gamma^2)}{1-2p+p \cdot (\beta^2+\gamma^2)}\]
    over $\beta, \gamma \ge 0$ subject to $\beta+\gamma \ge \sqrt{\delta_3}$ (since $\beta+\gamma \ge \sqrt{\delta_3 \cdot t}$ and $t \ge 1$). But, note that if $\beta+\sqrt{\delta_1/\delta_3}(\beta+\gamma) > 1+\sqrt{\delta_1}$, then any decrease in either $\beta$ or $\gamma$ until we have that  $\beta+\sqrt{\delta_1/\delta_3}(\beta+\gamma) = 1+\sqrt{\delta_1}$ will decrease both the numerator and the denominator by the same amount, and so will increase the fraction. Thus, we may assume that $\beta+\sqrt{\delta_1/\delta_3}(\beta+\gamma) \le 1+\sqrt{\delta_1}$.
    
    In this case, we may rewrite our goal as maximizing
\begin{equation} \label{eq:f_fraction}
    f(\beta, \gamma) := \frac{(1-2p) \cdot \left(\beta+\sqrt{\delta_1/\delta_3} \cdot (\beta+\gamma)\right)^2+p \cdot (\beta^2+\gamma^2)}{1-2p+p \cdot (\beta^2+\gamma^2)}
\end{equation}
    over $\beta, \gamma \ge 0$ subject to $\beta+\gamma \ge \sqrt{\delta_3}$ and $\beta+\sqrt{\delta_1/\delta_3}(\beta+\gamma) \le 1+\sqrt{\delta_1}$. Now, for any fixed $\beta, \gamma$, we note that $f(\lambda \beta, \lambda \gamma)$ for any $\lambda \ge 1$ multiplies the numerator of the fraction in \eqref{eq:f_fraction} by a $\lambda^2$ factor, but multiplies the denominator of the fraction by less than a $\lambda^2$ factor, since $1-2p \ge 0$. Therefore, the fraction increases overall, which means that to maximize $f(\beta, \gamma)$, we may always assume that $\beta+\sqrt{\delta_1/\delta_3}(\beta+\gamma)=1+\sqrt{\delta_1}$. It is easy to see that this automatically implies that $\beta+\gamma \ge \sqrt{\delta_3}$ when $\beta, \gamma \ge 0$.
    
    Thus, our goal is to maximize
\[\frac{(1-2p) \cdot (1+\sqrt{\delta_1})^2 + p(\beta^2+\gamma^2)}{(1-2p)+p\cdot(\beta^2+\gamma^2)}\]
    subject to $\beta, \gamma \ge 0$ and $\beta+\sqrt{\delta_1/\delta_3}(\beta+\gamma) = 1+\sqrt{\delta_1}$. Maximizing this, however, just entails to minimizing $\beta^2+\gamma^2$, which is easy to solve as $\beta=(\sqrt{\delta_3}+\sqrt{\delta_1}) \cdot \frac{\sqrt{\delta_3}(1+\sqrt{\delta_1})}{\delta_1+(\sqrt{\delta_1}+\sqrt{\delta_3})^2}$ and $\gamma=\sqrt{\delta_1} \cdot \frac{\sqrt{\delta_3}(1+\sqrt{\delta_1})}{\delta_1+(\sqrt{\delta_1}+\sqrt{\delta_3})^2},$ which means that $\beta^2+\gamma^2 = \frac{\delta_3 \cdot (1+\sqrt{\delta_1})^2}{\delta_1+(\sqrt{\delta_1}+\sqrt{\delta_3})^2}.$

\paragraph{K-median: Case \ref{eq:1.b'}:} It suffices to prove the following proposition.

\begin{proposition} \label{prop:triangle_2}
    Let $0 < p \le 1/2$, and suppose that we have $4$ points $i^*, i_1, i_3, j$ in Euclidean space such that $d(j, i^*) \le 1,$ $d(i^*, i_1) \le \sqrt{2} \cdot \min(t_{i^*}, t_{i_1}),$ $d(i^*, i_3) \le \sqrt{2} \cdot \min(t_{i^*}, t_{i_3}),$ $d(i_1, i_3) \ge \delta_2 \cdot \min(t_{i_1}, t_{i_3}),$ and $t_{i^*} \le 1$. Then, for any $T > 0$,
\[(1-p) \cdot d(j, i_1) + p \cdot d(j, i_3) \le \sqrt{3\left(X+Y\right)+2\sqrt{2\left(X+Y\right)^{2}-X \cdot \delta_2^2}}\]
where $X = p^2+p(1-p) \cdot T$ and $Y = (1-p)^2 + \frac{p(1-p)}{T}.$    
\end{proposition}

\begin{proof}
    We write $d(j, i_1) = \|j-i_1\|_2$ and $d(j, i_3) = \|j-i_3\|_2$. First, we have that
\begin{align*}
    &\hspace{0.5cm}\left((1-p) \cdot \|j-i_1\|_2 + p \cdot \|j-i_3\|_2\right)^2 \\
    &= (1-p)^2 \cdot \|j-i_1\|_2^2 + p^2 \cdot \|j-i_3\|_2^2 + 2p(1-p) \cdot \|j-i_1\|_2 \cdot \|j-i_3\|_2 \\
    &\le (1-p)^2 \cdot \|j-i_1\|_2^2 + p^2 \cdot \|j-i_3\|_2^2 + p(1-p) \cdot \left(\frac{1}{T} \cdot \|j-i_1\|_2^2 + T \cdot \|j-i_3\|_2^2\right)
\end{align*}
    for any $T > 0$. Writing $X = p^2+p(1-p) \cdot T$ and $Y = (1-p)^2 + \frac{p(1-p)}{T},$ we have that
\[\left((1-p) \cdot \|j-i_1\|_2 + p \cdot \|j-i_3\|_2\right)^2 \le X \cdot \|j-i_3\|_2^2 + Y \cdot \|j-i_1\|_2^2.\]
    We can now apply Proposition \ref{prop:triangle} on the points $A=j, B=i^*, C=i_3, D=i_1$, with $\nu_1=\nu_2=2, \nu_3=\delta_2^2$, and $\sigma_1=t_{i^*}, \sigma_2=t_{i_3},\sigma_3=t_{i_1}$ and where we replace the parameter $p$ in Proposition \ref{prop:triangle} with $\frac{X}{X+Y}$, to say that
\begin{align*}
    &\hspace{0.5cm} X \cdot \|j-i_3\|_2^2 + Y \cdot \|j-i_1\|_2^2 \\
    &\le (X+Y) \cdot \left(1 + \frac{X}{X+Y} \cdot 2 + \frac{Y}{X+Y} \cdot 2 + 2 \sqrt{\frac{X}{X+Y} \cdot 2 + \frac{Y}{X+Y} \cdot 2 - \frac{X}{X+Y} \cdot \frac{Y}{X+Y} \cdot \delta_2^2}\right) \\
    &= 3(X+Y) + 2\sqrt{2(X+Y)^2-\delta_2^2 \cdot XY}.
\end{align*}
    In summary, we have that for any choice of $T > 0$,
\begin{align*}
    \left((1-p) \cdot \|j-i_1\|_2 + p \cdot \|j-i_3\|_2\right)^2 &\le X \cdot \|j-i_3\|_2^2 + Y \cdot \|j-i_1\|_2^2 \\
    &\le 3(X+Y) + 2\sqrt{2(X+Y)^2-\delta_2^2 \cdot XY}. \qedhere
\end{align*}
\end{proof}

\paragraph{K-median: Case \ref{eq:1.g.i'}:} Our goal is to maximize
\[\frac{(1-p) \cdot (u+\delta_1 \cdot t)+p \cdot \max(0, u-t \cdot \delta_3)}{1-p+p \cdot \max(0, u-t \cdot \delta_3)}\]
over $0 \le t, u \le 1.$ First, note that if $u-t \cdot \delta_3 \le 0,$ then since $t \le 1,$ this means that $u \le \delta_3$. In this case, the fraction equals $u + \delta_1 \cdot t \le \delta_1+\delta_3 = 2$, since $\delta_1 = \sqrt{2}$ and $\delta_3 = 2-\sqrt{2}$. 

Alternatively, we have that $\max(0, u-t \cdot \delta_3) = u-t \cdot \delta_3.$ Let $u' = u-t \cdot \delta_3 \ge 0$, so $u = u' + t \cdot \delta_3$. In this case, we wish to maximize
\[\frac{(1-p) \cdot (u+\delta_1 \cdot t) + p \cdot (u-t \cdot \delta_3)}{1-p + p \cdot (u-t \cdot \delta_3)} = \frac{u' + (1-p) \cdot (\delta_1+\delta_3) t}{p \cdot u' + (1-p)} = \frac{u' + (1-p) \cdot 2t}{p \cdot u' + (1-p)}.\]
    over $0 \le t \le 1$ and $0 \le u' \le 1-t \cdot \delta_3$. Since $\frac{1}{p} \ge 2 \ge 2t$, we have that increasing $u'$ increases the fraction overall. So, we may assume that $u' = 1-t \cdot \delta_3.$ In this case, we are trying to maximize the fraction
\[\frac{1-t \cdot \delta_3 + (1-p) \cdot 2t}{p \cdot (1-t \cdot \delta_3) + (1-p)} =\frac{1 + (2-\delta_3) t - 2pt}{1 - p \cdot \delta_3 \cdot t}.\]
    Since $p \le \frac{1}{2}$, the numerator increases and the denominator decreases as $t$ increases, so the fraction increases overall. Thus, this fraction is maximized when $t = 1$, and equals
\[\frac{1+(2-\delta_3) - 2p}{1-p \cdot \delta_3} = \frac{1 + \delta_1 - (\delta_1+\delta_3)p}{1-p \cdot \delta_3}.\]

\section{Numerical Analysis for Euclidean $k$-means and $k$-median} \label{app:k_means_numerical_analysis}

\subsection{The $k$-means case}

We recall that our goal is to show, for an appropriate choice of $\rho$, that for any $0 \le \theta \le 1$ and any $r \ge 1$, we cannot simultaneously satisfy
\begin{align}
    \mathfrak{D}' &\ge \sum_{i = 1}^{5} \left(Q_i-\frac{p_1}{r} R_i\right), \label{eq:copy_1}\\
    \rho \cdot \mathfrak{D}' &< \frac{\theta}{r} \sum_{i = 1}^{5} \rho^{(i)}(p_1) \cdot (Q_i-p_1 \cdot R_i) + \left(1-\frac{\theta}{r}\right) \cdot \rho(p_1) \cdot \left(\mathfrak{D}' + p_1 \cdot \frac{\theta}{r} \sum_{i = 1}^{5} R_i\right), \label{eq:copy_2} \\
    \rho \cdot \mathfrak{D}' &< \sum_{i = 1}^{5} \rho^{(i)}\left(\frac{p_1}{r}\right) \cdot \left(Q_i-\frac{p_1}{r} \cdot R_i\right), \label{eq:copy_3}
\end{align}
and
\begin{equation}
    R_1 \le Q_1, \hspace{1cm} R_2 \le Q_2, \hspace{1cm} R_3 \le Q_3, \hspace{1cm} R_4 \le 1.75 Q_4, \hspace{1cm} R_5 \le 2 Q_5, \label{eq:copy_4}
\end{equation}
    where we will let $Q_1, Q_2, Q_3, Q_4, Q_5, R_1, R_2, R_3, R_4, R_5$ and $\mathfrak{D}'$ be arbitrary nonnegative reals. For $p_1 = 0.402$, we recall that $\rho(p_1) = 3+2\sqrt{2}$. 
    Now, note that if we increase $\mathfrak{D}'$, Equations \eqref{eq:copy_2} and \eqref{eq:copy_3} become harder to satisfy, since in both equations, the left hand side has a greater slope as a function of $\mathfrak{D}'$ than the right hand side. As a result, we may assume that $\mathfrak{D}' = \sum_{i=1}^{5} \left(Q_i-\frac{p_1}{r} R_i\right)$, which we know is nonnegative since $p_1 < 0.5$ and $r \ge 1$, so $Q_i \ge \frac{p_1}{r} \cdot R_i$ for all $1 \le i \le 5$.
    
    Now, we note that we may assume $r \ge 2.37$. This is because if $r \le 2.37,$ then $\frac{p_1}{r} \ge 0.169,$ and it is easy to verify that $\rho(p) \le 5.912$ for any $p \in [0.169, 0.402]$ (for instance, by using Lemma \ref{lem:more_bash} to bound $\rho^{(1)}(p), \rho^{(2)}(p)$, and $\rho^{(5)}(p)$, and using Cases \ref{eq:1.g.i} and \ref{eq:2.d} for $\rho^{(3)}(p)$ and $\rho^{(4)}(p)$). Therefore, for any $\rho \ge 5.912,$ if $r \le 2.37$ then Equations \eqref{eq:copy_1} and \eqref{eq:copy_3} cannot hold simultaneously. In addition, we may also assume that $r \le 4.18,$ since if $r \ge 3.5$, then we can use the simpler bound of $\rho(p_1) \cdot \left(1+\frac{1}{4 r \cdot (r/(2p_1)-1)}\right) \le 5.912$.
    
    We recall that for $p \in [0.096, 0.402]$, $\rho^{(1)}(p) = 3+2\sqrt{2}$, $\rho^{(2)}(p) = 1+2 \cdot p + (1-p)\cdot \delta_1+2\sqrt{2 \cdot p^2+(1-p)\cdot \delta_1},$ $\rho^{(3)}(p) = \frac{(1-p)\cdot(1+\sqrt{\delta_1})^2+p\cdot(1-\sqrt{\delta_3})^2}{1-p+p (1-\sqrt{\delta_3})^2}$, $\rho^{(4)}(p) = \frac{(1-2p)\cdot(1+\sqrt{\delta_1})^{2}\cdot(\delta_1+(\sqrt{\delta_1}+\sqrt{\delta_3})^{2})+p\cdot\left(1+\sqrt{\delta_1}\right)^2\cdot \delta_3}{(1-2p)\cdot(\delta_1+(\sqrt{\delta_1}+\sqrt{\delta_3})^{2})+p\cdot(1+\sqrt{\delta_1})^2\cdot \delta_3}$, and $\rho^{(5)}(p) = 5.68$.

    Now, let $\rho = \kmeansratio$, and suppose there exist $0 \le \theta_0 \le \theta \le \theta_1 \le 1$ and $1 \le r_0 \le r \le r_1$ such that Equations \eqref{eq:copy_1}, \eqref{eq:copy_2}, \eqref{eq:copy_3}, and \eqref{eq:copy_4} can be simultaneously satisfies for nonnegative $Q_1, Q_2, Q_3, Q_4, Q_5, R_1, R_2, R_3, R_4, R_5$. Then, in fact we must be able to satisfy the weaker conditions
\begin{align*}
    \mathfrak{D}' &= \sum_{i=1}^{5} \left(Q_i-\frac{p_1}{r_0} \cdot R_i\right) \\
    \rho \cdot \mathfrak{D}' &< \frac{\theta_1}{r_0} \cdot \sum_{i=1}^{5} \rho^{(i)}(p_1) \cdot (Q_i-p_1 \cdot R_1) + \left(1-\frac{\theta_0}{r_1}\right) \cdot \rho(p_1) \cdot \left(\mathfrak{D}' + \frac{\theta_1}{r_0} \cdot \sum_{i=1}^{4} R_i\right) \\
    \rho \cdot \mathfrak{D}' &< \sum_{i=1}^{5} \rho^{(i)} \left(\frac{p_1}{r_1}\right) \cdot \left(Q_i-\frac{p_1}{r_1} \cdot R_i\right),
\end{align*}
    and \eqref{eq:copy_4}, while having $Q_1, Q_2, Q_3, Q_4, Q_5,R_1, R_2, R_3, R_4, R_5$ all be nonnegative. Indeed, the conditions are weaker since we have decreased the value of $\mathcal{D}'$ and increased all terms on the right-hand side (noting that each $\rho^{(i)}$ is a non-increasing function in the range $[0, 0.402]$).
    
    For every $0 \le \theta_0 \le 0.99$ and $2.37 \le r_0 \le 4.17$ such that $\theta_0, r_0$ are integral multiples of $0.01$, we look at the intervals $\theta \in [\theta_0, \theta_0 + 0.01]$ and $r \in [r_0, r_0+0.01]$. If this region has a nonnegative solution to these inequalities, then we further partition the region $[\theta_0, \theta_0 + 0.01] \times [r_0, r_0+0.01]$ into a $10 \times 10$ grid of dimensions $0.001$. If one of these regions has a nonnegative solution to these inequalities, we partition one step further into a grid of dimensions $0.0001$ (we will not need to partition beyond this).
    Using this procecdure, we are able to obtain that there is \emph{no solution} in $\theta \in [0,1]$ and $r \in [2.37, 4.18]$ when $\rho = \kmeansratio$, which allows us to establish that our algorithm provides a polynomial-time $\boxed{\kmeansratio}$-approximation for Euclidean $k$-means clustering.
    
    See Appendix \ref{app:files} for links to the Python code.
    
\subsection{The $k$-median case}

The $k$-median case is almost identical, except for the modified equations and modified choices of $\rho^{(i)}$. This time, we wish to show that for any $\theta \in [0, 1]$ and $r \ge 1$, there exists $0 \le \theta_0 \le \theta \le \theta_1 \le 1$ and $1 \le r_0 \le r \le r_1$ such that one cannot satisfy
\begin{align*}
    \mathfrak{D}' &= \sum_{i=1}^{3} \left(Q_i-\frac{p_1}{r_0} \cdot R_i\right) \\
    \rho \cdot \mathfrak{D}' &< \frac{\theta_1}{r_0} \cdot \sum_{i=1}^{3} \rho^{(i)}(p_1) \cdot (Q_i-p_1 \cdot R_1) + \left(1-\frac{\theta_0}{r_1}\right) \cdot \rho(p_1) \cdot \left(\mathfrak{D}' + \frac{\theta_1}{r_0} \cdot \sum_{i=1}^{4} R_i\right) \\
    \rho \cdot \mathfrak{D}' &< \rho^{(1)} \left(\frac{p_1}{r_1}\right) \cdot \left(Q_i-\frac{p_1}{r_1} \cdot R_i\right)+\rho^{(2)} \left(\frac{p_1}{r_1}\right) \cdot \left(Q_i-\frac{p_1}{r_1} \cdot R_i\right)+\rho^{(3)} \left(\frac{p_1}{r_0}\right) \cdot \left(Q_i-\frac{p_1}{r_1} \cdot R_i\right),
\end{align*}
    where $Q_1, Q_2, Q_3, R_1, R_2, R_3$ are nonnegative.
    In the above equations, we set $p_1 = 0.068$, $\delta_1 = \sqrt{2}, \delta_2 = 1.395$, and $\delta_3 = 2-\sqrt{2}$. Also, recall that for $p \in [0.01, 0.068],$ we have that $\rho^{(1)}(p) \le \max\left(1+\delta_2, \sqrt{3\left(X+Y\right)+2\sqrt{2\left(X+Y\right)^{2}-\delta_2^2 \cdot XY}}\right),$
    where $X = p^2+p(1-p) \cdot 1.1$ and $Y = (1-p)^2 + \frac{p(1-p)}{1.1},$ $\rho^{(2)}(p) \le \frac{(1+\sqrt{2})-(3-\sqrt{2})\cdot(2p-2p^2)}{1-(2-\sqrt{2})\cdot(2p-2p^2)},$ and that $\rho^{(3)}(p) \le \frac{1}{\frac{1}{2}-2\cdot(2-d_2)\cdot p}$.
    
    We remark that in the final equation, we use $\rho^{(3)}(p_1/r_0)$ instead of $\rho^{(3)}(p_1/r_1)$ - this is because $\rho^{(3)}$ is a decreasing function on the region $[0, 0.068]$, as opposed to $\rho^{(1)}$ and $\rho^{(2)}$ which are both decreasing functions.
    
    First, we may assume that $r \in [2.4, 3.42].$ Indeed, if $1 \le r < 2.4$, one can use the more naive bound of $\rho(p_1/r)$, which is less than $2.406$. If $r > 3.42$, one can instead use the bound of 
\[\rho(p_1) \cdot \left(1+\frac{1}{4 r \cdot \left(\frac{p_0 \cdot r}{p_1}-1\right)}\right) \le 2.395 \cdot \left(1+\frac{1}{4 \cdot 3.42 \cdot \left(\frac{0.337 \cdot 3.42}{0.068}-1\right)}\right) < 2.406.\]
    To finish, we apply a similar method as in the $k$-means case. We split the region $(\theta, r) \in [0, 1] \times [2.4, 3.42]$ into grid blocks of size $0.005 \times 0.005$ with $\theta_0, \theta_1, r_0, r_1$ being the endpoints in each direction. We verify that the linear program has no solution when $\rho = 2.406$ for each grid block: if it does, we further refine the grid block into smaller $0.001 \times 0.001$-sized pieces and verify each of the smaller pieces.
    
    See Appendix \ref{app:files} for links to the Python code.
    
\section{Changes to Construction of Roundable Solutions} \label{app:why_am_i_doing_this}

In this section, we explain how Ahmadian et al.~\cite{ahmadian2017better} implicitly prove Theorem \ref{thm:ahmadian_roundable}, up to some minor modifications of their algorithm and analysis. Because the algorithm and analysis is almost entirely the same, we only describe the differences between the algorithm and analysis in \cite{ahmadian2017better} and what we need for our Theorem \ref{thm:ahmadian_roundable}.

The only changes in the overall algorithm will be as follows. We will set some small constant $\kappa$ such that $\eps = \kappa^2$. We will then set $K = \Theta(\eps^{-1} \gamma^{-4} \log \kappa)$ as opposed to $K = \Theta(\eps^{-1} \gamma^{-4})$ in \cite[Algorithm 2, Line 4]{ahmadian2017better}, and set the definition of \emph{stopped} \cite[Section 7]{ahmadian2017better} to be that $j \in \mathcal{D}$ is stopped if $\exists j' \neq j \in \mathcal{D}$ such that $(1+\kappa) \bar{\alpha}_j \ge d(j, j') + \kappa^{-1} \cdot \bar{\alpha}_{j'}$, as opposed to $\exists j' \neq j \in \mathcal{D}$ such that $2 \bar{\alpha}_j \ge d(j, j') + 6 \cdot \bar{\alpha}_{j'}$. Here, we are letting $\bar{\alpha}_j = \sqrt{\alpha_j}$ for $j \in \mathcal{D}$.

We now describe how this changes the claims throughout \cite[Sections 7-8]{ahmadian2017better}. We only describe the changes to the statements, because the proofs do not change at all. For nearly the remainder of this appendix, we will consider a new definition of roundable, neither the one in \cite{ahmadian2017better} nor our Definition \ref{def:roundable}. Our modified definition will instead have that: for all $j \in \mathcal{D} \backslash \mathcal{D}_B$ and all $A \in [2 \kappa, 1/(2 \kappa)]$, $(1+A + 10 \eps/\kappa)^2 \cdot \alpha_j \ge \left(d(j, w(j)) + A \cdot \sqrt{\tau_{w(j)}}\right)^2$, which replaces Conditions 3a and 3b in Definition \ref{def:roundable} (or Condition 2a in \cite[Definition 5.1]{ahmadian2017better}). In addition, our modified definition will have that for all clients $j \in \mathcal{D}$, $\kappa^{-2} \cdot \gamma \cdot \text{OPT}_{k'} \ge \sum_{j \in \mathcal{D}_B} (d(j, w(j)) + A \cdot \sqrt{\tau_{w(j)}})^2$ for all $A \in [2\kappa, 1/(2\kappa)]$, which replaces our Condition 3c in Definition \ref{def:roundable} (or condition 2b in \cite[Definition 5.1]{ahmadian2017better})

\medskip

We are now ready to describe how each claim in \cite{ahmadian2017better} changes (or stays the same).

\cite[Lemma 7.1]{ahmadian2017better} still holds with our new definition of \emph{stopped}: the same proof still works.

The (unnumbered) claim in the first paragraph of \cite[Section 8]{ahmadian2017better} still goes through, with our modified definition of \emph{roundable} (i.e., the one presented in this section).

In \cite[Section 8.1]{ahmadian2017better}, both \cite[Lemma 8.1]{ahmadian2017better} and \cite[Lemma 8.2]{ahmadian2017better} still hold with our new definition of stopped, with essentially no changes to the proof. Likewise, in \cite[Section 8.2]{ahmadian2017better}, Lemmas 8.3, 8.4, 8.5 and Corollary 8.6 in \cite{ahmadian2017better} still hold with our new definition of stopped.

In \cite[Section 8.3]{ahmadian2017better}, we change the definition of $\mathcal{B}$, the ``potentially bad'' clients (see \cite[Equation (8.1)]{ahmadian2017better}), to be $$\mathcal{B} = \left\{j \in \mathcal{D}: j \text{ is undecided and } (1+\kappa) \cdot \bar{\alpha}_j < d(j, j') + \frac{1}{\kappa} \cdot \bar{\alpha}_j^{(0)}\right\},$$
where $\bar{\alpha}_j^{(0)} = \sqrt{\alpha_j^{(0)}}$ refers to the value of $\alpha_j$ at the start of a \Call{RaisePrice} solution (i.e., when the solution $\mathcal{S}$ is labeled as $\mathcal{S}^{(0)}$ in our Algorithm \ref{alg:main}).
This contrasts to the original definition of $\mathcal{B}$, which was the set of undecided clients with $2 \bar{\alpha}_j < d(j, j') + 6 \cdot \bar{\alpha}_j^{(0)}$, in a similar way to how our definition of stopped contrasts with the original definition.

\cite[Lemma 8.7]{ahmadian2017better} is now as follows. For any $(\alpha, z)$ produced during \Call{RaisePrice}{}, for every client $j \in \mathcal{D}$ the following holds:
\begin{itemize}
    \item If $j \in \mathcal{D} \backslash \mathcal{B}$ then there exists a tight facility $i$ such that $(1+A+\eps/\kappa) \cdot \bar{\alpha}_j \ge d(j, i) \cdot A \cdot \sqrt{t_i}$ for all $A \in [2\kappa, 1/(2\kappa)]$.
    \item There exists a tight facility $i$ such that $\frac{1}{\kappa} \cdot \bar{\alpha}_j^{(0)} \ge d(j, i) + A \cdot \sqrt{t_i}$ for all $A \in [2\kappa, 1/(2\kappa)]$.
\end{itemize}
Again, the same proof holds.

We now move to \cite[Section 8.4]{ahmadian2017better}. We update \cite[Lemma 8.8]{ahmadian2017better} to be that if $j$ has a tight edge to some facility $i$, then $\alpha_{j'} \le \frac{5^2}{\kappa^4} \cdot \alpha_j$ for any $j'$ with a tight edge to $i$. In the proof, we would replace the stronger statement \cite[Equation (8.2)]{ahmadian2017better} with: $(1+\kappa) \cdot \bar{\alpha}_{j'} \le d(j', j) + \frac{4}{\kappa} \cdot \bar{\alpha}_j$. In addition, we would update the Claim inside the proof of \cite[Lemma 8.8]{ahmadian2017better} to be that: there is some tight facility $i^*$ in $(\alpha^{(0)}, z^{(0)})$ and also:
$$d(j_1, i^*) \le (1+\kappa) \bar{\alpha}_{j_1}^{(0)} \le (1+\kappa) \bar{\alpha}_j \hspace{0.5cm} \text{and} \hspace{0.5cm} \alpha_{j''}^{(0)} \le (1+\eps) \alpha_{j_1}^{(0)} \le (1+\eps) \alpha_j \text{ for all } j'' \in N^{(0)}(i^*).$$
Up to these changes, the rest of the proof of Lemma 8.8 in \cite{ahmadian2017better} is essentially unchanged.

Next, we update \cite[Lemma 8.9]{ahmadian2017better} to replace ``$\alpha_j^{(0)} \ge 20^2 \theta_s$ or $\alpha_j \ge 20^2 \theta_s$'' with ``$\alpha_j^{(0)} \ge \frac{6^2}{\kappa^4} \theta_s$ or $\alpha_j \ge \frac{6^2}{\kappa^4} \theta_s$'' - again the same proof still holds. 

\cite[Proposition 8.10]{ahmadian2017better} and its proof still hold, except that we have to replace $C_1 = \lceil \log_{1+\eps} (20^4) \rceil$ with $C_1 = \lceil \log_{1+\eps} 6^4/\kappa^8 \rceil = O(\eps^{-1} \log \kappa^{-1})$. So, if we set $\eps_z = n^{-6(K+C_1+2)-3}$ for our new choice of $C_1,$ Proposition 8.10 in \cite{ahmadian2017better} holds.

We update \cite[Proposition 8.11]{ahmadian2017better} to say that if $j \in \mathcal{B}$ for some $(\alpha, z)$ produced by \Call{RaisePrice}{}, then $\kappa^4 \cdot \theta_s \le \alpha_j^{(0)} \le \frac{6^4}{\kappa^{12}} \theta_s$. We replace the last equation in the Claim in the Proposition's proof with: $\kappa^2 \cdot \theta_s \le \alpha_{j'}^{(0)} \le \frac{6^4}{\kappa^8} \cdot \theta_s$. The same proof still holds. We also update the definition of $\mathcal{W}(\sigma)$ to be the set $\{j \in \mathcal{D}: \kappa^8/6^2 \cdot \theta_s \le \alpha_j^{(0)} \le 6^6/\kappa^{16} \cdot \theta_s \text{ for some s}\},$ where \Call{RaisePrice}{} defines the parameters $\theta_s$ based on the shift parameter $\sigma \in [0, K/2)$. With these definitions, and our modified choice of $K = \Theta(\eps^{-1} \gamma^{-4} \log \kappa^{-1})$, we will have that \cite[Corollary 8.12]{ahmadian2017better} still holds.

We now move to \cite[Section 8.5]{ahmadian2017better}. We keep their definitions of $\gamma$-\emph{close neighborhoods} and of \emph{dense} facilities and clients.
We also let the sets $\mathcal{F}_D, \mathcal{D}_D, \mathcal{F}_S^{(\ell)}, \mathcal{D}_S^{(\ell)}(i)$, and $\mathcal{D}_B$ be defined the same way (modulo our change in definition of $\mathcal{B}$). We also define $\tau_i$ the same way, and we will let $H^{(0)}$ and $IS^{(0)}$ simply represent the conflict graphs $H^{(0)}(\delta_1)$ and $IS_1^{(0)}$ as generated by our Algorithm \ref{alg:main}, respectively, at each iteration corresponding to making a new solution $\mathcal{S}^{(0)}$. Note that we are choosing $\delta = \delta_1 = \frac{4+8\sqrt{2}}{7}$, so $\sqrt{2} \le \sqrt{\delta} \le 2$.

With these, it is quite simple to see that \cite[Lemma 8.13]{ahmadian2017better} still holds, where the choice $\rho = \rho(0)$ in the proof is the approximation constant of the LMP algorithm in \cite{ahmadian2017better} with only a single parameter $\rho$ based on $\delta = \frac{4+8\sqrt{2}}{7}.$ In addition, \cite[Lemma 8.14]{ahmadian2017better} still holds, except that we replace $\text{OPT}_k$ with $\text{OPT}_{k'}$, since the final inequality in the proof relates $\sum_{j \in \mathcal{D}_{> \gamma}} d(j, IS^{(0)})^2 \le \sum_{j \in \mathcal{D}} d(j, IS^{(0)})^2$ to $\text{OPT}_{k'}$ now since $k'$ is the minimum of $k$ and all sizes of the sets that become $IS^{(0)}$ at some point. \cite[Corollary 8.15]{ahmadian2017better} (and the following Remark 8.16) also hold, due to our updated definition of $\mathcal{W}(\sigma)$ and $K$. 

Now, we update \cite[Lemma 8.17]{ahmadian2017better} to say: for any $j \in \mathcal{D}_D \cap \mathcal{B}$, either:
\begin{itemize}
    \item There exists a tight facility $i \in \mathcal{F}$ such that for all $A \in [2\kappa, 1/(2\kappa)]$, $(1+A + 10 \cdot \eps/\kappa) \bar{\alpha}_j \ge d(j, i) + A \cdot \sqrt{t_i}$
    \item There exists a special facility $i \in \mathcal{F}_S$ such that for all $A \in [2\kappa, 1/(2\kappa)]$, $(1+A + 10 \cdot \eps/\kappa) \cdot \bar{\alpha}_j \ge d(j, i) + A \cdot \sqrt{\tau_i}.$
\end{itemize}
    Again, the proof holds with minimal change.
    
We now move to \cite[Section 8.6]{ahmadian2017better}, the final section of Ahmadian et al.'s analysis. We first look at how \cite[Proposition 8.18]{ahmadian2017better} changes. The fact that $\alpha$ is feasible for $\text{DUAL}(\lambda+\frac{1}{n})$, that $\alpha_j \ge 1$ for all $j$, and that $z_i \in [\lambda, \lambda+\frac{1}{n}]$ are all still true. We now have that $(1+A+10 \eps/\kappa)^2 \alpha_j \ge (d(j, i)+A \sqrt{\tau_i})^2$ for all clients $j$ not in $\mathcal{B} \backslash \mathcal{D}_D \subset \mathcal{D}_B$ by using our modified versions of Lemma 8.7 and Lemma 8.17. In addition, our modified Lemma 8.7 tells us that even for bad clients $j \in \mathcal{D}_B$, there exists a tight facility $i$ such that $\kappa^{-2} \cdot \alpha_j^{(0)} \ge (d(j, i) + A \cdot \sqrt{t_i})^2$ for all $A \in [2\kappa, 1/(2\kappa)]$. Hence, we precisely have that $\kappa^{-2} \cdot \gamma \cdot \text{OPT}_{k'} \ge \sum_{j \in \mathcal{D}_B} (d(j, w(j)) + A \cdot \sqrt{\tau_{w(j)}})^2$ for all $A \in [2\kappa, 1/(2\kappa)]$, by adding over all clients $j \in \mathcal{D}_B$ and using Corollary 8.15, which still holds unchanged, apart from replacing $\text{OPT}_k$ with $\text{OPT}_{k'}$. The final part of proving \cite[Proposition 8.18]{ahmadian2017better}, i.e., verifying Condition 4 in our definition \ref{def:roundable}, holds where the only change is replacing $\text{OPT}_k$ with $\text{OPT}_{k'}$. Thus, we have that each solution that is generated is $(\lambda, k')$-roundable.

Finally, we have that \cite[Theorem 8.19]{ahmadian2017better} still holds with essentially no change, meaning that each call to \Call{RaisePrice}{} takes polynomial time and generates a polynomial number of $(\lambda, k')$-roundable solutions for our modified definition of roundable. This also implies that Algorithm \ref{alg:main} runs in polynomial time, since \Call{GraphUpdate}{} clearly takes polynomial time, and since the total number of times we call \Call{RaisePrice}{} is at most $|\mathcal{F}| \cdot L,$ which is polynomial since $\eps_z^{-1}$ and $m = |\mathcal{F}|$ are both polynomial in $n$.

Now, we note that each time we update our quasi-independent set in \Call{GraphUpdate}{}, the new set $(I_1^{(\ell, r)}, I_2^{(\ell, r)}, I_3^{(\ell, r)})$ only depends on $I_1^{(\ell, r-1)}$ and has no dependence on our choice of $I_2^{(\ell, r-1)}$ or $I_3^{(\ell, r-1)}$. Therefore, if we ignore the sets $I_2, I_3$ and only focus on $I_1$, the procedure of generating the sequence $\{I_1^{(\ell, r)}\}$ is in fact identical to the procedure in Ahmadian et al.\cite{ahmadian2017better}. The only difference is that we choose our stopping point based on the first time that $|I_1^{(\ell, r)}| + p_1 \cdot |I_2^{(\ell, r)} \cup I_3^{(\ell, r)}| < k,$ as opposed to the first time that $|I_1^{(\ell, r)}| \le k$ as done in \cite{ahmadian2017better}.
Because of this, our Algorithm \ref{alg:main} in fact works exactly as the main algorithm in \cite{ahmadian2017better} if we only focus on $I_1^{(\ell, r)}$ and set $\delta = \delta_1 = \frac{4+8\sqrt{2}}{7}.$ The only differences are the way we choose when to stop the procedure and the way we update $K$ and the definition of stopped clients and definition of the set $\mathcal{B}$.

As a result, we have that each solution $(\alpha, z, \mathcal{F}_S, \mathcal{D}_S)$ generated is $\lambda$ is $(\lambda, k)$-roundable, up to our modified definition of roundable (where $\lambda$ is chosen accordingly). By this, we mean that we replace Condition 3 in \ref{def:roundable} with:
\begin{enumerate}[label=\alph*)]
    \item For all $j \in \mathcal{D} \backslash \mathcal{D}_B$ and all $A \in [2\kappa, 1/(2\kappa)],$ $(1+A+10 \eps/\kappa)^2 \cdot \alpha_j \ge (d(j, w(j)) + A \cdot \sqrt{\tau_{w(j)}})^2$.
    \item For all $j \in \mathcal{D}$ and all $A \in [2\kappa, 1/(2 \kappa)]$, $\kappa^{-2} \cdot \gamma \cdot \text{OPT}_{k'} \ge \sum_{j \in \mathcal{D}_B} (d(j, w(j) + A \cdot \sqrt{\tau_{w(j)}})^2$.
\end{enumerate}
Recall that $\eps = \kappa^2 \ll 1$.
Now, by setting $A = 2\kappa$, we have that 
\[(1+12 \kappa)^2 \cdot \alpha_j = \left(1+2 \kappa + \frac{10 \eps}{\kappa}\right)^2 \cdot \alpha_j \ge \left(d(j, w(j)) + 2 \kappa \cdot \sqrt{\tau_{w(j)}}\right)^2 \ge c(j, w(j)),\]
and by setting $A = 1/(2\kappa)$, we have
\[\left(\frac{1+12 \kappa}{2\kappa}\right)^2 \cdot \alpha_j = \left(1 + \frac{1}{2\kappa} + 10 \cdot \frac{\eps}{\kappa}\right)^2 \cdot \alpha_j \ge \left(d(j, w(j)) + \frac{1}{2 \kappa} \cdot \sqrt{\tau_{w(j)}}\right)^2 \ge \frac{1}{(2\kappa)^2} \cdot \tau_{w(j)}.\]
Finally, by setting $A = 1$, we have
\[\kappa^{-2} \cdot \gamma \cdot \text{OPT}_{k'} \ge \sum_{j \in \mathcal{D}_B} \left(d(j, w(j)) + \sqrt{\tau_{w(j)}}\right)^2 \ge \sum_{j \in \mathcal{D}_B} \left(c(j, w(j)) + \tau_{w(j)}\right).\]

Therefore, by setting $\eps' = (1+12 \kappa)^2 - 1 = O(\sqrt{\eps})$, we have that $(1+\eps') \cdot \alpha_j \ge c(j, w(j))$ and $(1+\eps') \cdot \alpha_j \ge \tau_{w(j)}$. In addition, if we set $\gamma' = \kappa^{-2} \cdot \gamma$, we have that $\gamma' \cdot \text{OPT}_k \ge \sum_{j \in \mathcal{D}_B} (c(j, w(j)) + \tau_{w(j)}).$ Therefore, since $\eps = \kappa^2$, and if we assume that $\gamma \ll \eps^2$, then we have that $\eps' \ll 1$ and $\gamma' \ll \eps \ll \eps'$. So, we can replace $\eps'$ with $\eps$ and $\gamma'$ with $\gamma$, we have an algorithm that still runs in polynomial time (since the old values of $\kappa, \gamma, \eps$ are still polynomial factors in the new values of $\eps, \gamma$ which are all constants, even if they are arbitrary small). But now, we have that each solution $(\alpha^{(\ell)}, z^{(\ell)}, \mathcal{F}_S^{(\ell)}, \mathcal{D}_S^{(\ell)})$ satisfies the actual Condition 3 in Definition \ref{def:roundable}, for our new values of $\eps$ and $\gamma$.

Overall, we have that the algorithm runs in polynomial time, and each solution is $k'$-roundable, where $k'$ is the minimum of $k$ and $\min |I_1^{(0)}|$ over the course of the algorithm. Each pair of consecutive solutions is close as in Theorem 8.19 in \cite{ahmadian2017better} (which follows from their Proposition 8.10). Next, we have that each time we create a solution $\mathcal{S}^{(0)}$, Lemma 8.1 in \cite{ahmadian2017better}, which holds in our setting, tells us that every client $(\alpha^{(0)}, z^{(0)})$ is decided. Since $\mathcal{B}$ is a subset of undecided facilities, $\mathcal{B} = \emptyset$ for a solution $\mathcal{S}^{(0)}$, which means that $\mathcal{F}_S = \emptyset$ based on the definition of $\mathcal{F}_S$. In addition, our modified version of \cite[Lemma 8.7]{ahmadian2017better} holds for all $j$ since $\mathcal{B} = \emptyset,$ which means that we can set the bad clients $\mathcal{D}_B$ to be $\emptyset$. So, for each $\mathcal{S}^{(0)}$, the special facilities and bad clients are both empty. Next, we have that $I^{(\ell, r)}$ is a nested quasi-independent set because of how we defined $\mathcal{V}^{(\ell, r)}$ and how we defined $I^{(\ell, r)}$ in our \Call{GraphUpdate}{} procedure. Finally, we had that $|\mathcal{V}^{(\ell, r)} \backslash \mathcal{V}^{(\ell, r+1)}| \le 1$ as described at the end of Subsection \ref{subsec:k_means_alg_prelim}, and that we created $I_1^{(\ell, r+1)}$ from $I_1^{(\ell, r)}$ by removing a single point from $\mathcal{V}^{(\ell, r)}$ if $|\mathcal{V}^{(\ell, r)} \backslash \mathcal{V}^{(\ell, r+1)}| = 1$ (which may or may not be in $I_1^{(\ell, r)}$), and then extending to a maximal independent set of $\mathcal{V}^{(\ell+1, r)}$. So, we have that $|I_1^{(\ell, r)} \backslash I_1^{(\ell, r+1)}| \le 1.$ This means that all of the statements of Theorem \ref{thm:ahmadian_roundable} hold.

\section{Limit of \cite{ahmadian2017better} in Obtaining Improved Approximations} \label{app:limit}

In this section, we show that the algorithm of Ahmadian et al.~\cite{ahmadian2017better} cannot guarantee an LMP approximation better than $1+\sqrt{2}$ in the case of $k$-median.
In more detail, we show that there exists a set of clients $\clients$, facilities $\facilities$, and parameter $\lambda > 0$ such that for any choice $\delta \ge 1$ in the pruning phase, the LMP algorithm described in the preliminary subsection \ref{subsec:k_means_alg_prelim} does not obtain better than a $(1+\sqrt{2})$-approximation for $k$-median.
As a result, their technique cannot guarantee an LMP approximation for all choices $\lambda$, which means any improvement to their analysis would have to move significantly outside the LMP framework.


\medskip

We start with the $k$-median case. First, consider the points $j$, $i_1=j_1$, and $i_2=j_2$ such that $j, j_1, j_2$ are collinear in that order, $d(j, j_1) = T$, and $d(j_1, j_2) = \sqrt{2} \cdot T$ for some choice of $T > 0$. 
Consider applying the LMP algorithm described in Section \ref{subsec:witness_and_conflict} on just these points $\facilities = \{i_1, i_2\}$ and $\clients = \{j,j, \dots, j, j_1, j_2\}$, where we set $\lambda = T$ and will include a large number $N$ of copies of $j$.
In this case, the growing phase will set $\alpha_j = \alpha_{j_1} = \alpha_{j_2} = T$, where $i_1$ and $i_2$ both become tight. Also, $N(i_1) = \{j_1\}$ (with each copy of $j$ barely not being in it) and $N(i_2) = \{j_2\}$.
One also obtains that $t_{i_1} = t_{i_2} = T$. Then, if $\delta \ge \sqrt{2}$, then
$i_1, i_2$ are connected in the conflict graph $H(\delta),$ which means that the pruning phase will only allow either $i_1$ or $i_2$ to be in our set $S$.
The algorithm is arbitrary, and may set $S = \{i_2\}$. In this case, the total clustering cost is $N \cdot T \cdot (1+\sqrt{2}) + \sqrt{2} \cdot T = (1+\sqrt{2}) \cdot T \cdot (N+1) - T,$ whereas the dual is $\sum \alpha_j - \lambda \cdot 1 = T \cdot (N+2) - T = T \cdot (N+1)$.
If $\delta < \sqrt{2},$ then both $i_1, i_2$ are included, so the primal is $T \cdot N$ and the dual is $\sum \alpha_j - \lambda \cdot 2 = T \cdot N$.

Next, we consider a point $j$ as well as points $i_1, \dots, i_h$, such that $i_1, \dots, i_h$ form a regular simplex with centroid $j$ and pairwise distances $T' \cdot \sqrt{2} \cdot (1-\eps)$ between each $i_r$ and $i_s$, for some $T' > 0$ and arbitrarily small $\eps$.
Consider applying the LMP algorithm described in Section \ref{subsec:witness_and_conflict} on just these points $\facilities = \{i_1, \dots, i_h\}$ and $\clients = \{j\}$, where we set $\lambda = T' \cdot \left(1-(1-\eps)\sqrt{\frac{h-1}{h}}\right)$.
In this case, we will have that since $d(\alpha_j, i_r) = T' \cdot (1-\eps) \cdot \sqrt{\frac{h-1}{h}}$ for all $1 \le r \le h$, all facilities $i_r$ will become tight with $\alpha_j = t_{i_r} = T'$ for all $1 \le r \le h$. 
If $\delta < \sqrt{2}$, since the pairwise distances are more than $T' \cdot \delta$, the conflict graph will be empty so all facilities will be in the independent set. Therefore, the clustering cost will be $T' \cdot (1-\eps) \cdot \sqrt{\frac{h-1}{h}}$, and the dual will be 
\[\alpha_j - \lambda \cdot h = T' \left(1 - h\left(1-(1-\eps)\sqrt{\frac{h-1}{h}}\right)\right) = T' \cdot \left((1-\eps) \sqrt{h(h-1)} - (h-1)\right) \le T' \cdot \left(1-\eps \cdot h\right).\]
Else, if $\delta \ge \sqrt{2}$, the conflict graph $H(\delta)$ is complete on $i_1, \dots, i_h$, so only one facility will be in the independent set. The clustering cost is still $T' \cdot (1-\eps) \cdot \sqrt{\frac{h-1}{h}}$, and the dual will be 
\[\alpha_j - \lambda \cdot 1 = T' \left(1 - \left(1-(1-\eps)\sqrt{\frac{h-1}{h}}\right)\right) = T' \cdot (1-\eps) \cdot \sqrt{\frac{h-1}{h}}.\]

Now, we fix $\eps$ as a very small constant, and $h = \Theta(\eps^{-3})$. Finally, we set $T = 1 = \lambda$ and $T' = 1/\left(1-(1-\eps)\sqrt{\frac{h-1}{h}}\right) = \Theta(\eps^{-1})$. Finally, we set $N = \Theta(\eps^{-2})$, and consider the concatenation of each of the two cases described above, where the corresponding cases are sufficiently far apart in Euclidean space that there is no interaction.

If $\delta \ge \sqrt{2}$, then the overall clustering cost is 
\[(1+\sqrt{2}) \cdot T \cdot (N+1)-T + T' \cdot (1-\eps) \cdot \sqrt{\frac{h-1}{h}} = (1+\sqrt{2}) \cdot N \cdot (1 + O(\eps))\]
whereas the total dual is
\[T \cdot N + T' \cdot (1-\eps) \cdot \sqrt{\frac{h-1}{h}} = N \cdot (1+O(\eps)).\]
So, we do not obtain better than a $1+\sqrt{2}-O(\eps)$ approximation in this case.
If $\delta < \sqrt{2}$, then the total dual is in fact negative, as it is at most
\[T \cdot N + T' \cdot (1-\eps \cdot h) \le -\Omega(\eps^{-3}).\]
Overall, there is no choice of $\delta$ that we can set to improve over a $1+\sqrt{2}$ approximation.

\end{document}